\documentclass[12pt, reqno, letterpaper]{amsart}

\usepackage{graphicx}    
\usepackage{bbm, xcolor}
\usepackage{xfrac}
\usepackage{multicol}
\usepackage{soul}

\newcommand{\add}[1]{\textcolor{blue}{#1}}

\usepackage[margin=1in]{geometry}

\usepackage[english]{babel}

\usepackage{amsmath,amsthm,amsfonts,amssymb}

\newtheorem{theorem}{Theorem}[section]
\newtheorem{corollary}{Corollary}[section]

\newtheorem{proposition}{Proposition}[section]
\newtheorem{assumption}{Assumption}[section]
\newtheorem{definition}{Definition}[section]
\newtheorem{remark}{Remark}[section]
\numberwithin{figure}{section}

\begin{document}

\title[Short-maturity asymptotics for VIX and European options]
{Short-maturity asymptotics for VIX and European options in local-stochastic volatility models}

\author{Dan Pirjol}
\address{Stevens Institute of Technology, Hoboken, New Jersey, United States of America, dpirjol@gmail.com}

\author{Xiaoyu Wang}
\address{Hong Kong University of Science and Technology (Guangzhou), People's Republic of China, xiaoyuwang421@gmail.com}

\author{Lingjiong Zhu}
\address{Florida State University, Tallahassee, Florida, United States of America, zhu@math.fsu.edu}

\date{July 2024}


\begin{abstract}
We derive the short-maturity asymptotics for European and VIX option prices in local-stochastic volatility models where the volatility follows a continuous-path Markov process. 
Both out-of-the-money (OTM) and at-the-money (ATM) asymptotics are considered. 
Using large deviations theory methods, the asymptotics for the OTM options are expressed as 
a two-dimensional variational problem, which is reduced to an extremal problem for a function of two real variables. This extremal problem is solved explicitly in an expansion in log-moneyness. We derive series expansions for the implied volatility for European and VIX options which should be useful for model calibration.
We give explicit results for two classes of local-stochastic volatility models relevant in practice, with Heston-type and SABR-type stochastic volatility.
The leading-order asymptotics for at-the-money options are computed in closed-form.
The asymptotic results reproduce known results in the literature for the Heston and SABR models and for the uncorrelated local-stochastic volatility model. 
The asymptotic results are tested against numerical simulations for a local-stochastic volatility model with bounded local volatility.
\end{abstract}

\maketitle

\baselineskip18pt

\section{Introduction}

The CBOE Volatility Index (VIX) is the main volatility benchmark of the U.S stock market, and provides a measure of the implied volatility of options with maturity of 30 days on the S\&P 500 index. It is defined in terms of an expectation in the risk-neutral measure $\mathrm{VIX}_t^2 = - \frac{2}{\tau} \mathbb{E}[\log(S_{t+\tau}/S_t)|\mathcal{F}_t]$, where $S_t$ is the equity index S\&P 500 at time $t$, and $\tau = 30$ days. The expectation is computed by replication in terms of market observed SPX option prices - see the VIX White Paper \cite{VIXwp} for the details of the methodology.
Since 2022, CBOE has started reporting also 
the CBOE 1-day Volatility Index \texttt{(VIX1D)} \cite{VIX1Dwp}, which is an analog of the VIX index computed using the PM-settled weekly SPX options which mature on the same day and the next day $(\tau=1$ day) as the index date.

The volatility index VIX is used by market participants to speculate on and hedge volatility risk. Several volatility derivatives which can be used for this purpose are traded on CBOE Options Exchange: futures contracts on VIX are traded since 2004, and 
VIX options are traded since 2006. In view of the popularity of these contracts, a great deal of work has been devoted in the literature to the valuation of volatility derivatives. 

The simplest approach for pricing volatility options is based on modeling the instantaneous variance $V_t$ as a stochastic process. Detemple and Osakwe (2000) \cite{Detemple2000} presented both European and American volatility options pricing under several popular diffusion models for $V_t$. Carr et al. (2005) \cite{Carr2005} presented results for volatility options under pure jump models with independent increments.
Sepp (2008) \cite{Sepp2008a,Sepp2008} priced volatility derivatives under a square root volatility model with jumps. 
Goard and Mazur (2013) \cite{Goard2013} derived
analytical results for VIX options in the 3/2 stochastic volatility model, and 
Baldeaux and Badran (2014) \cite{Baldeaux2014} extended this model for VIX option pricing by adding jumps. 
A survey of existing results on volatility derivatives (up to 2010) was given by Carr and Lee (2010) \cite{Carr2010}.


Recently, the pricing of volatility derivatives has been extended to stochastic volatility models where the volatility is driven by a fractional Brownian motion. Horvath et al. \cite{Horvath2020} introduced the class of modulated Volterra processes which can accommodate observed VIX smiles. Jacquier et al. (2021) \cite{Jacquier2021}
derive short-maturity SPX and VIX option prices for a wide class of multi-factor models of this type. 
An empirical analysis of the SPX and VIX option markets under rough and stochastic volatility models was given 
by R\o{}mer (2022) \cite{Romer2022}.

We mention also the martingale optimal transport approach which was applied to the problem of simultaneous calibration to the SPX and VIX implied volatility smiles in \cite{Guyon2020}. This approach is model-independent and aims to calibrate the joint distribution of the underlying (SPX) and of the VIX at several maturities of interest, under appropriate martingale constraints. 

In this paper the asset price $S_t$ is assumed to follow a local-stochastic volatility model under the
risk-neutral probability measure $\mathbb{Q}$:
\begin{eqnarray}\label{LSvol}
&& \frac{dS_t}{S_t} = \eta(S_t) \sqrt{V_t} dW_t + (r - q) dt \,, \\
&& \frac{dV_t}{V_t} =\sigma(V_t) dZ_t + \mu(V_t) dt \,,\nonumber
\end{eqnarray}
with initial conditions $S_0 > 0, V_0 > 0$, 
where $W_t, Z_t$ are correlated standard Brownian motions with correlation $\rho$, $r$ is the risk-free rate and $q$ is the dividend yield.
For simplicity we assume that the functions $\eta(\cdot), \sigma(\cdot):\mathbb{R}^{+}\rightarrow\mathbb{R}^{+}$ and $\mu(\cdot):\mathbb{R}^{+}\rightarrow\mathbb{R}$ are time-homogeneous.


The model (\ref{LSvol}) is a continuous path Markovian local-stochastic volatility model. 
It nests several popular models in the literature. When $\eta(x) \equiv 1$ it reduces to the usual stochastic volatility models:
for example 
Heston model \cite{Heston1993} ($\sigma(v) = \sigma v^{-\frac12}, \mu(v) = \mu/v - \theta$), Hull-White model \cite{HullWhite1987}
($\sigma(v) \equiv \sigma, \mu(v) \equiv \mu$). When $\eta(x) = x^{\beta - 1}$ and $\mu(v)\equiv 0$ it reduces to the 
SABR model \cite{SABR}.

The short maturity asymptotics of European option prices in local-stochastic models have been studied by Forde and Jacquier
\cite{Forde2011} (in the uncorrelated limit) using large deviations methods. Local-stochastic volatility models were studied by Pagliarani and Pascucci \cite{Pagliarani2013} and
Lorig, Pagliarani and Pascucci  \cite{Lorig2017}, using PDE methods. These methods extend similar small 
maturity expansions which were obtained for stochastic volatility models in \cite{SABR,Forde2009,Forde2012}.
Bompis and Gobet (2018) \cite{Bompis2018} used Malliavin calculus methods to derive short-maturity asymptotics for the implied volatility of European options in local-stochastic volatility models with Heston-type volatility.

Fewer results are available in the literature on the short-maturity asymptotics of the VIX options in local-stochastic volatility models. We mention the work of Forde and Smith \cite{Forde2023}, where the asymptotics of VIX options is obtained in an uncorrelated local-stochastic model. As an application, they obtain the first two terms in the expansion of the VIX smile around the ATM point in the uncorrelated CEV-Heston model. 
However, their model includes a $S_t$-dependent drift for the variance process, and is different from our model (\ref{LSvol}).
To our knowledge, the short-maturity of VIX options in the correlated local-stochastic volatility model has not been treated previously in the literature.

An alternative to local-stochastic volatility models which allows independent control of the European and VIX smiles are stochastic volatility models with local correlation. 
The short-maturity asymptotics of VIX options in such models was obtained by Forde and Smith \cite{Forde2023} in a Markovian setting.


The paper is organized as follows.
In Section~\ref{sec:technical} we fully specify the model under appropriate technical conditions, and give the 
definition of the VIX volatility index and of options on this index.

Section~\ref{sec:SPX} presents the short-maturity asymptotics of European (SPX) options under the model (\ref{LSvol}). 
The main result is Theorem~\ref{Thm:E}, which uses large deviations theory, see \cite{Dembo1998,VaradhanLD} for background, 
to establish the short-maturity asymptotics in terms of a rate function $J_E$,  which is given by the solution of a 
two-dimensional variational problem. 
After a careful application of the Cauchy-Schwarz inequality to obtain a lower bound for the variational problem and showing the 
lower bound can be achieved, we reduce the variational problem further to that of finding the extrema of a function of two real variables, which is feasible for practical applications. 
The function depends on an auxiliary function $H(y,z)$ which depends only on the volatility process, 
and is represented as the solution of a one-dimensional variational problem.

In Section~\ref{sec:VIX}, we give the short-maturity asymptotics for VIX options under the model (\ref{LSvol}). 
The short-maturity asymptotics is given in terms of a rate function $J_V$,  which is given again by an extremal problem of a function of two variables. The extremal problem depends on the same auxiliary function $H(y,z)$ as in the European options case.

Section~\ref{sec:H} studies the properties of the function $H(y,z)$. We give explicit solutions for $H(y,z)$ for two particular forms of the variance process $V_t$ which are often encountered in practice: (i) $\sigma(v)\equiv\sigma$ corresponding to log-normal volatility (SABR-type models), and (ii) $\sigma(v)=\sigma v^{-1/2}$ corresponding to a square-root volatility specification (Heston-type models).

In Section~\ref{sec:applications}, we present a few applications of the theoretical results obtained in the paper, and give explicit results for the asymptotic implied volatility of European and VIX options in 
local-stochastic volatility model with log-normal (SABR-type) and Heston-type volatility. 
We check explicitly that our results recover existing results in the literature in various limiting cases: uncorrelated local stochastic volatility, pure stochastic volatility models, and local volatility model. 

Finally, in Section~\ref{sec:numerical} we compare the theoretical predictions for the asymptotic short-maturity of European and VIX options in local-stochastic volatility models with a numerical simulation of this model using Monte Carlo methods. For this test we use the Tanh-model for the local volatility function $\eta(x)$, which was introduced previously in \cite{Forde2011}. We observe good agreement between the asymptotic results and the numerical simulation for sufficiently small option maturity.

We present a few basic concepts about large deviations theory in Appendix~\ref{app:LD}. 
The proofs of the results presented in the main text
are presented in Appendix~\ref{sec:proofs:main:results}, Appendix~\ref{proofs:sec:H} and Appendix~\ref{proofs:sec:applications}.
The full result for the ATM VIX implied volatility convexity in the local-stochastic volatility model with SABR-type volatility
will be presented in Appendix~\ref{app:VIXcvx}.

\section{Model Specification}\label{sec:technical}

We start by formulating technical conditions and assumptions for the parameters of the model (\ref{LSvol}).
First, we assume that $\eta(\cdot),\mu(\cdot)$ and $\sigma(\cdot)$ are
uniformly bounded.

\begin{assumption}\label{assump:bounded}
We assume that $\eta(\cdot),\mu(\cdot)$ and $\sigma(\cdot)$ are
uniformly bounded:
\begin{equation}
\sup_{x\in\mathbb{R}}\eta(x)\leq M_{\eta},
\qquad
\sup_{x\in\mathbb{R}}|\mu(x)|\leq M_{\mu},
\qquad
\sup_{x\in\mathbb{R}}\sigma(x)\leq M_{\sigma}.
\end{equation}
\end{assumption}

We also assume that $\eta(\cdot)$ is decreasing, 
which satisfies the leverage effect in finance.
More precisely, when $\eta(\cdot)$ is not a constant function,
we assume that $\eta(\cdot)$ is strictly decreasing
so that its inverse function $\eta^{-1}(\cdot)$ exists. 
We also provide the following assumptions on Lipschitz continuity.

\begin{assumption}
\label{assump:lip}
We assume that $\eta$ is $L$-Lipschitz and $\sigma$ is $L'$-Lipschitz.
\end{assumption}

In addition, we impose the following
assumption on the $\eta(\cdot)$ and $\sigma(\cdot)$
that appear in the diffusion terms of \eqref{LSvol}
that is needed for the small-time large deviations estimates for \eqref{LSvol}.

\begin{assumption}\label{assump:LDP}
We assume that $\inf_{x\in\mathbb{R}}\sigma(x)>0$
and $\inf_{x\in\mathbb{R}}\eta(x)>0$. Moreover, 
there exist some constants $M,\alpha>0$ such that
for any $x,y\in\mathbb{R}$,
$|\sigma(e^{x})-\sigma(e^{y})|\leq M|x-y|^{\alpha}$
and $|\eta(e^{x})-\eta(e^{y})|\leq M|x-y|^{\alpha}$.
\end{assumption}

Next, we will show that under the Assumption~\ref{assump:bounded}, 
all the moments of $V_{t}$ process are finite.

\begin{proposition}\label{prop:V:moments}
Under Assumption~\ref{assump:bounded}, for any $p\geq 1$, there exists some $C_{p}\in(0,\infty)$, such that $\max_{0\leq t\leq T}\mathbb{E}[(V_t)^p]\leq C_{p}$ for any sufficiently small $T>0$.
\end{proposition}

Throughout the paper, we assume that the discounted asset price $S_{t}/e^{(r-q)t}$ is a martingale.
We quote the results of \cite{Lions2008} expressed in the notations of our paper, for the stochastic volatility model obtained by taking $\eta(x)\equiv 1$ and $r=q=0$ in Eqn.~\eqref{LSvol}.
\footnote{Some typos in \cite{Lions2008} were corrected in the paper \cite{Carr2019}.}

\begin{proposition}\label{prop:LM}
Consider the stochastic volatility model
\begin{equation}\label{SVmodel}
dS_t  =  S_t \sqrt{V_t} dW_t \,,\qquad
\frac{dV_t}{V_t} = \mu(V_t) dt + \sigma(V_t) dZ_t\,,
\end{equation}
where $W_t,Z_t$ are correlated standard Brownian motions with correlation $\rho$.
The asset price $S_t$ is a martingale if the following condition is satisfied
\begin{equation}\label{lim1}
\lim_{x\to \infty}\left\{ \rho \sigma(x^2) x + \mu(x^2) - \frac14 \sigma(x^2)\right\} < \infty\,.
\end{equation}
\end{proposition}

We have the following corollary.

\begin{corollary}
Assume $\mu(v)$ is bounded and $\lim_{x\to \infty} \sigma(x) = \sigma_\infty$ is finite. Then the limit (\ref{lim1}) is $+\infty$ for positive correlation $\rho>0$, 
takes a finite value if $\rho=0$ and is $-\infty$ for $\rho<0$. 
Thus, $S_t$ is a martingale provided that $\rho\leq 0$.
\end{corollary}

Finally, we assume the $p$-th
moment of $S_{T}$ is finite for some $p>1$.

\begin{assumption}\label{assump:S:T:p}
There exists some $p>1$, such that there exists some $C'_{p}\in(0,\infty)$, such that $\mathbb{E}[S_{T}^p]\leq C'_{p}$ for any sufficiently small $T>0$.
\end{assumption}

\begin{remark}\label{remark:S:T:p}
Assumption~\ref{assump:S:T:p} is a mild assumption.
As an illustration, we give a condition for the stochastic volatility model (\ref{SVmodel}) such that Assumption~\ref{assump:S:T:p} holds.
This condition follows directly from Lions and Musiela (2008) \cite{Lions2008}
which states that if the following limits exist
\begin{equation}
\lim_{x\to \infty}  \sigma(x^2) =: \sigma_\infty\,,\quad
\lim_{x\to \infty} \left\{ \frac{\mu(x^2)}{x} - \frac14 \frac{\sigma(x^2)}{x} \right\}=: b_\infty,
\end{equation}
then for any
\begin{equation}\label{2}
\rho < - \sqrt{\frac{p-1}{p}} - \frac{b_\infty}{p\sigma_\infty}\,,\quad p > 1 \,,
\end{equation}
there exists some $C'_{p}\in(0,\infty)$, such that $\max_{0\leq t\leq T}\mathbb{E}[(S_t)^p]\leq C'_{p}$ for any sufficiently small $T>0$. In particular, under Assumption~\ref{assump:bounded}, $b_{\infty}=0$, 
and the condition \eqref{2} reduces to 
\begin{equation}\label{2:simplified}
\rho < - \sqrt{\frac{p-1}{p}}\,,\quad p > 1 \,,
\end{equation}
which provides a sufficient condition for Assumption~\ref{assump:S:T:p} to hold for the stochastic volatility model (\ref{SVmodel}).
\end{remark}


\subsection{VIX futures and VIX options}

The CBOE Volatility Index (VIX) is a measure of the S\&P500 expected volatility, which is published by the 
Chicago Board Options Exchange. This index is defined by the risk-neutral expectation
\begin{equation}
\mathrm{VIX}_T^2 = \mathbb{E}\left[-\frac{2}{\tau} \log\left(\frac{S_{T+\tau}}{S_T}\right) \Big|\mathcal{F}_T \right],
\end{equation}
with $\tau=30$ days. This expectation is estimated from the prices of current (as of $T$)
call and put options on the SPX index, see \cite{VIXwp}.

CBOE lists futures and options on the $\mathrm{VIX}_T$ index with several maturities $T>0$.
VIX option contracts pay at time $T$ an amount linked to the $\mathrm{VIX}_T$ observed at the same time.

Under a model of type
\begin{equation}
\frac{dS_t}{S_t} = \sqrt{V_t} dW_t + (r-q) dt\,, 
\end{equation}
with $\{V_t\}_{t\geq 0}$ a non-negative stochastic process with continuous paths (no jumps) which may depend on $S_t$,
the $\mathrm{VIX}_T$ index is given by the risk-neutral expectation
\begin{equation}
\mathrm{VIX}_T^2 = \mathbb{E}\left[\frac{1}{\tau} \int_T^{T+\tau} V_t dt \Big|\mathcal{F}_T\right].
\end{equation}
More generally, under model \eqref{LSvol}, the $\mathrm{VIX}_T$ index is given by:
\begin{equation}\label{VIX:formula}
\mathrm{VIX}_T^2 = \mathbb{E}\left[\frac{1}{\tau} \int_T^{T+\tau} \eta^{2}(S_{t})V_t dt \Big|\mathcal{F}_T\right].
\end{equation}

The price of a futures contract on the $\mathrm{VIX}$ index with maturity $T$ is given by the risk-neutral expectation
\begin{equation}
F_V(T) = \mathbb{E}[\mathrm{VIX}_T] \,.
\end{equation}

The prices of VIX calls and puts are given by risk-neutral expectations
\begin{align}
&C_{V}(K,T)= e^{-rT } \mathbb{E}[(\mathrm{VIX}_{T}-K)^{+}]\,, \\
&P_{V}(K,T)= e^{-rT} \mathbb{E}[(K-\mathrm{VIX}_{T})^{+}]\,.
\end{align}

We impose the following definition to distinguish the VIX options into three cases.

\begin{definition}
VIX options with maturity $T$ are at-the-money (ATM) if $K = F_V(T)$. VIX call options are in-the-money (ITM) 
if $K<F_V(T)$ and out-of-money (OTM) if $K > F_V(T)$. Analogously, VIX put options are ITM
if $K>F_V(T)$ and OTM if $K < F_V(T)$.
\end{definition}

\section{Short-maturity asymptotics of European options}\label{sec:SPX}

In this section, we consider the European options in the model (\ref{LSvol}), 
with $C_{E}(K,T):=e^{-rT}\mathbb{E}[(S_{T}-K)^{+}]$
denoting the price of call option
and $P_{E}(K,T):=e^{-rT}\mathbb{E}[(K-S_{T})^{+}]$
denoting the price of call option.
We have the following short-maturity asymptotics for European options. 

\begin{theorem}\label{Thm:E}
Suppose Assumptions~\ref{assump:bounded},  \ref{assump:LDP} and \ref{assump:S:T:p} hold. The short maturity asymptotics of OTM European options in the local-stochastic volatility model (\ref{LSvol}) are as follows.

(i) The short-maturity asymptotics of OTM European call options is
\begin{equation}
\lim_{T\to 0} T \log C_E(K,T) = - J_E(K,S_0,V_0;\rho) \,,\quad K > S_0\,,
\end{equation}
with
\begin{equation}\label{J:E:eqn}
J_E(K,S_0,V_0;\rho) = \inf_{y,z}
\left\{\frac{1}{2(1-\rho^{2})z}\left(\int_{S_{0}}^{K}\frac{dx}{x\eta(x)}-\int_{V_{0}}^{e^{y}}\frac{\rho dx}{\sqrt{x}\sigma(x)}\right)^{2}
+H(y,z)\right\},
\end{equation}
where
\begin{equation}\label{Hdef1}
H(y,z):=\inf_{\substack{h(0)=\log V_{0},h(1)=y\\
\int_{0}^{1}e^{h(t)}dt=z}}\frac{1}{2}\int_{0}^{1}\left(\frac{h'(t)}{\sigma(e^{h(t)})}\right)^{2}dt.
\end{equation}

(ii) The short-maturity asymptotics of OTM European put options is 
\begin{equation}
\lim_{T\to 0} T \log P_E(K,T) = - J_E(K,S_0,V_0;\rho) \,,
\end{equation}
where $J_{E}$ is defined in \eqref{J:E:eqn} with $K<S_{0}$.
\end{theorem}

This result simplifies in the uncorrelated case, and is expressed as the solution of a one-dimensional extremal problem.
We have
\begin{equation}
J_E(K,S_0,V_0;0) =  \inf_{z}
\left\{\frac{1}{2z}\left(\int_{S_{0}}^{K}\frac{dx}{x\eta(x)}\right)^{2}
+H(z)\right\},
\end{equation}
where
\begin{equation}
H(z):=\inf_{h(0)=\log V_{0},\int_{0}^{1}e^{h(t)}dt=z}\frac{1}{2}\int_{0}^{1}\left(\frac{h'(t)}{\sigma(e^{h(t)})}\right)^{2}dt.
\end{equation}

The function $H(z)$ coincides with the rate function for Asian options in local volatility models with local volatility 
$\sigma(\cdot )$. The solution of the variational problem for this function was given in \cite{PZAsian}, and explicit 
solutions were given for $\sigma(v)\equiv\sigma$ in \cite{PZAsian} and for $\sigma(v) = \sigma v^\beta$ in \cite{PZAsianCEV}.

Next, let us present the asymptotics for ATM European options. 

\begin{theorem}\label{thm:European:ATM}
Suppose Assumptions~\ref{assump:bounded} and \ref{assump:lip} hold.
We also assume that there exists some $C'\in(0,\infty)$ 
such that $\max_{0\leq t\leq T}\mathbb{E}[(S_{t})^{4}]\leq C'$
for any sufficiently small $T>0$.
The short-maturity asymptotics of ATM European options are given by:
\begin{equation}
\lim_{T\rightarrow 0}\frac{1}{\sqrt{T}}C_{E}(S_0,T)
=\lim_{T\rightarrow 0}\frac{1}{\sqrt{T}}P_{E}(S_0,T)
=\frac{\eta(S_{0})\sqrt{V_{0}}}{\sqrt{2\pi}}.
\end{equation}
\end{theorem}

Theorem~\ref{thm:European:ATM} shows that the prices of ATM European options
are of the order $\sqrt{T}$ as $T\rightarrow 0$, 
and it provides the exact formula for the leading-order term.

\begin{remark}\label{remark:finite:moments}
In Theorem~\ref{thm:European:ATM}, we assumed the finiteness
of $\max_{0\leq t\leq T}\mathbb{E}[(S_{t})^{4}]$.
This is a mild condition. For the stochastic volatility model (\ref{SVmodel}), this holds when
$\rho<-\sqrt{3}/2$ which can be seen from Remark~\ref{remark:S:T:p}
by taking $p=4$ in \eqref{2:simplified}.
\end{remark}

\section{Short-maturity asymptotics of VIX options}\label{sec:VIX}

As $T\to 0$, the VIX futures prices in the model (\ref{LSvol}) $F_V(T)$ approach
to $F_V(0) = \sqrt{V_0} \eta(S_0)$. 
Therefore, in the short-maturity limit, we will refer to VIX options as OTM/ITM by referencing to $F_V(0)$.

For sufficiently small $\tau$, VIX options are essentially European options on combinations of $\sqrt{V_T}$ and $S_T$.
The following result makes this statement more precise. 

\begin{proposition}\label{prop:VIXsmalltau}
If Assumption~\ref{assump:bounded} and~\ref{assump:lip} hold, and 
\begin{equation}\label{eta2p.assumption}
\sup_{s\geq 0}|(\eta^{2})''(s)s^{2}|\leq M_{\eta,2} \,.
\end{equation}
Then we have
\begin{align}
\left|\mathrm{VIX}_T^2
-V_{T}\eta^{2}(S_{T})\right|\leq C_{1}(\tau)S_{T}+C_{2}(\tau)V_{T},
\end{align}
and moreover,
\begin{equation}
\mathbb{E}\left|\mathrm{VIX}_T^2-V_{T}\eta^{2}(S_{T})\right|
\leq
C_{1}(\tau)S_{0}e^{(r-q)T}
+C_{2}(\tau)V_{0}e^{T(M_{\mu}+M_{\sigma}^{2})},
\end{equation}
where
\begin{align}
&C_{1}(\tau):=2LM_{\eta}|r-q|e^{|r-q|\tau}\tau,\label{C:1:tau}
\\
&C_{2}(\tau):=M_{\eta}^{2}\left(e^{2\tau M_{\mu}}
e^{4\tau M_{\sigma}^{2}}+1-2e^{-\tau M_{\mu}-\frac{1}{2}\tau M_{\sigma}^{2}}\right)^{1/2}
+\frac{\tau}{2}M_{\eta,2}M_{\eta}^{2} e^{\tau(M_{\mu}+M_{\sigma}^{2})}.\label{C:2:tau}
\end{align}
\end{proposition}

Note that $C_{1}(\tau)$ is of order $\tau$ and $C_{2}(\tau)$ is of order $\tau^{1/2}$ as $\tau\rightarrow 0$.
Proposition~\ref{prop:VIXsmalltau}
implies that $\left| \mathrm{VIX}_T^2 - V_T \eta^2(S_T) \right|$ is of the order $O(\tau^{1/2})$
in terms of expectation. 
Next, as a corollary of Proposition~\ref{prop:VIXsmalltau}, we will 
provide an upper bound for $\left| \mathrm{VIX}_T - \sqrt{V_T} \eta(S_T) \right|$ 
and show that it is also of the order $O(\tau^{1/2})$
in terms of expectation.

\begin{corollary}\label{cor:VIXsmalltau}
Suppose the same assumptions in Proposition~\ref{prop:VIXsmalltau} hold,
and further assume that $m_{\eta}:=\inf_{s\geq 0}\eta(s)>0$ and $\mathbb{E}[S_{T}^{2}]=O(1)$ as $T\rightarrow 0$. 
Then, we have
\begin{align}
\left|\mathrm{VIX}_T
-\sqrt{V_{T}}\eta(S_{T})\right|\leq\frac{C_{1}(\tau)}{m_{\eta}}\frac{S_{T}}{\sqrt{V_{T}}}+\frac{C_{2}(\tau)}{m_{\eta}}\sqrt{V_{T}},
\end{align}
where $C_{1}(\tau),C_{2}(\tau)$ are given in \eqref{C:1:tau}-\eqref{C:2:tau}
and moreover
\begin{equation}
\label{tau:order}
\mathbb{E}\left|\mathrm{VIX}_T-\sqrt{V_{T}}\eta(S_{T})\right|=O(\tau^{1/2}),
\qquad\text{as $\tau\rightarrow 0$}.
\end{equation}
\end{corollary}

Note that in Corollary~\ref{cor:VIXsmalltau}, a sufficient condition for the additional assumption 
$\mathbb{E}[S_{T}^{2}]=O(1)$ as $T\rightarrow 0$ to hold
is when $\rho<-1/\sqrt{2}$ (see the discussions in Remark~\ref{remark:S:T:p}). 

\subsection{Particular cases}

For a few particular cases of the drift function 
$\mu( \cdot )$ for the $V_t$ process we have closed form expressions for $\mathrm{VIX}_T$ in terms of $V_T, S_T$.
Recall that the stochastic volatility models are obtained in the limit $\eta(\cdot ) \equiv 1$:
\begin{eqnarray}\label{Svol}
&& \frac{dS_t}{S_t} = \sqrt{V_t} dW_t + (r - q) dt \,, \\
&& \frac{dV_t}{V_t} =\sigma(V_t) dZ_t + \mu(V_t) dt\,, \nonumber
\end{eqnarray}
with initial conditions $S_0 > 0, V_0 > 0$, 
where $W_t, Z_t$ are correlated standard Brownian motions with correlation $\rho$, $r$ is the risk-free rate
and $q$ is the dividend yield.

\begin{proposition}\label{prop:StochVolVIX}

(i) Assume that the asset price follows the stochastic volatility model (\ref{Svol}) and
$\mu(\cdot)\equiv\mu$ is constant. Then the price of a VIX call option is expressed as
\begin{equation}
C_{V}(K,T)
=\mathbb{E}\left[\left(\sqrt{V_{T}}\sqrt{\frac{e^{\mu\tau}-1}{\tau\mu}}-K\right)^{+}\right]\,.
\end{equation}

(ii) Assume that the asset price follows the stochastic volatility model (\ref{Svol}) and
that the drift term of the $V_t$ process is mean-reverting 
$\mu(V_{t})V_{t}=a(b-V_{t})$, where $a,b>0$. This includes the Cox-Ingersoll-Ross process as a special case.
For this case we have
\begin{equation}
C_{V}(K,T)
=\mathbb{E}\left[\left(\sqrt{V_{T}\frac{1-e^{-a\tau}}{a\tau}+b\left(1-\frac{1-e^{-a\tau}}{a\tau}\right)}-K\right)^{+}\right]\,.
\end{equation}
\end{proposition}

\begin{remark}\label{rmk:VIX}
In general, since $V_{t}$ is a time-homogeneous Markov process, in stochastic volatility models (\ref{Svol}), we have
\begin{equation}
\mathrm{VIX}_T^2 = \mathcal{F}(V_T) \mbox{ with }
\mathcal{F}(x):=\frac{1}{\tau}\int_{0}^{\tau}\mathbb{E}[V_{s}|V_{0}=x]ds\,,
\end{equation}
and the VIX option prices are given by
\begin{align}
\label{CVtime-h}
C_{V}(K,T)
=\mathbb{E}\left[\left(\sqrt{\mathcal{F}(V_{T})}-K\right)^{+}\right],
\qquad
P_{V}(K,T)
=\mathbb{E}\left[\left(K - \sqrt{\mathcal{F}(V_{T})} \right)^{+}\right] \,.
\end{align}
\end{remark}

\begin{remark}
The result of Proposition~\ref{prop:StochVolVIX} can be extended to the more general 
local-stochastic volatility model (\ref{LSvol}) with $\eta(S) = \eta_0 \sqrt{S}$
for the particular case when $W_t$ and $Z_t$ are uncorrelated. 
\footnote{Note that $\eta(S)=\eta_0 \sqrt{S}$ does not satisfy Assumptions~\ref{assump:bounded}, \ref{assump:lip} and \ref{assump:LDP}. However, 
Assumptions~\ref{assump:bounded}, \ref{assump:lip} and \ref{assump:LDP} are used to obtain the asymptotic results in Section~\ref{sec:SPX} and Section~\ref{sec:VIX} as $T\rightarrow 0$, whereas here we obtain some explicit formula under this special case for any finite $T$, that does not rely on Assumptions~\ref{assump:bounded}, \ref{assump:lip} and \ref{assump:LDP} 
A similar comment holds for case 2) in Proposition~\ref{prop:StochVolVIX} where $\mu(v) = b(a/v-1)$ is not bounded.}

(i) Assuming $\mu(v)\equiv\mu$ we have
\begin{equation}
\mathrm{VIX}_T^2 = \eta_0^2 S_T V_T \frac{e^{(r-q+\mu)\tau}-1}{\tau(r-q+\mu)}\,.
\end{equation}
For this case VIX options become essentially options on the product $\sqrt{S_T V_T}$.
\begin{equation}
C_{V}(K,T)
=\mathbb{E}\left[\left(\eta_{0}\sqrt{S_{T}V_{T}}\sqrt{\frac{e^{(\mu+r-q)\tau}-1}{\tau(\mu+r-q)}}-K\right)^{+}\right],
\end{equation}
and the short-maturity asymptotics can be easily obtained as a European call option.

(ii) Assuming $\mu(v) v = a(b - v)$, we have 
\begin{align}
&\mathbb{E}\left[\frac{1}{\tau}\int_{T}^{T+\tau}V_{s}\eta^{2}(S_{s})ds\Big|\mathcal{F}_{T}\right]
\nonumber
\\
&=\frac{\eta_{0}^{2}}{\tau}\int_{T}^{T+\tau}S_{T}e^{(r-q)(s-T)}\left(V_{T}e^{-a(s-T)}+b\left(1-e^{-a(s-T)}\right)\right)ds
\nonumber
\\
&=\eta_{0}^{2}S_{T}V_{T}\frac{e^{(r-q-a)\tau}-1}{\tau(r-q-a)}
+\eta_{0}^{2}S_{T}b\left(\frac{e^{(r-q)\tau}-1}{\tau(r-q)}-\frac{e^{(r-q-a)\tau}-1}{\tau(r-q-a)}\right).
\end{align}
such that
\begin{equation}
C_{V}(K,T)
=\mathbb{E}\left[\left(\eta_{0}\sqrt{S_{T}V_{T}\frac{e^{(r-q-a)\tau}-1}{\tau(r-q-a)}
+S_{T}b\left(\frac{e^{(r-q)\tau}-1}{\tau(r-q)}-\frac{e^{(r-q-a)\tau}-1}{\tau(r-q-a)}\right)}-K\right)^{+}\right].
\end{equation}
If we let $\tau\rightarrow 0$, then the second term proportional to $S_T$ vanishes, and we obtain 
$C_{V}(K,T)\rightarrow\mathbb{E}[(\sqrt{S_{T}V_{T}}-K)^{+}]$, which is similar to the previous case. 
The short-maturity asymptotics can be again easily obtained as a European call option.
However, at finite $\tau$, it is a bit more complicated than European options,
and it will involve solving a slightly different variational problem.
\end{remark}

\subsection{The main result}

We first present the main result for OTM VIX options in the local-stochastic volatility model (\ref{LSvol}).

\begin{theorem}\label{Thm:VIX}
Under the settings of Corollary~\ref{cor:VIXsmalltau}, suppose Assumption~\ref{assump:LDP} holds and further assume $\tau=o(1)$ as $T\rightarrow 0$, then the short maturity asymptotics of OTM VIX options in the local-stochastic volatility model (\ref{LSvol}) are as follows.

(i) The asymptotics of the OTM VIX call option is 
\begin{equation}
\lim_{T\to 0} T \log C_V(K,T) = - J_V(K,S_0, V_0;\rho)\,,\quad K > \eta(S_0) \sqrt{V_0}\,,
\end{equation}
where 
\begin{align}
J_V(K,S_0,V_0;\rho) =
\inf_{y,z}
\left\{\frac{1}{2(1-\rho^{2})z}\left(\int_{S_{0}}^{(\eta^{2})^{-1}(K^{2}e^{-y})}\frac{dx}{x\eta(x)}-\int_{V_{0}}^{e^{y}}\frac{\rho dx}{\sqrt{x}\sigma(x)}\right)^{2}
+H(y,z)\right\},\label{J:V:eqn}
\end{align}
where
\begin{equation}\label{Hdef}
H(y,z):=\inf_{\substack{h(0)=\log V_{0},h(1)=y\\
\int_{0}^{1}e^{h(t)}dt=z}}\frac{1}{2}\int_{0}^{1}\left(\frac{h'(t)}{\sigma(e^{h(t)})}\right)^{2}dt.
\end{equation}

(ii) The asymptotics of OTM VIX put options is
\begin{equation}
\lim_{T\to 0} T \log P_V(K,T) = - J_V(K,S_0, V_0;\rho)\,,
\end{equation}
where $J_V$ is defined in \eqref{J:V:eqn} with $K < \eta(S_0) \sqrt{V_0}$.
\end{theorem}

The rate function for OTM VIX options depends on the auxiliary function $H(y,z)$. 
The same function appears in the short-maturity asymptotics of the 
European options, see (\ref{Hdef1}). We study the general properties of this function and give closed form evaluations for commonly used cases for $\sigma(v)$ in Section~\ref{sec:H}.

\textbf{Stochastic volatility models.}
The stochastic volatility model (\ref{Svol}) is obtained by taking $\eta(x)\equiv 1$ in (\ref{LSvol}). This case is not covered by directly taking
$\eta(x) \equiv 1$ into Theorem ~\ref{Thm:VIX}. VIX options in the stochastic volatility model are effectively European-type options on $V_T$. 

Restricting further to models with time-homogeneous volatility dynamics, 
we have $\mathrm{VIX}_T^2= F(V_T)$, see Remark~\ref{rmk:VIX}, 
and the VIX options are European-type options on $V_T$.
For these models, the VIX futures price is $F_V(T) = \mathbb{E}[\sqrt{\mathcal{F}(V_T)}] = \sqrt{\mathcal{F}(V_0)} + O(T)$ as $T\rightarrow 0$. 

\begin{proposition}\label{prop:VIX}
Consider the OTM VIX options in the stochastic volatility models (\ref{Svol}) with time-homogeneous volatility process.
In these models one has $\mathrm{VIX}_T^2 = \mathcal{F}(V_T)$, see Remark~\ref{rmk:VIX}.
Assuming that there exists some $p>1$ such that 
$\mathbb{E}[(\mathcal{F}(V_T))^{p/2}]=O(1)$ as $T\rightarrow 0$, the short maturity of these options
is given by 
\begin{equation}
\lim_{T\to 0} T \log C_V(K,T) = - J_V(K, V_0)\,,\quad
\lim_{T\to 0} T \log P_V(K,T) = - J_V(K, V_0)\,,
\end{equation}
where 
\begin{align}
& J_V(K,V_0) =
\frac{1}{2}\left(\int_{V_0}^{\mathcal{F}^{-1}(K^2)}
\frac{dx}{x\sigma(x)}\right)^{2}\,.
\end{align}
\end{proposition}

\textbf{Uncorrelated case.}
In the uncorrelated limit $\rho=0$ the variational problem for the rate function in Theorem~\ref{Thm:VIX} simplifies, as shown in the next result.

\begin{corollary}
Under the settings of Theorem~\ref{Thm:VIX}.
The rate function for OTM VIX call options $J_V(K,S_0,V_0;0)$ for the uncorrelated $(\rho=0)$ 
local-stochastic volatility model (\ref{LSvol}) is given by
\begin{align}
& J_V(K,S_0,V_0;0) =
\inf_{y,z}
\left\{\frac{1}{2z}\left(\int_{S_{0}}^{(\eta^{2})^{-1}(K^{2}e^{-y})}\frac{dx}{x\eta(x)}\right)^{2}
+H(y,z)\right\}\,,
\end{align}
for $K > \eta(S_0) \sqrt{V_0}$
and the rate function for OTM VIX put option
is also given by $J_{V}(K,S_0,V_0;0)$ with
$K < \eta(S_0) \sqrt{V_0}$.
\end{corollary}

Next, we present the asymptotics for ATM VIX options.

\begin{theorem}\label{thm:VIX:ATM}
Suppose the assumptions in Theorem~\ref{thm:European:ATM} hold. Furthermore, assume that $\tau = O(T^{1+\epsilon})$ for some 
$\epsilon >0$ in~\eqref{tau:order} under the settings of Corollary~\ref{cor:VIXsmalltau} and suppose Assumption~\ref{assump:LDP} holds. The asymptotics of ATM VIX options are given by
\begin{align}
&\lim_{T\rightarrow 0}\frac{1}{\sqrt{T}}C_{V}(K,T)
=\lim_{T\rightarrow 0}\frac{1}{\sqrt{T}}P_{V}(K,T)
\nonumber
\\
&=\frac{1}{\sqrt{2\pi}}\sqrt{\left((\eta(S_{0})\frac{1}{2}\sigma(V_{0})\sqrt{V_{0}}+\eta'(S_{0})\eta(S_{0})S_{0}V_{0}\rho\right)^{2}+\left(\eta'(S_{0})\eta(S_{0})S_{0}V_{0}\sqrt{1-\rho^{2}}\right)^{2}}.
\end{align}
\end{theorem}

Theorem~\ref{thm:VIX:ATM} shows that the prices of ATM VIX options
are of the order $\sqrt{T}$ as $T\rightarrow 0$, 
and it provides the exact formula for the leading-order term.

\section{The function $H(y,z)$}
\label{sec:H}

We study in this section in more detail the function
\begin{equation}\label{Hgen}
H(y,z):=\inf_{\substack{h(0)=\log V_{0},h(1)=y\\
\int_{0}^{1}e^{h(t)}dt=z}}\frac{1}{2}\int_{0}^{1}\left(\frac{h'(t)}{\sigma(e^{h(t)})}\right)^{2}dt\,,
\end{equation}
which appears in the short maturity limit of both European and VIX options in the local-stochastic volatility model considered. 

The function $H(z)$ appearing in the short-maturity limit for the European options is related to $H(y,z)$ as
\begin{equation}
\inf_{y\geq 0} H(y,z) = H(z)\,.
\end{equation}

Next, we will show that the function $H(y,z)$ can be evaluated 
explicitly for two commonly used vol-of-vol functions $\sigma(\cdot)$.

\subsubsection{Constant $\sigma(v) \equiv \sigma$}

This case corresponds to log-normal type process for $V_t$. For this case we have an explicit result.

\begin{proposition}\label{prop:Hlognorm}
The function $H(y,z)$ for the case $\sigma(v) = \sigma$ is given by
\begin{equation}\label{HLN}
H(y,z) = \frac{1}{2\sigma^2} I\left(\frac{z}{V_0}, \frac{e^{\frac12 y}}{\sqrt{V_0}} \right)\,,
\end{equation}
where $I(u,v)$ is 
\begin{equation}\label{Iuv}
I(u,v) = 8 F(v/u) + 4 \frac{1+v^2}{u} - 4\pi^2 \,.
\end{equation}
The function $F(\rho)$ is defined as
\begin{equation}\label{Fdef}
F(\rho) :=
\begin{cases}
\frac12 x_1^2 - \rho \cosh x_1 + \frac{\pi^2}{2} \,, & 0 < \rho < 1\,, \\
-\frac12 y_1^2 + \rho \cos y_1 + \pi y_1 \,, & \rho \geq 1\,, \\
\end{cases}
\end{equation}
where $x_1, y_1$ are the solutions of the equation
\begin{equation}
\rho \frac{\sinh x_1}{x_1} = 1\,,\qquad
y_1 + \rho \sin y_1 = \pi\,.
\end{equation}
\end{proposition}

\begin{remark}
An equivalent solution for $I(u,v)$ was given also in 
Proposition~4 in \cite{PZSABR}, where it was obtained by solving the Euler-Lagrange equation for the variational problem (\ref{Hgen}). The solution given here appears also in Section 3.2 in \cite{PZSABR} and is more convenient for practical applications.
\end{remark}

For numerical evaluation of $F(\rho)$ it is convenient to use the expansion around $\rho=1$
\begin{equation}\label{Fexp}
F(\rho) = \frac{\pi^2}{2}-1 - \log\rho + \log^2\rho + \frac{2}{15}\log^3\rho + O\left(\log^4 \rho\right) \,.
\end{equation}
The coefficients of the first ten terms in this expansion are tabulated in Section 5 of \cite{Nandori2022}. This series converges for $|\log\rho| < 3.42925$, see Proposition 4.2(i) in \cite{Nandori2022}. Outside of the convergence region, the function $F(\rho)$ can be well approximated by tail expansions for $\rho\to 0,\infty$ obtained in \cite{Pirjol2020}.

The function $I(u,v)$ has the following properties:

(i) $I(1,1)  = 0$. At this point the optimal path $h(t) = \log V_0$ is constant.

(ii) The function $I(u,v)$ has an expansion around its minimum at $u=v=1$ as
\begin{equation}
I(u,v) = 12 \log^2 u - 24 \log u \log v + 16\log^2 v + \cdots\,,
\end{equation}
where the terms neglected are of order $O(\log^a u \log^b v)$ with $a + b \geq 3$.
This is easily obtained from using the expansion (\ref{Fexp}) in (\ref{Iuv}).

\subsubsection{$\sigma(v) = \sigma v^{-\frac12}$}

This corresponds to a Heston-type model, where the variance process has a square-root type volatility:
\begin{equation}\label{VH}
dV_t = \mu(V_t) V_t dt + \sigma \sqrt{V_t} dW_t\,.
\end{equation}

\begin{proposition}\label{prop:Heston:2}
The function $H(y,z)$ for the Heston-type model is 
\begin{equation}
H(y,z) = V_0 I_H\left( \frac{z}{V_0}, \frac{e^y}{V_0} \right)\,,
\end{equation}
where $I_H(x,y)$ is a rate function giving the joint asymptotics of the time-integral
and terminal value for the process (\ref{VH}) as $T\to 0$. In this limit 
$\mathbb{Q}\left(\frac{1}{T V_0} \int_0^T V_t dt \in\cdot, \frac{V_T}{V_0} \in\cdot\right) $ satisfies a LDP with rate function
\begin{equation}\label{IH}
I_H(x,y) = \sup_{\theta,\phi} \left[\theta x + \phi y - \Lambda_H(\theta, \phi)\right]\,,
\end{equation}
where the cumulant function is
\begin{equation}\label{cumulant:limit}
\Lambda_H(\theta,\phi) := \lim_{T\to 0} T \log \mathbb{E}\left[e^{\frac{\theta}{T^2} \int_0^T V_t dt + \frac{\phi}{T} V_T}\right] \,.
\end{equation}
\end{proposition}

We are now in a position to compute the rate function $I_H(x,y)$ from (\ref{IH}). 
The result can be put into an explicit form as a double expansion in 
$\epsilon_x := \log x, \epsilon_y := \log y$.

\begin{proposition}\label{prop:HHeston}
The first few terms in the expansion of the rate function $I_H(x,y)$ for the square root model $\sigma(v) = \sigma/\sqrt{v}$ are given by:
\begin{align}
I_H(x,y) &= \frac{1}{\sigma^2} \Big\{
6 \epsilon_x^2 - 6 \epsilon_x \epsilon_y + 2 \epsilon_y^2+ \frac{12}{5} \epsilon_x^3 - \frac{3}{5} \epsilon_x^2 \epsilon_y - \frac{11}{5} \epsilon_x \epsilon_y^2 + \frac{6}{5}\epsilon_y^3 \nonumber\\
&\qquad\qquad\qquad+ \frac{271}{350} \epsilon_x^4 - \frac{61}{175} \epsilon_x^3 \epsilon_y
+ \frac{39}{350} \epsilon_x^2 \epsilon_y^2 - \frac{129}{175}\epsilon_x\epsilon_y^3
+ \frac{473}{1050} \epsilon_y^4 + O(\epsilon^5) \Big\}\,,\label{IHexp}
\end{align}
where we denote by $\epsilon^k$ the set of all terms of the form $\epsilon_x^i \epsilon_y^j$ with $i+j=k$.
\end{proposition}

\begin{remark}
As a consistency check, we can calculate that 
$J_1(x) = \inf_{y} I(x,y)$ 
has the expansion:
$J_1(x) = \frac{1}{\sigma^2}(\frac32 \epsilon_x^2 + \frac35 \epsilon_x^3 + \frac{271}{1400} \epsilon_x^4 
+ O(\epsilon_x^5))$.
The first three terms reproduce the expansion of the rate function for Asian options in the square-root model, given in equation (19) of \cite{PZAsianCEV}.
\end{remark}

\begin{remark}
In a similar way, we get that
$J_2(y) := \inf_{\epsilon_x} I_H(x,y)$ 
has the expansion
\begin{equation}
J_2(y) = \frac{1}{\sigma^2}\left(\frac23 \epsilon_y^2 + \frac14 \epsilon_y^3 + \frac{7}{96} \epsilon_y^4 
+ O\left(\epsilon_x^5\right) \right)\,,
\end{equation}
which is the same as the expansion of the rate function for European options in the square-root model
\begin{equation}
J_E(y) = \frac{2}{\sigma^2} \left(e^{\frac12 y} - 1\right)^2\,.
\end{equation}
This follows by substituting $\sigma(x) = \sigma/\sqrt{x}$ into
\begin{equation}
J_E(K,S_0) = \frac12 \Big( \int_{S_0}^K \frac{dx}{x\sigma(x)}\Big)^2
= \frac12 \Big( \frac{2\sqrt{S_0}}{\sigma} \Big( \sqrt{K/S_0} - 1\Big) \Big)^2 = \frac{2S_0}{\sigma^2}(e^{\frac12 \log\frac{K}{S_0}}-1)^2\nonumber
\end{equation}
and taking $S_0=1$.
\end{remark}

\section{Detailed predictions and comparison with the literature}\label{sec:applications}

We present in this section predictions following from the theoretical results obtained above. We start in Section~\ref{sec:6.1} with the example of a simple stochastic volatility model, the log-normal SABR model, for which the exact short maturity asymptotics is known. We show that the results of this paper reproduce the known results for this case. 

In Sections~\ref{sec:6.2} and \ref{sec:6.3} we discuss two local-stochastic volatility models with popular volatility specification: log-normal (SABR-type) and square-root (Heston-type) volatility, respectively. 
For both cases we derive analytical results for the ATM implied volatility and skew for both European and VIX options, for arbitrary local volatility function $\eta(x)$. These expressions are relevant for calibration to SPX and VIX smiles. We show that our results reduce to previously known expressions in various limiting cases of pure stochastic volatility models $(\eta(x)=1)$ and of the uncorrelated local-stochastic volatility models \cite{Forde2011}.

\subsection{Log-normal SABR model}
\label{sec:6.1}

The log-normal SABR model is obtained by taking $\eta(s) \equiv 1$
and log-normal volatility $\sigma(v)\equiv\sigma$.

\textbf{European options.}
The rate function of the European options given by Theorem~\ref{Thm:E} is
\begin{equation}
J_E(K) = \inf_{y,z}
\left\{\frac{1}{2(1-\rho^{2})z}\left(\int_{S_{0}}^{K}\frac{dx}{x\eta(x)}-\int_{V_{0}}^{e^{y}}\frac{\rho dx}{\sqrt{x}\sigma(x)}\right)^{2}
+H(y,z)\right\}\,.
\end{equation}

Taking here $\eta(x)\equiv 1$ and substituting the explicit form of the function $H(y,z)$ from (\ref{HLN}) we have
\begin{equation}
J_E(K) = \inf_{y,z}
\left\{\frac{1}{2(1-\rho^{2})z}\left(\log\frac{K}{S_0} - \frac{2\rho}{\sigma} \left(\sqrt{e^y} - \sqrt{V_0}\right) \right)^{2}
+ \frac{1}{2\sigma^2} I\left(\frac{z}{V_0}, \frac{e^{\frac12 y}}{\sqrt{V_0}} \right)\right\}\,.
\end{equation}

Denote $u := \frac{z}{V_0}$ and $w:=e^{\frac12 y}/\sqrt{V_0}$. The rate function becomes
\begin{align}
J_E(K) &= \inf_{u,w}
\left\{\frac{1}{2(1-\rho^{2}) V_0 u}\left(\log\frac{K}{S_0} - \frac{2\rho \sqrt{V_0}}{\sigma}  (w-1) \right)^{2}
+ \frac{1}{2\sigma^2} I(u, w )\right\}\nonumber \\
& = \frac{1}{\sigma^2} \inf_{u,w}
\left\{\frac{2}{(1-\rho^{2}) u}\left( \frac{\sigma}{2\sqrt{V_0}} \log\frac{K}{S_0} - \rho (w-1) \right)^{2}
+ \frac{1}{2} I(u, w )\right\}\,. \label{JESABR}
\end{align}

Let us compare this with the rate function for the short maturity asymptotics of 
European options in the log-normal SABR model given in equation (6.2) of \cite{PZSABR}.
Expressed in the notations of the current paper, the SDE of the model $dS_t = \sigma_t S_t dW_t\,,
d\sigma_t = \omega \sigma_t dZ_t$ becomes $dS_t = \sqrt{V_t} S_t dW_t\,,
dV_t = 2 \omega V_t dZ_t$, which corresponds to $\sigma(v) \equiv 2\omega$.
The rate function from \cite{PZSABR} is
\begin{equation}\label{J2d}
J(K) = \inf_{u,v>0} \left\{ \frac{2}{(1-\rho^2) u} 
\left( \frac{\omega}{\sqrt{V_0}} \log \frac{K}{S_0} - \rho (v-1) \right)^2 +
\frac12 I(u,v) \right\}\,.
\end{equation}
Substituting $\omega = \sigma/2$ we see that they agree.

In \cite{PZSABR} evidence has been presented that the solution of the extremal problem (\ref{J2d}) can be expressed in closed form as
\begin{equation}\label{Jsol}
J(K) = 2 \log^2 \left(\frac{\sqrt{1+2\rho \zeta + \zeta^2} + \zeta + \rho}{1+\rho}\right)\,,\qquad 
\zeta := \frac{\omega}{\sqrt{V_0}} \log\left(\frac{K}{S_0}\right)\,.
\end{equation}
This was tested by verifying that it correctly reproduces the first two terms of the series expansion in $x=\log\frac{K}{S_0}$ of the solution of the extremal problem (\ref{J2d}), and also by comparing with numerical solution of the extremal problem. In the uncorrelated limit $\rho=0$ the analytical result was proved explicitly. 

The result (\ref{Jsol}) reproduces the well-known formula for the short-maturity asymptotics of the implied volatility in the 
log-normal SABR model \cite{SABR}
\begin{equation}
\sigma_{\mathrm{BS}}(K,S_0) = \sqrt{V_0} \frac{\zeta}{\log\left(\frac{\sqrt{1+2\rho \zeta + \zeta^2} + \zeta + \rho}{1+\rho}\right)}\,.
\end{equation}

\textbf{VIX options.} We consider next the short-maturity asymptotics of VIX options in the mean-reverting log-normal SABR model with volatility specification
\begin{equation*}
dV_t = a(b-V_t) dt + \sigma V_t dZ_t \,.
\end{equation*}

This is a particular case of the class of models covered by Proposition~\ref{prop:VIX}. For this case we have $\mathrm{VIX}_T = \sqrt{\alpha(\tau) V_T + \beta(\tau)}$ with $\alpha(\tau) = \frac{1-e^{-a \tau}}{a\tau}$ and $\beta(\tau) = b(1 - \alpha(\tau))$. 

The VIX futures price is $F_V(T) = \sqrt{\alpha(\tau) V_0 + \beta(\tau)} + O(T)$. 
As $T\to 0$, VIX call options are OTM for $K > F_V(0)$ and VIX put options are ITM for $K< F_V(0)$.

The short maturity limit of the OTM VIX options is given by Proposition~\ref{prop:VIX}
with the replacement $\mathcal{F}(v) = \sqrt{\alpha(\tau) v + \beta(\tau)}$. We get
\begin{equation}
J_V(K) = \frac{1}{2\sigma^2} \log^2\left(\frac{ K^2 - \beta(\tau)}{\alpha(\tau) V_0}\right)\,.
\end{equation}

The short-maturity limit of the VIX implied volatility is given by:
\begin{equation}\label{VIXimpvol}
\lim_{T\to 0} \sigma_{\mathrm{VIX}}^2(K, V_0,T) = \frac{\log^2\left(\frac{K}{F_V(0)}\right)}{2J_V(K)}
= \sigma^2 \frac{\log^2\left(\frac{K}{F_V(0)}\right)}{\log^2\left(\frac{ K^2 - \beta(\tau)}{\alpha(\tau) V_0}\right)}\,.
\end{equation}
This agrees with the result in Sec.~1.8.1 of Forde and Smith \cite{Forde2023}. Denote the
short-maturity asymptotics of the VIX implied volatility given by (\ref{VIXimpvol}) as $\sigma_{VIX}(K,V_0)$.

The first few terms in the expansion of the asymptotic VIX implied volatility in powers of log-strike $z=\log\frac{K}{F_V(0)}$ are
\begin{equation}\label{sigVIXexp}
\sigma_{\mathrm{VIX}}(K,V_0) = \frac{\sigma}{2} \left\{ \frac{\alpha V_0}{\alpha V_0 + \beta}
+ \frac{\beta}{\alpha V_0 + \beta} z - 
\frac{(2\alpha V_0 + \beta) \beta}{3\alpha V_0 (\alpha V_0 + \beta )} z^2 + O\left(z^3\right)
\right\} \,.
\end{equation} 

The asymptotic VIX implied volatility (\ref{VIXimpvol}) has the following properties:

\begin{itemize}

\item The VIX implied volatility vanishes for $K \leq \sqrt{\beta(\tau)}$, since the VIX is bounded 
below as $\mathrm{VIX}_T \geq \sqrt{\beta(\tau)}$.

\item From the expansion (\ref{sigVIXexp}) we see that the ATM VIX volatility is 
$$\sigma_{\mathrm{VIX}}(K=F_V(0),V_0) = \frac{\sigma}{2} \cdot \frac{\alpha(\tau) V_0} {\alpha(\tau) V_0 + \beta(\tau)}\,.$$

\item For $\beta(\tau) >0$, the smile is up-sloping and concave in $K$. For $\beta(\tau)=0$ the smile is flat with $\sigma_{\mathrm{VIX}}(K,V_0) = \frac12 \sigma$.
\end{itemize}

\subsection{Local-stochastic volatility model with log-normal volatility}
\label{sec:6.2}

In this section we consider the local-stochastic volatility model with log-normal volatility
\begin{equation}\label{LSlognorm}
dS_t  =  S_t \sqrt{V_t} \eta(S_t) dW_t \,,\qquad
dV_t = V_t \mu(V_t) dt + \sigma V_t dZ_t\,,
\end{equation}
where $W_t,Z_t$ are correlated with correlation $\rho$.
Denote $\eta_0 = \eta(S_0), \eta_1 = S_0 \eta'(S_0), \eta_2=\frac12 S_0 \eta'(S_0)+\frac12 S_0^2 \eta''(S_0)$ the first few coefficients in the expansion of the local volatility function $\eta(x)$ in powers of log-price $\log S$ around $S_0$.

We give analytical results for the ATM volatility, skew and convexity of the implied volatility for the European and VIX options in this model.

\subsubsection{European options} The implied volatility of European options in the lognormal local-stochastic volatility model is given by the following result.

\begin{proposition}\label{prop:Elognorm}
The implied volatility of European options in the model (\ref{LSlognorm}) has the expansion in log-strike 
$k = \log(K/S_0)$
\begin{equation}
\sigma_{\mathrm{BS}}(k) = \sigma_{\mathrm{BS}}(0) + s_{E}\cdot k + 
\kappa_E k^2 + O(k^3)\,,
\end{equation}
where the at-the-money implied volatility is 
\begin{equation}\label{SPXvolATM}
\sigma_{\mathrm{BS}}(0) = \eta_0 \sqrt{V_0}\,,
\end{equation}
the ATM skew is
\begin{equation}\label{SPXskewATM}
s_E = \frac14 \left(\rho \sigma + 2\eta_1 \sqrt{V_0} \right)\,,
\end{equation}
and the ATM convexity is 
\begin{equation}\label{SPXcvxATM}
\kappa_E = \frac{(2-3\rho^2) \sigma^2 + 4(4\eta_0\eta_2 - \eta_1^2) V_0}
{48\eta_0 \sqrt{V_0}} \,.
\end{equation}

\end{proposition}

\begin{remark}
The result (\ref{SPXvolATM}) reproduces the result of Theorem~\ref{thm:European:ATM} for the ATM European options.
\end{remark}

\begin{remark}
The ATM skew (\ref{SPXskewATM}) is the sum of two terms, which correspond to the skew in the stochastic volatility model (obtained by taking $\eta(\cdot)\equiv 1$) $(s_{E})_1 = \frac14\rho\sigma$, and to the skew in the local volatility model (obtained in the limit $\sigma=0$)
$(s_{E})_2 = \frac12 \eta_1\sqrt{V_0}$. 
\end{remark}

\begin{remark}
A similar decomposition holds also for the ATM convexity (\ref{SPXcvxATM}).
The first term is the ATM convexity in the log-normal SABR model, and the second term is the ATM smile convexity in the local volatility model.
\end{remark}

\begin{remark}
In the uncorrelated limit $\rho=0$, the results for the European options ATM volatility, skew and 
convexity reproduce the results in Theorem~4.1 of Forde and Jacquier (2011) \cite{Forde2011} for local-stochastic volatility models by specializing to the log-normal $V_t$ process.
\end{remark}

\subsubsection{VIX options}
We give next the ATM expansion of the implied volatility of the VIX options.
As shown in Corollary~\ref{cor:VIXsmalltau}, we can approximate the VIX options as $$C_V(K,T) = e^{-rT}\mathbb{E}\left[\left(\eta(S_T) \sqrt{V_T} - K\right)^+\right]\,,$$
and the corrections to this approximation are of order $O(\tau^\frac12)$.
The VIX futures price is also approximated as $F_V(T) = \eta(S_{0}) \sqrt{V_0} + O(T)$.

\begin{proposition}\label{prop:Vlognorm}
The implied volatility of VIX options in the model (\ref{LSlognorm}) has the expansion in log-strike 
$x_{\mathrm{VIX}} = \log\left(K/(\eta(S_{0})\sqrt{V_0})\right)$:
\begin{equation}
\sigma_{\mathrm{VIX}}(x_{\mathrm{VIX}}) = \sigma_{\mathrm{VIX,ATM}}(0) + s_{\mathrm{VIX}}\cdot x_{\mathrm{VIX}} + \kappa_{\mathrm{VIX}}\cdot x^2_{\mathrm{VIX}} + O(x^3_{\mathrm{VIX}})\,,
\end{equation}
where the at-the-money VIX implied volatility is 
\begin{equation}\label{VIXvolATM}
\sigma_{\mathrm{VIX,ATM}}(0) = \frac12 \sqrt{\sigma^2 + 4\rho \sigma \eta_1 \sqrt{V_0} + 4  \eta_1^2 V_0}\,,
\end{equation}
and the ATM VIX skew is
\begin{align}
s_{\mathrm{VIX}} = 
\frac12 \sqrt{V_0} \frac{\rho \sigma + 2\eta_1 \sqrt{V_0}}{(\sigma^2 + 4\eta_1 \rho \sigma \sqrt{V_0} + 4\eta_1^2 V_0)^{3/2}}
\cdot
\left(\sigma^2 \eta_1 + 2\rho \sigma \sqrt{V_0} (\eta_1^2 + 2 \eta_0 \eta_2) + 8\eta_0 \eta_1\eta_2 V_0 \right)\,,\label{sVIX}
\end{align}
and the ATM VIX convexity is
\begin{equation}
\kappa_{\mathrm{VIX}} = \frac{\sqrt{V_0}}{6}\frac{K_{VIX}}{(\sigma^2 + 4\eta_1 \rho \sigma \sqrt{V_0} + 4\eta_1^2 V_0)^{7/2}}\,,
\end{equation}
where the numerator $K_{\mathrm{VIX}}$ has a lengthy expression and is given in the Appendix~\ref{app:VIXcvx}.
\end{proposition}

\begin{remark}
The result (\ref{VIXvolATM}) reproduces the result of Theorem~\ref{thm:VIX:ATM} for the ATM VIX options
by taking into account that $S_0 \eta'(S_0) = \eta_1$. 
\end{remark}

\begin{remark}
In the stochastic volatility limit $\eta(x)=1$, the result (\ref{VIXvolATM}) reduces to \\
$\sigma_{\mathrm{VIX}}(\mathrm{ATM}) = \frac12\sigma$,
which is just the vol-of-vol of the stochastic process. In this limit the VIX skew vanishes $s_{\mathrm{VIX}}=0$.
\end{remark}

\begin{corollary}
Varying the correlation in the range $\rho \in [-1,1]$ gives bounds on the ATM VIX implied volatility
\begin{equation}
\left| \frac12\sigma - \eta_1 \sqrt{V_0} \right| \leq \sigma_{\mathrm{VIX}}\left(K=\eta(S_{0})\sqrt{V_0}\right)  \leq
\left| \frac12\sigma + \eta_1 \sqrt{V_0} \right|\,.
\end{equation}
\end{corollary}

\subsection{Local-stochastic volatility model with Heston-type volatility}
\label{sec:6.3}

In this section we consider the local-stochastic volatility model with square-root volatility, which we will call Heston-type
\begin{equation}\label{LSHeston}
dS_t  =  S_t \sqrt{V_t} \eta(S_t) dW_t \,,\qquad
dV_t = V_t \mu(V_t) dt + \sigma \sqrt{V_t} dZ_t\,,
\end{equation}
where $W_t,Z_t$ are correlated with correlation $\rho$.
Denote $\eta_0 = \eta(S_0), \eta_1 = S_0 \eta'(S_0), \eta_2=\frac12 S_0 \eta'(S_0)+\frac12 S_0^2 \eta''(S_0)$ the first few derivatives of the local volatility function around $S_0$.

This model takes $\sigma(S)=\sigma S^{-1/2}$ in Eqn.~\eqref{LSvol} which does not satisfy Assumptions~\ref{assump:bounded} and \ref{assump:LDP}. Although Theorems~\ref{Thm:E} and \ref{Thm:VIX} were obtained under the Assumptions~\ref{assump:bounded} and \ref{assump:LDP}, one can see that these results hold also for this case as long as $\mathbb{Q}(\{(\log S_{Tt},\log V_{Tt}),0\leq t\leq 1\}\in\cdot)$ satisfies a sample-path large deviation principle as in the proof of Theorem~\ref{Thm:E} (which is true for Heston-type and CEV-type SDEs without Assumptions~\ref{assump:bounded} and \ref{assump:LDP}, see e.g. \cite{Baldi}) and the following two assumptions hold.
The first assumption is on the finiteness of the moment generating function of the integrated variance $\int_0^t V_udu$, 
which is known to hold under the Heston model.

\begin{assumption}\label{assump:V:integral}
We assume that for any $\theta>0$, there exists some $C_{\theta}\in(0,\infty)$, such that $\mathbb{E}[e^{\theta \int_0^T V_u du}]\leq C_{\theta}$ for any sufficiently small $T>0$.    
\end{assumption}

The second assumption is the finiteness of the moments of $V_{t}$ process, which also holds under the Heston model.
\begin{assumption}
For any $p>1$, there exists some $C_{p}\in(0,\infty)$, such that $\max_{0\leq t\leq T}\mathbb{E}[(V_t)^p]\leq C_{p}$ for any sufficiently small $T>0$.   
\end{assumption}

We give next a result on 
moment finiteness for $S_{t}$ process under a certain assumption on the moment generating function of the integrated variance $\int_0^t V_udu$ so that Assumption~\ref{assump:S:T:p} is satisfied.

\begin{proposition}\label{prop:S:moments}
Suppose that $\eta(\cdot)$ is uniformly bounded. Also, suppose Assumption~\ref{assump:V:integral} holds. Then for any $p>1$, there exists some $C'_{p}\in(0,\infty)$, such that $\max_{0\leq t\leq T}\mathbb{E}[(S_t)^p]\leq C'_{p}$ for any sufficiently small $T>0$.
\end{proposition}

Finally, we note that the extremal problems for the rate functions $J_E(K), J_V(K)$ appearing in Theorem~\ref{Thm:E} and \ref{Thm:VIX} are well defined, and the function $H(y,z)$ is calculable as shown in Proposition~\ref{prop:HHeston}.

\subsubsection{European options} The implied volatility of European options in the Heston-type local-stochastic volatility model is given by the following result.

\begin{proposition}\label{prop:EHeston}
The implied volatility of European options in the Heston-type model (\ref{LSHeston}) has the expansion in log-strike 
$k = \log(K/S_0)$
\begin{equation}
\sigma_{\mathrm{BS}}(k) = \sigma_{\mathrm{BS}}(0) + s_{E}\cdot k + 
\kappa_E k^2 + O(k^3)\,,
\end{equation}
where the at-the-money implied volatility is 
\begin{equation}\label{SPXvolATMHeston}
\sigma_{\mathrm{BS}}(0) = \eta_0 \sqrt{V_0}\,,
\end{equation}
the ATM skew is
\begin{equation}\label{SPXskewATMHeston}
s_E = \frac{1}{4\sqrt{V_0}} \left(\rho \sigma + 2\eta_1 V_0 \right)\,,
\end{equation}
and the ATM convexity is 
\begin{equation}\label{SPXcvxATMHeston}
\kappa_E = \frac{(2-5\rho^2) \sigma^2 + 4(4\eta_0\eta_2 - \eta_1^2) V_0}
{48\eta_0 V_0^{3/2}} \,.
\end{equation}
\end{proposition}

\begin{remark}
In the stochastic volatility limit $\eta(x)=1$ we can compare the results with the short-maturity implied volatility expansion around the ATM point for the Heston model, which is known from the literature \cite{Lewis1,Durrleman2004,Forde2009}.
Expressed in our notations, the first three terms in this expansion are given by:
\begin{equation}
\sigma_{\mathrm{BS}}^{\mathrm{Heston}}(K) = 
\sqrt{V_0} \left( 1 + \frac{\rho\sigma}{4V_0} k + \frac{1}{24} \left(1 - \frac52 \rho^2\right) \frac{\sigma^2 k^2}{V_0^2} + O\left(k^3\right) \right)\,.
\end{equation}


Our results for skew and convexity reproduce the coefficients in this 
expansion after taking $\eta_0\to 1, \eta_{1,2}\to 0$.
\end{remark}

\begin{remark}
The results for the European options ATM volatility, skew and 
convexity in the local Heston model reproduce the results in equations (4.4)-(4.6) of Bompis and Gobet (2018) \cite{Bompis2018}. 
\end{remark}

\subsubsection{VIX options} The implied volatility of VIX options in the Heston-type local-stochastic volatility model is given by the following result.
For this result, as shown in Corollary~\ref{cor:VIXsmalltau}, VIX options can be approximated as
\begin{equation}
    C_V(K,T) = e^{-rT} \mathbb{E}\left[\left(\eta(S_T) \sqrt{V_T} - K\right)^+\right]\,.
\end{equation}
The corrections to this approximation are of order $O(\tau^{1/2})$.

\begin{proposition}\label{prop:VHeston}
The implied volatility of VIX options in the model (\ref{LSHeston}) has the expansion in log-strike $x_{\mathrm{VIX}} = \log(K/(\eta_0 \sqrt{V_0}))$:
\begin{equation}
\sigma_{\mathrm{VIX}}(x_{\mathrm{VIX}}) = \sigma_{\mathrm{VIX}}(0)  + s_{\mathrm{VIX}}\cdot x_{\mathrm{VIX}} + O\left(x^2_{\mathrm{VIX}}\right)\,,
\end{equation}
where the at-the-money VIX implied volatility is 
\begin{equation}\label{VIXvolATMHeston}
\sigma_{\mathrm{VIX}}(0) = \frac{1}{\sqrt{V_0}} \sqrt{\frac14 \sigma^2 + \eta_1\rho \sigma V_0 + \eta_1^2 V_0^2} \,,
\end{equation}
and the ATM VIX skew is
\begin{equation}\label{VIXskewATMHeston}
s_{\mathrm{VIX}} = 
\frac{1}{4\sqrt{V_0}} 
\frac{-\sigma^4 - 2\eta_1 \rho V_0 \sigma^3 + 4 \sigma^2 V_0^2 (\eta_1^2 + 2\eta_0 \eta_2 \rho^2) + 8\eta_1 \rho V_0^3 \sigma (4\eta_0\eta_2 + \eta_1^2) + 32 \eta_0 \eta_1^2 \eta_2 V_0^4}
{(\sigma^2 + 4\eta_1\rho \sigma V_0 + 4\eta_1^2 V_0^2)^{3/2}}\,.
\end{equation}
\end{proposition}

\begin{remark}
The ATM VIX volatility (\ref{VIXvolATMHeston})
agrees with the prediction from Theorem~\ref{thm:VIX:ATM}.
\end{remark}

We can compare these results with the prediction for VIX options in the stochastic model with Heston-type volatility process $dV_t = \sigma \sqrt{V_t} dW_t + a(b - V_t) dt$, similar to the analysis in Section \ref{sec:6.1} for the SABR-type model.
As before, we apply the results of Proposition~\ref{prop:VIX}
with $\mathrm{VIX}_T = \sqrt{\alpha(\tau) V_T + \beta(\tau)}$ and $\alpha(\tau) = \frac{1-e^{-a \tau}}{a\tau}$ and $\beta(\tau) = b(1 - \alpha(\tau))$. 

The VIX futures price is $F_V(T) = \sqrt{\alpha(\tau) V_0 + \beta(\tau)} + O(T)$. 
As $T\to 0$, VIX call options are OTM for $K > F_V(0)$ and VIX put options are ITM for $K< F_V(0)$.

The short maturity limit of the OTM VIX options is given by Proposition~\ref{prop:VIX} with $\sigma(v) = \sigma v^{-1/2}$ and $\mathcal{F}(v) = \sqrt{\alpha(\tau) v + \beta(\tau)}$. The rate function for OTM VIX options is
\begin{equation}
J_V(K) = \frac{1}{2\sigma^2} \left( \int_{V_0}^{\mathcal{F}^{-1}(K^2)}
\frac{dx}{\sqrt{x}} \right)^2 = \frac{2}{\sigma^2}
( \sqrt{\mathcal{F}^{-1}(K^{2})} - \sqrt{V_0})^2 =
\frac{2}{\sigma^2}
\left( \sqrt{\frac{K^2-\beta(\tau)}{\alpha(\tau)}} - \sqrt{V_0}\right)^2\,,
\end{equation}
such that the short-maturity limit of the VIX implied volatility is:
\begin{equation}\label{VIXimpvolHeston}
\lim_{T\to 0} \sigma_{\mathrm{VIX}}(K, V_0,T) = 
\frac{\sigma}{2} \frac{\log(K/\sqrt{V_0})}
{\sqrt{\frac{K^2-\beta(\tau)}{\alpha(\tau)}} - \sqrt{V_0}}\,.
\end{equation}

This agrees with the result in Sec.~1.8.2 of Forde and Smith \cite{Forde2023}. They choose equal mean-reverting level and spot variance $b=V_0$ to simplify the result, but the result above is more general and holds for all parameters.

In the small averaging time $\tau\to 0$ (or equivalently small mean-reversion limit $a\to 0$) we have $\alpha(\tau)\to 1,\beta(\tau)\to 0$ and the asymptotic VIX smile becomes
\begin{equation}\label{sigVIXHeston}
\lim_{T\to 0,\tau\to 0} \sigma_{\mathrm{VIX}}(K, V_0,T) = 
\frac{\sigma}{2} \frac{\log(K/\sqrt{V_0})}
{K - \sqrt{V_0}}\,.
\end{equation}

The first few terms in the expansion of the asymptotic VIX implied volatility in powers of log-strike $z=\log\frac{K}{F_V(0)}$ (with $F_V(0)=\sqrt{V_0}$) are
\begin{equation}\label{sigVIXexpHeston}
\sigma_{\mathrm{VIX}}(K,V_0,0) = \frac{\sigma}{2\sqrt{V_0}} 
\frac{z}{e^z-1}=
\frac{\sigma}{2\sqrt{V_0}} \left\{ 
1 - \frac12 z + \frac{1}{12}z^2 + O(z^3) \right\}\,.
\end{equation} 
The VIX smile in the Heston model is down-sloping and convex, which is 
well known to contradict empirical evidence from market data, and disfavors this model for modeling volatility products.

The Heston model result (\ref{sigVIXHeston}) can be compared with the general results of Proposition \ref{prop:VHeston} for the VIX smile in the Heston-type local-stochastic volatility model. Taking $\eta_0=1, \eta_{1,2}=0$ in 
Proposition \ref{prop:VHeston} we get
\begin{equation}
    \sigma_{\mathrm{VIX}}(0) = \frac{\sigma}{2\sqrt{V_0}}\,,\qquad
    s_{\mathrm{VIX}} = - \frac{\sigma}{4\sqrt{V_0}}\,,
\end{equation}
which reproduces the first two coefficients of the series expansion for the Heston model (\ref{sigVIXexpHeston}).


\section{Numerical Illustrations}\label{sec:numerical}

In this section we compare the asymptotic results for the implied volatility of 
European and VIX options with the actual implied volatility, obtained by Monte Carlo simulations of 
a local-stochastic volatility model. For this test we choose the local-stochastic volatility model

\begin{equation}\label{SDETanh}
dS_t  =  \eta(S_t) S_t \sqrt{V_t} dW_t \,,\qquad \frac{dV_t}{V_t} = \sigma dZ_t\,,
\end{equation}
where $W_t,Z_t$ are correlated standard Brownian motions with correlation $\rho$.
The local volatility function is taken as
\begin{equation}\label{loc:vol:eta:S}
\eta(S) := f_0 + f_1 \tanh \left( \log\frac{S}{S_0} - x_0 \right)\,.
\end{equation}
This is the so-called Tanh model which was used in Forde and Jacquier (2011) \cite{Forde2011}
 to test the predictions of their asymptotic results for the uncorrelated local-stochastic volatility model. 
The coefficients $\eta(x),\sigma(v)$ for the model (\ref{SDETanh}) are bounded, and satisfy the technical conditions assumed in our paper. 

The local volatility function \eqref{loc:vol:eta:S} is expanded in powers of the log-asset $\log(S/S_0)$ as
\begin{equation}
\eta(S) = \eta_0 + \eta_1 \log\frac{S}{S_0} +  \eta_2 \log^2\frac{S}{S_0} + \cdots\,,
\end{equation}
with 
\begin{eqnarray}
&& \eta_0 := f_0 - f_1 \tanh x_0\,, \\
&& \eta_1 := \frac{f_1}{\cosh^2 x_0}\,,  \\
&& \eta_2 := \frac{f_1}{\cosh^2 x_0} \tanh x_0 \,.
\end{eqnarray}

The short-maturity asymptotics of the implied volatility of European and VIX options in the model (\ref{SDETanh}) were obtained in Section~\ref{sec:6.2}. The asymptotic predictions for European options are given in Proposition~\ref{prop:Elognorm} and those for the VIX options in Proposition~\ref{prop:Vlognorm}.
The information about the local volatility function $\eta(x)$ enter these predictions only through the expansion coefficients $\eta_{0,1,2}$. We will compare these predictions against Monte Carlo simulations of the model. 

\textbf{Model parameters.} In the numerical test, we will assume that the parameters for the local volatility function $\eta(x)$ are given by:
\begin{equation}\label{params1}
f_0 = 1.0\,,\quad f_1 = -0.5\,,\quad x_0 = 0\,,
\end{equation}
and the parameters of the volatility process are
\begin{equation}\label{params2}
\sigma = 2.0\,,\quad V_0 = 0.1 \,.
\end{equation}
The spot asset price is taken as $S_0=1$.

The correlation $\rho$ will be varied in the range $\{ -0.7, 0, +0.7\}$. The MC simulation will use $N_{\mathrm{MC}}=100k$ paths and $n=200$ time steps. The variance $V_t$ is simulated exactly as a geometric Brownian motion, and the process for $S_t$ is simulated using a Euler discretization.

\subsection{Numerical tests for European options}

The short-maturity asymptotic implied volatility of the European options will be approximated as a quadratic function of log-strike
\begin{equation}\label{EurTh}
\sigma_{\mathrm{BS}}(k) = \sigma_{\mathrm{ATM}} + s_E k + \kappa_E k^2\,,
\end{equation}
where the ATM level $\sigma_{\mathrm{ATM}}$, skew $s_E$ and convexity $\kappa_E$ are given in Proposition~\ref{prop:Elognorm}. Their numerical values for this test are shown in Table~\ref{tab:tanh}.

The asymptotic result (\ref{EurTh}) is shown as the solid curve in Figure~\ref{Fig:Eur}. 
The red dots show the results of a MC simulation for European options with maturity $T=1/12$ (1 month). The agreement is reasonably good for strikes sufficiently close to the ATM point.

\begin{figure}[h]
\centering
\includegraphics[width=1.9in]{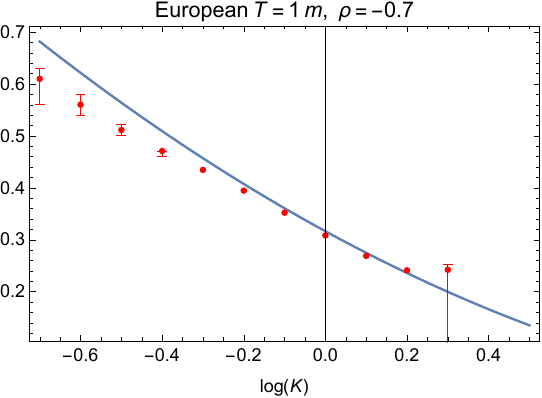}
\includegraphics[width=1.9in]{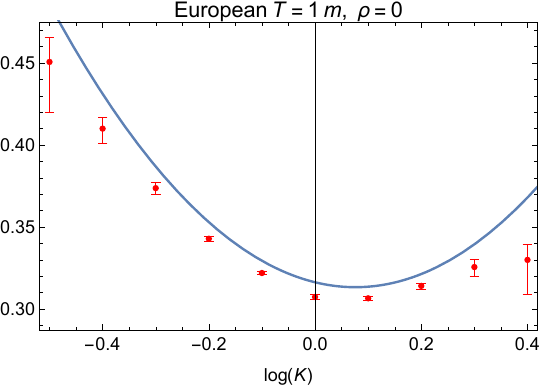}
\includegraphics[width=1.9in]{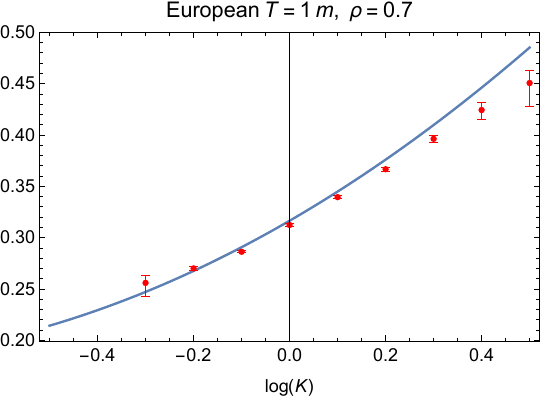}
\caption{Numerical tests for European option pricing in the Tanh model with parameters (\ref{params1}), (\ref{params2}) 
and correlation
$\rho =  - 0.7, 0, +0.7$ respectively. The solid curve is the asymptotic prediction (\ref{EurTh}) for the implied volatility and
the red dots show the MC result for European options with maturity $T=1/12$ (1 month).}
\label{Fig:Eur}
\end{figure}

\begin{table}[h!]
  \centering
  \caption{The parameters for the short-maturity asymptotics of the European and VIX options in the Tanh model used for the numerical test. }
    \begin{tabular}{|c|ccc||ccc|}
    \hline
$\rho$ & $\sigma_{E,\mathrm{ATM}}$ & $s_E$ & $\kappa_E$ 
           & $\sigma_{\mathrm{VIX},\mathrm{ATM}}$ & $s_{\mathrm{VIX}}$ & $\kappa_{\mathrm{VIX}}$ \\
    \hline\hline
$-0.7$ & 0.316 & $-0.429$ & 0.133 & 
          1.116  & 0.054  & 0.004 \\
0 & 0.316 & $-0.079$ & 0.520 & 
      1.012 & 0.012 & 0.002 \\
$+0.7$ & 0.316 & 0.271 & 0.133 & 0.896 & $-0.053$ & $-0.005$ \\
    \hline
    \end{tabular}%
  \label{tab:tanh}%
\end{table}%

\subsection{Numerical tests for VIX options}

Next consider the VIX options. 
We compute a quadratic approximation for the VIX implied volatility as
\begin{equation}\label{VIXth}
\sigma_{\mathrm{VIX}}(x_{\mathrm{VIX}}) = \sigma_{\mathrm{VIX},\mathrm{ATM}} + s_{\mathrm{VIX}}\cdot x_{\mathrm{VIX}} + \kappa_{\mathrm{VIX}} 
x_{\mathrm{VIX}}^2\,, 
\end{equation}
where $x_{\mathrm{VIX}} := \log\frac{K}{\mathrm{VIX}_0}$ with $\mathrm{VIX}_0:=\eta_0 \sqrt{V_0}$.
The ATM level $\sigma_{\mathrm{VIX},ATM}$, skew $s_{\mathrm{VIX}}$ and convexity $\kappa_{\mathrm{VIX}}$ are given in Proposition~\ref{prop:Vlognorm}. Their numerical values corresponding to the parameters (\ref{params1}), (\ref{params2}) are listed in Table~\ref{tab:tanh}.

\begin{figure}[h]
\centering
\includegraphics[width=1.9in]{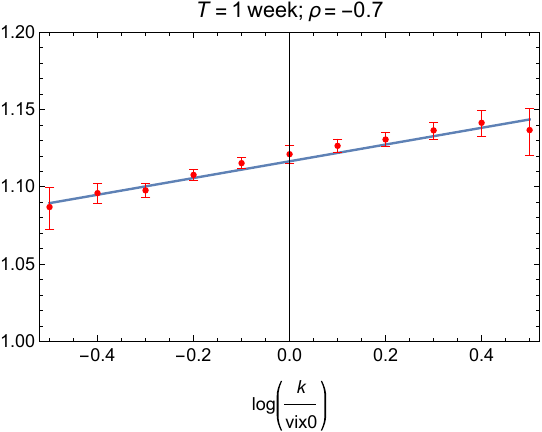}
\includegraphics[width=1.9in]{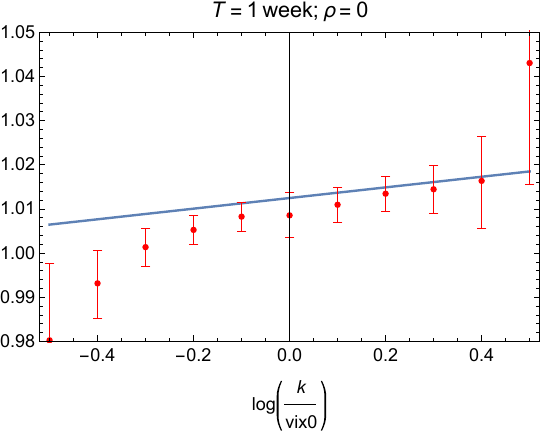}
\includegraphics[width=1.9in]{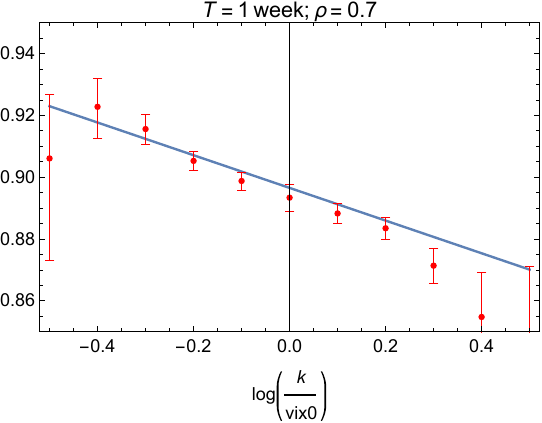}
\caption{Numerical tests for VIX option pricing in the Tanh model with parameters (\ref{params1}), (\ref{params2}) 
and correlation
$\rho = - 0.7, 0, +0.7$ respectively. The solid curve is the asymptotic prediction (\ref{VIXth}) for the VIX implied volatility and 
the red dots show the MC simulation for 
VIX options with maturity $T=1/52$ (1 week).}
\label{Fig:VIX}
\end{figure}

The asymptotic prediction (\ref{VIXth}) is shown as the solid curve in Figure~\ref{Fig:VIX}. The red dots show the results of a MC simulation of the model for VIX options with maturity $T=1/52$ (1 week).
The range of strikes covered in the testing is constrained by the spread of the values of $\mathrm{VIX}_T$ in the simulation. This is sufficiently wide, even for the shorter maturity $T=1/52$ considered.
(On the other hand, the range of simulated values for $S_T$ at $T=1/52$ is less dispersed, so in order to get a wider range of strikes we used a longer maturity $T=1/12$ for the European options testing.)
The agreement of the asymptotic result for the VIX implied volatility with the MC simulation is again reasonably good for strikes sufficiently close to the ATM point. 

\section*{Acknowledgements}
Xiaoyu Wang is supported by the Guangzhou-HKUST(GZ) Joint Funding Program\\ (No.2024A03J0630), 
Guangzhou Municipal Key Laboratory of Financial Technology Cutting-Edge Research.
Lingjiong Zhu is partially supported by the grants NSF DMS-2053454, NSF DMS-2208303.

\appendix

\section{Background on Large Deviations Theory}\label{app:LD}

We give in this Appendix a few basic concepts of large deviations theory from probability theory
which are in the proofs.
We refer to Varadhan \cite{VaradhanLD}, Dembo and Zeitouni \cite{Dembo1998} for more details on large deviations and its applications.

\begin{definition}[Large Deviation Principle]
A sequence $(P_\epsilon)_{\epsilon \in \mathbb{R}^+}$ of probability measures
on a topological space $X$ satisfies the large deviation principle with rate function $I: X \to \mathbb{R}$
if $I$ is non-negative, lower semicontinuous and for any measurable set $A$, we have
\begin{equation}
- \inf_{x\in A^o} I(x) \leq \liminf_{\epsilon\to 0} \epsilon \log P_\epsilon(A) \leq
\limsup_{\epsilon\to 0} \epsilon \log P_\epsilon(A) \leq - \inf_{x\in \bar A} I(x) \,,
\end{equation}
where $A^o$ denotes the interior of $A$ and $\bar A$ its closure.
\end{definition}

\begin{theorem}[Contraction Principle, see e.g. Theorem 4.2.1. in \cite{Dembo1998}]\label{Contraction:Thm}
If $F:X\rightarrow Y$ is a continuous map and 
$P_{\epsilon}$ satisfies a large deviation principle on $X$ with the rate 
function $I(x)$,
then the probability measures $Q_{\epsilon}:=P_{\epsilon}F^{-1}$ satisfies
a large deviation principle on $Y$ with the rate function
$J(y)=\inf_{x: F(x)=y}I(x)$.
\end{theorem}

\section{Proofs of the Main Results}\label{sec:proofs:main:results}

We give in this Appendix the proofs of the main results in the paper.

\subsection{Model specification}

\begin{proof}[Proof of Proposition~\ref{prop:V:moments}]
For any $p\geq 1$, we can compute that
\begin{align}
\mathbb{E}\left[V_{t}^{p}\right]
&=V_{0}^{p}\mathbb{E}\left[e^{\int_{0}^{t}(p\mu(V_{u})-\frac{p}{2}\sigma^{2}(V_{u}))du+p\int_{0}^{t}\sigma(V_{u})dZ_{u}}\right]
\nonumber
\\
&\leq
V_{0}^{p}e^{pM_{\mu}t+\frac{p^{2}}{2}M_{\sigma}^{2}t}
\mathbb{E}\left[e^{-\int_{0}^{t}\frac{p^{2}}{2}\sigma^{2}(V_{u})du+p\int_{0}^{t}\sigma(V_{u})dZ_{u}}\right]
\nonumber
\\
&\leq V_{0}^{p}e^{pM_{\mu}t+\frac{p^{2}}{2}M_{\sigma}^{2}t},
\end{align}
where we used Assumption~\ref{assump:bounded}
and the fact that $e^{-\int_{0}^{t}\frac{p^{2}}{2}\sigma^{2}(V_{u})du+p\int_{0}^{t}\sigma(V_{u})dZ_{u}}$
is a non-negative local martingale and thus a supermartingale.
Hence, we conclude that for any $p\geq 1$
\begin{equation}
\max_{0\leq t\leq T}\mathbb{E}\left[V_{t}^{p}\right]  
\leq V_{0}^{p}e^{pM_{\mu}T+\frac{p^{2}}{2}M_{\sigma}^{2}T}
=O(1),
\end{equation}
as $T\rightarrow 0$.
This completes the proof.
\end{proof}

\begin{proof}[Proof of Proposition~\ref{prop:LM}]
The result follows from Theorem 2.4(i) of Lions and Musiela \cite{Lions2008}.
\end{proof}


\subsection{European options}

\begin{proof}[Proof of Theorem~\ref{Thm:E}]
(i) OTM call options $K > S_0$.
The starting point of the proof is a relation between the small-time asymptotics of the call option price with $K>S_{0}$
and the small-time asymptotics of the density of the asset price in the right tail
\begin{equation}\label{limT}
\lim_{T\to 0} T \log \mathbb{E}\left[(S_{T}-K)^{+}\right] = \lim_{T\to 0}
T \log \mathbb{Q}(S_T \geq K)\,,\quad K > S_0.
\end{equation}
This relation follows by upper and lower bounds for \eqref{limT}. 

Let us first prove the upper bound for \eqref{limT}. 
We include the following argument for the sake
of completeness, which can be found in \cite{Friz2018}.
For any $U>K>S_{0}$, by applying H\"{o}lder's inequality, we have
\begin{align}
\mathbb{E}[(S_T - K)^+]  
&= \mathbb{E}[(S_T - K) 1_{S_T \in (K, U)}] + \mathbb{E}[(S_T - K) 1_{S_T \geq U}] 
\nonumber
\\
&\leq
(U - K) \mathbb{Q}(S_T \in (K, U)) + \left(\mathbb{E}[(S_T)^p ]\right)^{1/p} \left(\mathbb{E}[1_{S_T > U}]\right)^{1/q}
\nonumber
\\
&\leq
(U - K) \mathbb{Q}(S_T\geq K) + \left(\mathbb{E}[(S_T)^p ]\right)^{1/p} \left(\mathbb{Q}(S_T \geq U)\right)^{1/q}
\label{take:log}
\end{align}
for any $p,q>1$ such that $\frac{1}{p}+\frac{1}{q}=1$,
where $p$ is chosen 
such that which $\mathbb{E}[(S_T)^p]=O(1)$ as $T\rightarrow 0$ under Assumption~\ref{assump:S:T:p}.

By taking the logarithm in \eqref{take:log} and multiplying with $T$ and letting $T\rightarrow 0$, we obtain
\begin{align}
\limsup_{T\rightarrow 0}T\log 
\mathbb{E}[(S_T - K)^+]  
\leq
\max\left\{\limsup_{T\rightarrow 0}T\log\mathbb{Q}(S_{T}\geq K),\frac{1}{q}\limsup_{T\rightarrow 0}T\log\mathbb{Q}(S_{T}\geq U)\right\}.\label{eqn:K:U}
\end{align}
Next, let us show that the limits
$\lim_{T\rightarrow 0}T\log\mathbb{Q}(S_{T}\geq K)$ and $\lim_{T\rightarrow 0}T\log\mathbb{Q}(S_{T}\geq U)$ exist.


Under Assumptions~\ref{assump:bounded} and \ref{assump:LDP}, by the sample-path large deviations for 
small time diffusions (see for example \cite{Varadhan} and \cite{Robertson2010}),
one can see that $\mathbb{Q}(\{(\log S_{tT},\log V_{tT}),0\leq t\leq 1\}\in\cdot)$
satisfies a sample-path large deviation principle with the rate function:
\begin{equation}\label{rate:function:LDP}
\frac{1}{2(1-\rho^{2})}\int_{0}^{1}\left(\frac{g'(t)}{\eta(e^{g(t)})\sqrt{e^{h(t)}}}-\frac{\rho h'(t)}{\sigma(e^{h(t)})}\right)^{2}dt
+\frac{1}{2}\int_{0}^{1}\left(\frac{h'(t)}{\sigma(e^{h(t)})}\right)^{2}dt,    
\end{equation}
with $g(0)=\log S_{0}$, $h(0)=\log V_{0}$ and $g,h$ being absolutely continuous
and the rate function is $+\infty$ otherwise.

By an application of the contraction principle (see for example Theorem 4.2.1. in \cite{Dembo1998}, restated in Theorem~\ref{Contraction:Thm}),
one can compute that
\begin{align}
&\lim_{T\rightarrow 0}T\log\mathbb{Q}\left(S_{T}\geq K\right)
\nonumber
\\
&=-\inf_{\substack{g(0)=\log S_{0}\\
h(0)=\log V_{0}\\
g(1)=\log K}}
\left\{\frac{1}{2(1-\rho^{2})}\int_{0}^{1}\left(\frac{g'(t)}{\eta(e^{g(t)})\sqrt{e^{h(t)}}}-\frac{\rho h'(t)}{\sigma(e^{h(t)})}\right)^{2}dt
+\frac{1}{2}\int_{0}^{1}\left(\frac{h'(t)}{\sigma(e^{h(t)})}\right)^{2}dt\right\}.
\end{align}
Similarly, we can obtain the limit
$\lim_{T\rightarrow 0}T\log\mathbb{Q}(S_{T}\geq U)$
with $\lim_{U\rightarrow\infty}\lim_{T\rightarrow 0}T\log\mathbb{Q}(S_{T}\geq U)=-\infty$.
Since $U>K>S_{0}$ is arbitrary, by letting $U\rightarrow\infty$ in \eqref{eqn:K:U},
we obtain the upper bound for \eqref{limT}, i.e.
\begin{equation*}
\limsup_{T\to 0} T \log \mathbb{E}\left[(S_{T}-K)^{+}\right] 
\leq\limsup_{T\to 0}
T \log \mathbb{Q}(S_T \geq K)\,,\quad K > S_0.
\end{equation*}
The argument for the lower bound for \eqref{limT} is standard, see e.g. \cite{Pham2007} and we omit the details here.
Hence, we proved \eqref{limT}.

A similar relation holds between the small-time asymptotics of the put options 
and of the density of $S_T$ in the left wing ($K < S_0$).

For both cases, based on the previous discussions, the limit \eqref{limT} can be computed using large deviations theory as:
\begin{align}
\label{rate:function}
&\lim_{T\rightarrow 0}T\log\mathbb{Q}\left(S_{T}\geq K\right)
\nonumber
\\
&=-\inf_{\substack{g(0)=\log S_{0}\\
h(0)=\log V_{0}\\
g(1)=\log K}}
\left\{\frac{1}{2(1-\rho^{2})}\int_{0}^{1}\left(\frac{g'(t)}{\eta(e^{g(t)})\sqrt{e^{h(t)}}}-\frac{\rho h'(t)}{\sigma(e^{h(t)})}\right)^{2}dt
+\frac{1}{2}\int_{0}^{1}\left(\frac{h'(t)}{\sigma(e^{h(t)})}\right)^{2}dt\right\}.
\end{align}

Given $h$, we can determine the optimal $g$ as follows.
By Cauchy-Schwarz inequality, we have
\begin{align}
&\int_{0}^{1}\left(\frac{g'(t)}{\eta(e^{g(t)})\sqrt{e^{h(t)}}}-\frac{\rho h'(t)}{\sigma(e^{h(t)})}\right)^{2}dt
\cdot\int_{0}^{1}\left(\sqrt{e^{h(t)}}\right)^{2}dt
\nonumber
\\
&\geq
\left(\int_{0}^{1}\left(\frac{g'(t)}{\eta(e^{g(t)})}-\frac{\rho h'(t)\sqrt{e^{h(t)}}}{\sigma(e^{h(t)})}\right)dt\right)^{2},
\end{align}
where the integrals on the right-hand side can be expressed in a simpler form as
\begin{align}
\int_{0}^{1}\frac{g'(t)}{\eta(e^{g(t)})}dt
=\int_{0}^{1}\frac{d(e^{g(t)})}{e^{g(t)}\eta(e^{g(t)})}
=\int_{e^{g(0)}}^{e^{g(1)}}\frac{dx}{x\eta(x)},
\label{Thm:E:call:1}
\end{align}
where $e^{g(0)}=S_{0}$ and $e^{g(1)}=K$, 
and
\begin{equation}
\int_{0}^{1}\frac{\rho h'(t)\sqrt{e^{h(t)}}}{\sigma(e^{h(t)})}dt
=\int_{0}^{1}\frac{\rho h'(t)e^{h(t)}\sqrt{e^{h(t)}}}{e^{h(t)}\sigma(e^{h(t)})}dt
=\int_{e^{h(0)}}^{e^{h(1)}}\frac{\rho dx}{\sqrt{x}\sigma(x)},
\end{equation}
where $h(0)=\log V_{0}$.
Therefore, we have
\begin{align}
\frac{1}{2}\int_{0}^{1}\left(\frac{g'(t)}{\eta(e^{g(t)})\sqrt{e^{h(t)}}}-\frac{\rho h'(t)}{\sigma(e^{h(t)})}\right)^{2}dt
\geq
\frac{1}{2}\left(\int_{S_{0}}^{K}\frac{dx}{x\eta(x)}-\int_{V_{0}}^{e^{h(1)}}\frac{\rho dx}{\sqrt{x}\sigma(x)}\right)^{2}
\left(\int_{0}^{1}e^{h(t)}dt\right)^{-1},
\end{align}
and by Cauchy-Schwarz inequality, the equality is achieved
when 
\begin{equation}
\frac{g'(t)}{\eta(e^{g(t)})}-\frac{\rho\sqrt{e^{h(t)}}h'(t)}{\sigma(e^{h(t)})}=C_{1}e^{h(t)},
\end{equation}
for some constant $C_{1}$ so that $g(t)$ can be solved
via the equation:
\begin{equation}
\int_{S_{0}}^{e^{g(t)}}\frac{dx}{x\eta(x)}-\int_{V_{0}}^{e^{h(t)}}\frac{\rho dx}{\sqrt{x}\sigma(x)}=C_{1}\int_{0}^{t}e^{h(s)}ds,
\end{equation}
where
\begin{equation}
C_{1}=\frac{\int_{S_{0}}^{K}\frac{dx}{x\eta(x)}-\int_{V_{0}}^{e^{h(1)}}\frac{\rho dx}{\sqrt{x}\sigma(x)}}{\int_{0}^{1}e^{h(s)}ds}.
\end{equation}
Since with fixed $h$, we can solve for the optimal $g$,
by the discussions above, we conclude that
\begin{align}
&\lim_{T\rightarrow 0}T\log\mathbb{Q}\left(S_{T}\geq K\right)
\nonumber
\\
&=-\inf_{h(0)=\log V_{0}}
\Bigg\{\frac{1}{2(1-\rho^{2})}\left(\int_{S_{0}}^{K}\frac{dx}{x\eta(x)}-\int_{V_{0}}^{e^{h(1)}}\frac{\rho dx}{\sqrt{x}\sigma(x)}\right)^{2}
\left(\int_{0}^{1}e^{h(t)}dt\right)^{-1}
\nonumber
\\
&\qquad\qquad\qquad\qquad\qquad
+\frac{1}{2}\int_{0}^{1}\left(\frac{h'(t)}{\sigma(e^{h(t)})}\right)^{2}dt\Bigg\}
\nonumber
\\
&=-\inf_{y,z}
\left\{\frac{1}{2(1-\rho^{2})z}\left(\int_{S_{0}}^{K}\frac{dx}{x\eta(x)}-\int_{V_{0}}^{e^{y}}\frac{\rho dx}{\sqrt{x}\sigma(x)}\right)^{2}
+\inf_{\substack{h(0)=\log V_{0},h(1)=y\\
\int_{0}^{1}e^{h(t)}dt=z}}\frac{1}{2}\int_{0}^{1}\left(\frac{h'(t)}{\sigma(e^{h(t)})}\right)^{2}dt\right\}
\nonumber
\\
&=-\inf_{y,z}
\left\{\frac{1}{2(1-\rho^{2})z}\left(\int_{S_{0}}^{K}\frac{dx}{x\eta(x)}-\int_{V_{0}}^{e^{y}}\frac{\rho dx}{\sqrt{x}\sigma(x)}\right)^{2}
+H(y,z)\right\},
\label{Thm:E:call:2}
\end{align}
where
\begin{equation}
H(y,z):=\inf_{\substack{h(0)=\log V_{0},h(1)=y\\
\int_{0}^{1}e^{h(t)}dt=z}}\frac{1}{2}\int_{0}^{1}\left(\frac{h'(t)}{\sigma(e^{h(t)})}\right)^{2}dt.
\end{equation}

(ii) OTM put options $K < S_0$. 
The case for OTM put options is analogous to the case for call options. Similar to~\eqref{limT}, we have
\begin{equation}
\lim_{T\to 0} T \log \mathbb{E}[(K - S_{T})^{+}] = \lim_{T\to 0}
T \log \mathbb{Q}(K \geq S_{T})\,,\quad K < S_0.
\end{equation}
By large deviations theory, the rate function for $\lim_{T\rightarrow 0}T\log\mathbb{Q}\left(K \geq S_{T}\right)$ is the same as \eqref{rate:function}. Following the steps to get \eqref{Thm:E:call:1}, we compute
\begin{equation}\label{put:variation:1}
\int_{0}^{1}\frac{g'(t)}{\eta(e^{g(t)})}dt
=\int_{0}^{1}\frac{d(-e^{-g(t)})}{e^{-g(t)}\eta(e^{g(t)})} = \int_{e^{-g(1)}}^{e^{-g(0)}}\frac{dx}{x\eta(x^{-1})},
\end{equation}
where $e^{g(1)} = K < e^{g(0)} = S_0$. And
\begin{equation}\label{put:variation:2}
\int_{0}^{1}\frac{\rho h'(t)\sqrt{e^{h(t)}}}{\sigma(e^{h(t)})}dt
=\int_{0}^{1}\frac{-\rho h'(t)e^{-h(t)}\sqrt{e^{h(t)}}}{-e^{-h(t)}\sigma(e^{h(t)})}dt
=\int_{e^{-h(1)}}^{e^{-h(0)}}\frac{\rho dx}{\sqrt{x^3}\sigma(x^{-1})},
\end{equation}
where $h(0)=\log V_{0}$. Therefore, we can get the inequality as follows
\begin{align}
&\frac{1}{2}\int_{0}^{1}\left(\frac{g'(t)}{\eta(e^{g(t)})\sqrt{e^{h(t)}}}-\frac{\rho h'(t)}{\sigma(e^{h(t)})}\right)^{2}dt
\nonumber
\\
&\geq
\frac{1}{2}\left(\int_{K^{-1}}^{S_0^{-1}}\frac{dx}{x\eta(x^{-1})}-\int_{e^{-h(1)}}^{V_0^{-1}}\frac{\rho dx}{\sqrt{x^3}\sigma(x^{-1})}\right)^{2}
\left(\int_{0}^{1}e^{h(t)}dt\right)^{-1},
\end{align}
and by Cauchy-Schwarz inequality, the equality is achieved
when $\frac{g'(t)}{\eta(e^{g(t)})}-\frac{\rho\sqrt{e^{h(t)}}h'(t)}{\sigma(e^{h(t)})}$ and $e^{h(t)}$ are linearly dependent. Hence, $g(t)$ can be solved
via the equation:
\begin{equation}
\int_{K^{-1}}^{S_0^{-1}}\frac{dx}{x\eta(x^{-1})}-\int_{e^{-h(1)}}^{V_0^{-1}}\frac{\rho dx}{\sqrt{x^3}\sigma(x^{-1})}=C_{1}\int_{0}^{t}e^{h(s)}ds,
\end{equation}
where
\begin{equation}
C_{1}=\frac{\int_{K^{-1}}^{S_0^{-1}}\frac{dx}{x\eta(x^{-1})}-\int_{e^{-h(1)}}^{V_0^{-1}}\frac{\rho dx}{\sqrt{x^3}\sigma(x^{-1})}}{\int_{0}^{1}e^{h(s)}ds}.
\end{equation}
Since with fixed $h$, we can solve for the optimal $g$ similar to~\eqref{Thm:E:call:2} such that 
\begin{align}
& \lim_{T\to 0}
T \log \mathbb{Q}(K \geq S_{T}) \nonumber \\
&=-\inf_{y,z}
\left\{\frac{1}{2(1-\rho^{2})z}\left(\int_{K^{-1}}^{S_0^{-1}}\frac{dx}{x\eta(x^{-1})}-\int_{e^{-y}}^{V_0^{-1}}\frac{\rho dx}{\sqrt{x^3}\sigma(x^{-1})}\right)^{2}
\right.\nonumber \\
& \left.\qquad\qquad\qquad\qquad\qquad\qquad\qquad\qquad\qquad\qquad +\inf_{\substack{h(0)=\log V_{0},h(1)=y\\
\int_{0}^{1}e^{h(t)}dt=z}}\frac{1}{2}\int_{0}^{1}\left(\frac{h'(t)}{\sigma(e^{h(t)})}\right)^{2}dt\right\}
\nonumber
\\
&=-\inf_{y,z}
\left\{\frac{1}{2(1-\rho^{2})z}\left(\int_{S_0^{-1}}^{K^{-1}}\frac{dx}{x\eta(x^{-1})}-\int_{V_0^{-1}}^{e^{-y}}\frac{\rho dx}{\sqrt{x^3}\sigma(x^{-1})}\right)^{2}
+H(y,z)\right\}.
\label{Thm:E:put}
\end{align}
Note that $\int_{K^{-1}}^{S_0^{-1}}\frac{dx}{x\eta(x^{-1})} < 0$ when $K<S_0$, and $H(y,z)$ is the same as the one defined in~\eqref{Thm:E:call:2}. Hence, the optimum is achieved in the regime $0 < z < V_0$ where $h'(t) < 0$ and $H(z)$ is decreasing. The last equation in \eqref{Thm:E:put} holds by the simple fact that $(a-b)^2 = (b-a)^2$.
Finally, by changing the variable $x\mapsto x^{-1}$ in \eqref{Thm:E:put}, we have
\begin{align}
\left(\int_{S_0^{-1}}^{K^{-1}}\frac{dx}{x\eta(x^{-1})}-\int_{V_0^{-1}}^{e^{-y}}\frac{\rho dx}{\sqrt{x^3}\sigma(x^{-1})}\right)^{2}
&=\left(-\int_{S_0}^{K}\frac{dx}{x\eta(x)}+\int_{V_0}^{e^{y}}\frac{\rho dx}{\sqrt{x}\sigma(x)}\right)^{2}
\nonumber
\\
&=\left(\int_{S_0}^{K}\frac{dx}{x\eta(x)}-\int_{V_0}^{e^{y}}\frac{\rho dx}{\sqrt{x}\sigma(x)}\right)^{2}.
\end{align}
This completes the proof.
\end{proof}

\begin{proof}[Proof of Theorem~\ref{thm:European:ATM}]
We only provide the proof for ATM European call option.
The case for the ATM European put option can be handled
similarly.

\textbf{Step 1.} 
First, we define a Gaussian approximation for $S_t$ as
\begin{equation}
\hat{S}_{t}=S_{0}+\eta(S_{0})S_{0}\sqrt{V_{0}}\left(\rho Z_{t}+\sqrt{1-\rho^{2}}B_{t}\right),\qquad 0\leq t\leq T,
\end{equation}
where $Z_{t}$ and $B_{t}$ are independent standard Brownian motions
and we will show that $S_t$ can be approximated by $\hat{S}_{t}$ in the $L_{2}$-norm.
We can rewrite $\hat{S}_{t}$ as
\begin{equation}
\hat{S}_{t}=S_{0}+\int_{0}^{t}\eta(S_{0})S_{0}\sqrt{V_{0}}dW_{s},
\end{equation}
where $W_{s}:=\rho Z_{s}+\sqrt{1-\rho^{2}}B_{s}$ is a standard Brownian motion and has correlation $\rho$ with $Z_{s}$.
Recall that
\begin{equation}
S_{t}=S_{0}+\int_{0}^{t}(r-q)S_{s}ds+\int_{0}^{t}\eta(S_{s})S_{s}\sqrt{V_{s}}dW_{s}.
\end{equation}
Therefore, 
\begin{align}
&\mathbb{E}\left|S_{t}-\hat{S}_{t}\right|^{2}
\nonumber
\\
&\leq
2\mathbb{E}\left[\left(\int_{0}^{t}(r-q)S_{s}ds\right)^{2}\right]
+2\mathbb{E}\left[\left(\int_{0}^{t}\left(\eta(S_{s})S_{s}\sqrt{V_{s}}-\eta(S_{0})S_{0}\sqrt{V_{0}}\right)dW_{s}\right)^{2}\right].\label{two:terms:to:bound}
\end{align}

\textbf{Step 2.}
Next, let us provide an upper bound for the first term in \eqref{two:terms:to:bound}.
By Cauchy-Schwarz inequality,
\begin{align}
\mathbb{E}\left[\left(\int_{0}^{t}(r-q)S_{s}ds\right)^{2}\right]
\leq
(r-q)^{2}t\int_{0}^{t}\mathbb{E}[S_{s}^{2}]ds
\leq
C_{1}(r-q)^{2}t^{2},
\end{align}
where $C_{1}:=\max_{0\leq t\leq T}\mathbb{E}[S_{t}^{2}]=O(1)$ as $T\rightarrow 0$ under our assumptions.

\textbf{Step 3.}
Next, let us provide an upper bound for the second term in \eqref{two:terms:to:bound}.
By It\^{o}'s isometry, 
\begin{align}
&\mathbb{E}\left[\left(\int_{0}^{t}\left(\eta(S_{s})S_{s}\sqrt{V_{s}}-\eta(S_{0})S_{0}\sqrt{V_{0}}\right)dW_{s}\right)^{2}\right]
\nonumber
\\
&=\int_{0}^{t}\mathbb{E}\left[\left(\eta(S_{s})S_{s}\sqrt{V_{s}}-\eta(S_{0})S_{0}\sqrt{V_{0}}\right)^{2}\right]ds
\nonumber
\\
&\leq
2\int_{0}^{t}\mathbb{E}\left[\left(\eta(S_{s})S_{s}\sqrt{V_{s}}-\eta(S_{s})S_{s}\sqrt{V_{0}}\right)^{2}\right]ds
\nonumber
\\
&\qquad\qquad\qquad\qquad\qquad
+2\int_{0}^{t}\mathbb{E}\left[\left(\eta(S_{s})S_{s}\sqrt{V_{0}}-\eta(S_{0})S_{0}\sqrt{V_{0}}\right)^{2}\right]ds.\label{two:terms:to:bound:2}
\end{align}

Let us first bound the second term in \eqref{two:terms:to:bound:2}. 
Since $\eta$ is $L$-Lipschitz and $M_{\eta}$-uniformly bounded, we can deduce that for any $t$:
$$
\vert\eta(S_t)S_t - \eta(S_0)S_0\vert = \vert\eta(S_t)(S_t - S_0) + (\eta(S_t) - \eta(S_0))S_0\vert \leq L_{\eta}\vert S_t - S_0\vert,
$$
where $L_{\eta} := M_{\eta}+S_0L$.
Thus, we get
\begin{align}
&2\int_{0}^{t}\mathbb{E}\left[\left(\eta(S_{s})S_{s}\sqrt{V_{0}}-\eta(S_{0})S_{0}\sqrt{V_{0}}\right)^{2}\right]ds
\nonumber
\\
&\leq
2V_{0}L_{\eta}^{2}\int_{0}^{t}\mathbb{E}\left[\left(S_{s}-S_{0}\right)^{2}\right]ds
\nonumber
\\
&\leq
4V_{0}L_{\eta}^{2}\int_{0}^{t}\mathbb{E}\left[\left(S_{s}-\hat{S}_{s}\right)^{2}\right]ds
+4V_{0}L_{\eta}^{2}\int_{0}^{t}\mathbb{E}\left[\left(\hat{S}_{s}-S_{0}\right)^{2}\right]ds
\nonumber
\\
&=4V_{0}L_{\eta}^{2}\int_{0}^{t}\mathbb{E}\left[\left(S_{s}-\hat{S}_{s}\right)^{2}\right]ds
+4V_{0}L_{\eta}^{2}\eta^{2}(S_{0})S_{0}^{2}V_{0}\int_{0}^{t}sds
\nonumber
\\
&=4V_{0}L_{\eta}^{2}\int_{0}^{t}\mathbb{E}\left[\left(S_{s}-\hat{S}_{s}\right)^{2}\right]ds
+2V_{0}L_{\eta}^{2}\eta^{2}(S_{0})S_{0}^{2}V_{0}t^{2}.
\end{align}

Next, let us bound the first term in \eqref{two:terms:to:bound:2}. We can compute that
\begin{align}
&2\int_{0}^{t}\mathbb{E}\left[\left(\eta(S_{s})S_{s}\sqrt{V_{s}}-\eta(S_{s})S_{s}\sqrt{V_{0}}\right)^{2}\right]ds
\nonumber
\\
&\leq
2\int_{0}^{t}\left(\mathbb{E}\left[\left(\eta(S_{s})S_{s}\right)^{4}\right]\right)^{1/2}
\left(\mathbb{E}\left[\left(\sqrt{V_{s}}-\sqrt{V}_{0}\right)^{4}\right]\right)^{1/2}ds
\nonumber
\\
&\leq
2\int_{0}^{t}L_{\eta}^{2}\left(\mathbb{E}\left[(S_{s})^{4}\right]\right)^{1/2}
\left(\mathbb{E}\left[\left(\sqrt{V_{s}}-\sqrt{V}_{0}\right)^{4}\right]\right)^{1/2}ds
\nonumber
\\
&\leq
2L_{\eta}^{2}\sqrt{C_{2}}\int_{0}^{t}
\left(\mathbb{E}\left[V_{s}^{2}-4\sqrt{V_{0}}V_{s}^{3/2}+6V_{0}V_{s}-4V_{0}^{3/2}\sqrt{V_{s}}+V_{0}^{2}\right]\right)^{1/2}ds
\nonumber
\\
&\leq
2L_{\eta}^{2}\sqrt{C_{2}}\sqrt{t}\left(\int_{0}^{t}
\mathbb{E}\left[V_{s}^{2}-4\sqrt{V_{0}}V_{s}^{3/2}+6V_{0}V_{s}-4V_{0}^{3/2}\sqrt{V_{s}}+V_{0}^{2}\right]ds\right)^{1/2},\label{to:be:continued}
\end{align}
where $C_{2}:=\max_{0\leq t\leq T}\mathbb{E}[S_{t}^{4}]=O(1)$ as $T\rightarrow 0$ under our assumptions and we applied Cauchy-Schwarz inequality to obtain the last inequality above.

\textbf{Step 4.} In order to finish the calculations to bound
the second term in \eqref{two:terms:to:bound:2} in \textbf{Step 3}, 
we need to provide lower bounds for $\mathbb{E}[V_{s}^{1/2}]$ and $\mathbb{E}[V_{s}^{3/2}]$
and upper bounds for $\mathbb{E}[V_{s}]$ and $\mathbb{E}[V_{s}^{2}]$ that will be used
to complete the upper bound in \eqref{to:be:continued}. 
Let us recall that
\begin{equation}
V_{s}=V_{0}e^{\int_{0}^{s}(\mu(V_{u})-\frac{1}{2}\sigma^{2}(V_{u}))du+\int_{0}^{s}\sigma(V_{u})dZ_{u}},
\end{equation}
and by Jensen's inequality,
\begin{align}
\mathbb{E}\left[V_{s}^{1/2}\right]
&=V_{0}^{1/2}\mathbb{E}\left[e^{\int_{0}^{s}(\frac{1}{2}\mu(V_{u})-\frac{1}{4}\sigma^{2}(V_{u}))du+\frac{1}{2}\int_{0}^{s}\sigma(V_{u})dZ_{u}}\right]
\nonumber
\\
&\geq
V_{0}^{1/2}
e^{\mathbb{E}[\int_{0}^{s}(\frac{1}{2}\mu(V_{u})-\frac{1}{4}\sigma^{2}(V_{u}))du
+\int_{0}^{s}\frac{1}{2}\sigma(V_{u})dZ_{u}]}
\nonumber
\\
&=
V_{0}^{1/2}
e^{\mathbb{E}[\int_{0}^{s}(\frac{1}{2}\mu(V_{u})-\frac{1}{4}\sigma^{2}(V_{u}))du]}
\geq
V_{0}^{1/2}e^{-\frac{1}{2}sM_{\mu}-\frac{1}{4}sM_{\sigma}^{2}}.\label{V:1:2}
\end{align}
Similarly,
\begin{align}
\mathbb{E}\left[V_{s}^{3/2}\right]
=V_{0}^{3/2}\mathbb{E}\left[e^{\int_{0}^{s}(\frac{3}{2}\mu(V_{u})-\frac{3}{4}\sigma^{2}(V_{u}))du+\frac{3}{2}\int_{0}^{s}\sigma(V_{u})dZ_{u}}\right]
\geq
V_{0}^{3/2}e^{-\frac{3}{2}sM_{\mu}-\frac{3}{4}sM_{\sigma}^{2}}.\label{V:3:2}
\end{align}
On the other hand,
\begin{align}
\mathbb{E}[V_{s}]
&=V_{0}\mathbb{E}\left[e^{\int_{0}^{s}(\mu(V_{u})-\frac{1}{2}\sigma^{2}(V_{u}))du+\int_{0}^{s}\sigma(V_{u})dZ_{u}}\right]
\nonumber
\\
&\leq
V_{0}e^{sM_{\mu}}
\mathbb{E}\left[e^{\int_{0}^{s}(-\frac{1}{2}\sigma^{2}(V_{u}))du+\int_{0}^{s}\sigma(V_{u})dZ_{u}}\right]
=V_{0}e^{sM_{\mu}}.\label{V:s:upper:bound}
\end{align}
Moreover, 
\begin{align}
\mathbb{E}\left[V_{s}^{2}\right]
=V_{0}^{2}\mathbb{E}\left[e^{\int_{0}^{s}(2\mu(V_{u})-\sigma^{2}(V_{u}))du
+2\int_{0}^{s}\sigma(V_{u})dZ_{u}}\right]
\leq
V_{0}^{2}e^{2sM_{\mu}}
\mathbb{E}\left[e^{2\int_{0}^{s}\sigma(V_{u})dZ_{u}}\right],
\end{align}
and by Cauchy-Schwarz inequality, we can further compute that
\begin{align}
\mathbb{E}\left[e^{2\int_{0}^{s}\sigma(V_{u})dZ_{u}}\right]
&=\mathbb{E}\left[e^{\int_{0}^{s}2\sigma(V_{u})dZ_{u}-\int_{0}^{s}4\sigma^{2}(V_{u})du}e^{\int_{0}^{s}4\sigma^{2}(V_{u})du}\right]
\nonumber
\\
&\leq
\left(\mathbb{E}\left[e^{\int_{0}^{s}4\sigma(V_{u})dZ_{u}-\frac{1}{2}\int_{0}^{s}(4\sigma)^{2}(V_{u})du}\right]\right)^{1/2}
\left(\mathbb{E}\left[e^{8\int_{0}^{s}\sigma^{2}(V_{u})du}\right]\right)^{1/2}
\nonumber
\\
&=\left(\mathbb{E}\left[e^{8\int_{0}^{s}\sigma^{2}(V_{u})du}\right]\right)^{1/2}
\leq e^{4sM_{\sigma}^{2}}.
\end{align}
Hence, we have
\begin{equation}\label{V:s:upper:bound:2}
\mathbb{E}[V_{s}^{2}]
\leq
V_{0}^{2}e^{2sM_{\mu}}
e^{4sM_{\sigma}^{2}}.
\end{equation}

Hence, by applying \eqref{V:1:2}, \eqref{V:3:2}, \eqref{V:s:upper:bound} and \eqref{V:s:upper:bound:2}, 
we conclude that in the upper bound in \eqref{to:be:continued}, we have
\begin{align}
&\int_{0}^{t}\mathbb{E}\left[V_{s}^{2}-4\sqrt{V_{0}}V_{s}^{3/2}+6V_{0}V_{s}-4V_{0}^{3/2}\sqrt{V_{s}}+V_{0}^{2}\right]ds
\nonumber
\\
&\leq
V_{0}^{2}\int_{0}^{t}\left(e^{2sM_{\mu}}
e^{4sM_{\sigma}^{2}}
-4e^{-\frac{3}{2}sM_{\mu}-\frac{3}{4}sM_{\sigma}^{2}}
+6e^{sM_{\mu}}
-4e^{-\frac{1}{2}sM_{\mu}-\frac{1}{4}sM_{\sigma}^{2}}
+1\right)ds
\nonumber
\\
&\leq
C_{3}t^{2},
\end{align}
for some universal constant $C_{3}>0$.

\textbf{Step 5.}
Putting everything together, i.e. by combining the estimates in \textbf{Step 2}, \textbf{Step 3} and \textbf{Step 4}, 
we have for any $0\leq t\leq T$,
\begin{align}
\mathbb{E}\left|S_{t}-\hat{S}_{t}\right|^{2}
&\leq
2C_{1}(r-q)^{2}t^{2}
+8V_{0}L_{\eta}^{2}\int_{0}^{t}\mathbb{E}\left[\left(S_{s}-\hat{S}_{s}\right)^{2}\right]ds
\nonumber
\\
&\qquad\qquad
+4V_{0}L_{\eta}^{2}\eta^{2}(S_{0})S_{0}^{2}V_{0}t^{2}
+4L_{\eta}^{2}\sqrt{C_{2}}\sqrt{C_{3}}t^{3/2}.
\end{align}
By applying Gronwall's inequality, we conclude that
\begin{equation}
\mathbb{E}\left|S_{T}-\hat{S}_{T}\right|^{2}
\leq O(T^{3/2}),
\end{equation}
as $T\rightarrow 0$.

Since $x\mapsto x^{+}$ is $1$-Lipschitz, 
\begin{align}
\left|\mathbb{E}\left[\left(S_{T}-S_{0}\right)^{+}\right]-
\mathbb{E}\left[\left(\hat{S}_{T}-S_{0}\right)^{+}\right]\right|
\leq
\mathbb{E}\left|S_{T}-\hat{S}_{T}\right|
\leq O(T^{3/4}),
\end{align}
as $T\rightarrow 0$.

Finally, we can compute that
\begin{equation}
\mathbb{E}\left[\left(\hat{S}_{T}-S_{0}\right)^{+}\right]
=\mathbb{E}\left[\eta(S_{0})S_{0}\sqrt{V_{0}}\left(\rho Z_{T}+\sqrt{1-\rho^{2}}B_{T}\right)^{+}\right]
=\sqrt{T}\frac{\eta(S_{0})S_{0}\sqrt{V_{0}}}{\sqrt{2\pi}}.
\end{equation}
Therefore,
\begin{equation}
\lim_{T\rightarrow 0}\frac{1}{\sqrt{T}}C_{E}(S_{0},T)
=\frac{\eta(S_{0})S_{0}\sqrt{V_{0}}}{\sqrt{2\pi}}.
\end{equation}
This completes the proof.
\end{proof}

\subsection{VIX options}

\begin{proof}[Proof of Proposition~\ref{prop:VIXsmalltau}]
First of all, we have
\begin{align}
\left| \mathrm{VIX}_T^2 - V_T \eta^2(S_T) \right|  & = \left|\frac{1}{\tau} \int_T^{T+\tau} \mathbb{E}[V_s \eta^2(S_s)|\mathcal{F}_{T}] ds 
- V_T \eta^2(S_T) \right|\nonumber \\
& \leq
\frac{1}{\tau} \int_T^{T+\tau} \left| \mathbb{E}[V_s \eta^2(S_s) |\mathcal{F}_{T}]  - V_T \eta^2(S_T) \right| ds\,. 
\end{align}

The integrand is bounded as 
\begin{equation}\label{18}
\left| \mathbb{E}[V_s \eta^2(S_s) |\mathcal{F}_{T}]  - V_T \eta^2(S_T) \right| 
\leq 
\left| \mathbb{E}[( V_s - V_T) \eta^2(S_s) |\mathcal{F}_{T}] \right| +
\left| \mathbb{E}[V_T (\eta^2(S_s) - \eta^2(S_T)) |\mathcal{F}_{T}] \right| \,.
\end{equation}
We bound each term on the right-hand side separately, and we will show that their sum is of $O(\tau^{1/2})$.

\textbf{Step 1.} \textit{First term in (\ref{18}).}
The first term in \eqref{18} can be bounded as 
\begin{align}
\left|\mathbb{E}\left[V_{s}\eta^{2}(S_{s})|\mathcal{F}_{T}\right]
-\mathbb{E}\left[V_{T}\eta^{2}(S_{s})|\mathcal{F}_{T}\right]\right|
&\leq
M_{\eta}^{2}\left[\mathbb{E}|V_{s}-V_{T}||\mathcal{F}_{T}\right]
\nonumber
\\
&\leq
M_{\eta}^{2}\left(\left[\mathbb{E}|V_{s}-V_{T}|^{2}|\mathcal{F}_{T}\right]\right)^{1/2}.
\end{align}
We can further compute that
\begin{equation}
V_{s}=V_{T}e^{\int_{T}^{s}(\mu(V_{u})-\frac{1}{2}\sigma^{2}(V_{u}))du+\int_{T}^{s}\sigma(V_{u})dZ_{u}},
\end{equation}
and by Jensen's inequality and Assumption~\ref{assump:bounded},
\begin{align}
\mathbb{E}[V_{s}|\mathcal{F}_{T}]
&=V_{T}\mathbb{E}\left[e^{\int_{T}^{s}(\mu(V_{u})-\frac{1}{2}\sigma^{2}(V_{u}))du
+\int_{T}^{s}\sigma(V_{u})dZ_{u}}\Big|\mathcal{F}_{T}\right]
\nonumber
\\
&\geq
V_{T}
e^{\mathbb{E}[\int_{T}^{s}(\mu(V_{u})-\frac{1}{2}\sigma^{2}(V_{u}))du
+\int_{T}^{s}\sigma(V_{u})dZ_{u}|\mathcal{F}_{T}]}
\nonumber
\\
&=
V_{T}
e^{\mathbb{E}[\int_{T}^{s}(\mu(V_{u})-\frac{1}{2}\sigma^{2}(V_{u}))du|\mathcal{F}_{T}]}
\geq
V_{T}e^{-(s-T)M_{\mu}-\frac{1}{2}(s-T)M_{\sigma}^{2}}.
\end{align}
On the other hand, by Assumption~\ref{assump:bounded}, 
\begin{align}
\mathbb{E}[V_{s}^{2}|\mathcal{F}_{T}]
&=V_{T}^{2}\mathbb{E}\left[e^{\int_{T}^{s}(2\mu(V_{u})-\sigma^{2}(V_{u}))du
+2\int_{T}^{s}\sigma(V_{u})dZ_{u}}\Big|\mathcal{F}_{T}\right]
\nonumber
\\
&\leq
V_{T}^{2}e^{2(s-T)M_{\mu}}
\mathbb{E}\left[e^{2\int_{T}^{s}\sigma(V_{u})dZ_{u}}\Big|\mathcal{F}_{T}\right],
\end{align}
and by Cauchy-Schwarz inequality, we can further compute that
\begin{align}
&\mathbb{E}\left[e^{2\int_{T}^{s}\sigma(V_{u})dZ_{u}}\Big|\mathcal{F}_{T}\right]
\nonumber
\\
&=\mathbb{E}\left[e^{\int_{T}^{s}2\sigma(V_{u})dZ_{u}-\int_{T}^{s}4\sigma^{2}(V_{u})du}e^{\int_{T}^{s}4\sigma^{2}(V_{u})du}\Big|\mathcal{F}_{T}\right]
\nonumber
\\
&\leq
\left(\mathbb{E}\left[e^{\int_{T}^{s}4\sigma(V_{u})dZ_{u}-\frac{1}{2}\int_{T}^{s}(4\sigma)^{2}(V_{u})du}\Big|\mathcal{F}_{T}\right]\right)^{1/2}
\left(\mathbb{E}\left[e^{8\int_{T}^{s}\sigma^{2}(V_{u})du}\Big|\mathcal{F}_{T}\right]\right)^{1/2}
\nonumber
\\
&=\left(\mathbb{E}\left[e^{8\int_{T}^{s}\sigma^{2}(V_{u})du}\Big|\mathcal{F}_{T}\right]\right)^{1/2}
\leq e^{4(s-T)M_{\sigma}^{2}}.
\end{align}
Hence, 
\begin{equation}
\mathbb{E}[V_{s}^{2}|\mathcal{F}_{T}]
\leq
V_{T}^{2}e^{2(s-T)M_{\mu}}
e^{4(s-T)M_{\sigma}^{2}}.
\end{equation}
Therefore, for any $T\leq s\leq T+\tau$,
\begin{align}
\left[\mathbb{E}|V_{s}-V_{T}|^{2}|\mathcal{F}_{T}\right]
&=\mathbb{E}\left[V_{s}^{2}|\mathcal{F}_{T}\right]
+V_{T}^{2}
-2V_{T}\mathbb{E}[V_{s}|\mathcal{F}_{T}]
\nonumber
\\
&\leq
V_{T}^{2}\left(e^{2(s-T)M_{\mu}}
e^{4(s-T)M_{\sigma}^{2}}+1-2e^{-(s-T)M_{\mu}-\frac{1}{2}(s-T)M_{\sigma}^{2}}\right)
\nonumber
\\
&\leq
V_{T}^{2}\left(e^{2\tau M_{\mu}}
e^{4\tau M_{\sigma}^{2}}+1-2e^{-\tau M_{\mu}-\frac{1}{2}\tau M_{\sigma}^{2}}\right).
\end{align}
Hence, we conclude that, for any $T\leq s\leq T+\tau$,
\begin{align}
\left|\mathbb{E}\left[V_{s}\eta^{2}(S_{s})|\mathcal{F}_{T}\right]
-\mathbb{E}\left[V_{T}\eta^{2}(S_{s})|\mathcal{F}_{T}\right]\right|
\leq
M_{\eta}^{2}V_{T}\left(e^{2\tau M_{\mu}}
e^{4\tau M_{\sigma}^{2}}+1-2e^{-\tau M_{\mu}-\frac{1}{2}\tau M_{\sigma}^{2}}\right)^{1/2}.
\end{align}

\textbf{Step 2.}
\textit{The second term in (\ref{18})}.
The second term in (\ref{18}) is bounded further as
\begin{align}
\left|\mathbb{E}\left[V_{T}\eta^{2}(S_{s})|\mathcal{F}_{T}\right]
-V_{T}\eta^{2}(S_{T})\right|
=V_{T}
\left|\mathbb{E}\left[\eta^{2}(S_{s})|\mathcal{F}_{T}\right]
-\eta^{2}(S_{T})\right|.
\end{align}
By It\^{o}'s formula, 
\begin{equation}
d\eta^{2}(S_{t})
=(\eta^{2})'(S_{t})(r-q)S_{t}dt
+\frac{1}{2}(\eta^{2})''(S_{t})\eta^{2}(S_{t})S_{t}^{2}V_{t}dt
+(\eta^{2})'(S_{t})\eta(S_{t})S_{t}\sqrt{V_{t}}dW_{t}.
\end{equation}
Therefore, 
\begin{align}
\mathbb{E}\left[\eta^{2}(S_{s})|\mathcal{F}_{T}\right]
-\eta^{2}(S_{T})
=\int_{s}^{T}\mathbb{E}\left[(\eta^{2})'(S_{t})(r-q)S_{t}
+\frac{1}{2}(\eta^{2})''(S_{t})\eta^{2}(S_{t})S_{t}^{2}V_{t}\bigg|\mathcal{F}_{T}\right]dt.
\end{align}

By our assumption, $\eta$ is $L$-Lipschitz, 
so that $(\eta^{2})' = 2\eta\eta'$, which implies that $\eta^{2}$ is $2LM_{\eta}$-Lipschitz.
Also we have the bound (\ref{eta2p.assumption}) on the second derivative $(\eta^2)''(s)$.
Therefore, for any $T\leq s\leq T+\tau$, we have
\begin{align}
&\left|\mathbb{E}\left[\eta^{2}(S_{s})|\mathcal{F}_{T}\right]
-\eta^{2}(S_{T})\right|
\nonumber
\\
&\leq
\int_{T}^{s}\mathbb{E}\left[2LM_{\eta}|r-q|S_{t}
+\frac{1}{2}M_{\eta,2}M_{\eta}^{2}V_{t}\bigg|\mathcal{F}_{T}\right]dt
\nonumber
\\
&
=\int_{T}^{s}\left[2LM_{\eta}|r-q|S_{T}e^{(r-q)(t-T)}
+\frac{1}{2}M_{\eta,2}M_{\eta}^{2}\mathbb{E}[V_{t}|\mathcal{F}_{T}]\right]dt
\nonumber
\\
&\leq\int_{T}^{s}\left[2LM_{\eta}|r-q|S_{T}e^{(r-q)(t-T)}
+\frac{1}{2}M_{\eta,2}M_{\eta}^{2}V_{T}e^{\tau(M_{\mu}+M_{\sigma}^{2})}\right]dt
\nonumber
\\
&\leq
\tau
\left[2LM_{\eta}|r-q|S_{T}e^{|r-q|\tau}
+\frac{1}{2}M_{\eta,2}M_{\eta}^{2}V_{T}e^{\tau(M_{\mu}+M_{\sigma}^{2})}\right].
\end{align}

Hence, we conclude that for any $T\leq s\leq T+\tau$,
\begin{align}
&\left|\mathbb{E}\left[V_{s}\eta^{2}(S_{s})|\mathcal{F}_{T}\right]
-V_{T}\eta^{2}(S_{T})\right|
\nonumber
\\
&\leq
M_{\eta}^{2}V_{T}\left(e^{2\tau M_{\mu}}
e^{4\tau M_{\sigma}^{2}}+1-2e^{-\tau M_{\mu}-\frac{1}{2}\tau M_{\sigma}^{2}}\right)^{1/2}
\nonumber
\\
&\qquad\qquad
+\tau\left[2LM_{\eta}|r-q|S_{T}e^{|r-q|\tau}
+\frac{1}{2}M_{\eta,2}M_{\eta}^{2}V_{T}e^{\tau(M_{\mu}+M_{\sigma}^{2})}\right].
\end{align}
By recalling the formula in \eqref{VIX:formula}, we conclude that
\begin{equation}
\left|\mathrm{VIX}_T^2-V_{T}\eta^{2}(S_{T})\right|
\leq C_{1}(\tau)S_{T}+C_{2}(\tau)V_{T},
\end{equation}
where $C_{1}(\tau),C_{2}(\tau)$ are defined in \eqref{C:1:tau}-\eqref{C:2:tau}. 
Finally, we can compute that
\begin{equation}
\mathbb{E}\left|\mathrm{VIX}_T^2-V_{T}\eta^{2}(S_{T})\right|
\leq C_{1}(\tau)\mathbb{E}[S_{T}]+C_{2}(\tau)\mathbb{E}[V_{T}]
\leq
C_{1}(\tau)S_{0}e^{(r-q)T}
+C_{2}(\tau)V_{0}e^{T(M_{\mu}+M_{\sigma}^{2})}.
\end{equation}
This completes the proof.
\end{proof}

\begin{proof}[Proof of Corollary~\ref{cor:VIXsmalltau}]
One can compute that
\begin{align}
\left|\mathrm{VIX}_{T}-\sqrt{V_{T}}\eta(S_{T})
\right|
=\frac{\left|\mathrm{VIX}_{T}^{2}-V_{T}\eta^{2}(S_{T})
\right|}{\mathrm{VIX}_{T}+\sqrt{V_{T}}\eta(S_{T})}
\leq
\frac{\left|\mathrm{VIX}_{T}^{2}-V_{T}\eta^{2}(S_{T})
\right|}{\sqrt{V_{T}}m_{\eta}}.
\end{align}
Therefore, it follows from Proposition~\ref{prop:VIXsmalltau} that
\begin{align}
\left|\mathrm{VIX}_{T}-\sqrt{V_{T}}\eta(S_{T})
\right|
\leq\frac{C_{1}(\tau)}{m_{\eta}}\frac{S_{T}}{\sqrt{V_{T}}}+\frac{C_{2}(\tau)}{m_{\eta}}\sqrt{V_{T}}.
\end{align}
Therefore, we have
\begin{align}
\mathbb{E}\left|\mathrm{VIX}_{T}-\sqrt{V_{T}}\eta(S_{T})
\right|
\leq\frac{C_{1}(\tau)}{m_{\eta}}\left(\mathbb{E}[S_{T}^{2}]\right)^{1/2}
\left(\mathbb{E}[V_{T}^{-1}]\right)^{1/2}+\frac{C_{2}(\tau)}{m_{\eta}}\left(\mathbb{E}[V_{T}]\right)^{1/2}.
\end{align}
Note that under our assumption, $\mathbb{E}[S_{T}^{2}]=O(1)$ as $T\rightarrow 0$.
Moreover, we have shown that $\mathbb{E}[V_{T}]\leq V_{0}e^{T(M_{\mu}+M_{\sigma}^{2})}$ (see the proof of Proposition~\ref{prop:V:moments})
and similarly, 
\begin{align*}
\mathbb{E}[V_{T}^{-1}]
&=V_{0}^{-1}
\mathbb{E}\left[e^{\int_{0}^{T}(-\mu(V_{u})+\frac{1}{2}\sigma^{2}(V_{u}))du-\int_{0}^{T}\sigma(V_{u})dZ_{u}}\right]
\\
&\leq
V_{0}^{-1}e^{M_{\mu}T+M_{\sigma}^{2}T}
\mathbb{E}\left[e^{\int_{0}^{T}(-\frac{1}{2}\sigma^{2}(V_{u}))du-\int_{0}^{T}\sigma(V_{u})dZ_{u}}\right]
=V_{0}^{-1}e^{M_{\mu}T+M_{\sigma}^{2}T}.
\end{align*}
This completes the proof.
\end{proof}


\begin{proof}[Proof of Proposition~\ref{prop:StochVolVIX}]
(i) In this case,
\begin{equation}
dV_{t}=\mu V_{t}dt+\sigma(V_{t})V_{t}dZ_{t}\,,
\end{equation}
so that we can easily compute that for any $s\geq T$,
\begin{equation}
\mathbb{E}[V_{s}|\mathcal{F}_{T}]
=e^{\mu(s-T)}V_{T}\,,
\end{equation}
which gives the result quoted.

(ii) For this case we have 
\begin{equation}
dV_{t}=a(b-V_{t})dt+\sigma(V_{t})V_{t}dZ_{t}\,.
\end{equation}
In this case, one can compute that for any $s\geq T$,
\begin{equation}
\mathbb{E}[V_{s}|\mathcal{F}_{T}]
=V_{T}e^{-a(s-T)}+b\left(1-e^{-a(s-T)}\right)\,,
\end{equation}
which yields the stated result.
\end{proof}


\begin{proof}[Proof of Theorem~\ref{Thm:VIX}]

(i) OTM VIX call option ($K^2 > V_0\eta^2(S_0)$).

First, by \eqref{VIX:formula}
and Jensen's inequality, 
we can compute that for any $p\geq 2$:
\begin{align}
\mathbb{E}\left[\left(\mathrm{VIX}_T\right)^{p}\right]
&=\mathbb{E}\left[\left(\mathrm{VIX}_{T}^{2}\right)^{\frac{p}{2}}\right]
\nonumber
\\
&=\mathbb{E}\left[\left(\frac{1}{\tau} \int_T^{T+\tau} \mathbb{E}\left[\eta^{2}(S_{t})V_t|\mathcal{F}_T\right] dt \right)^{\frac{p}{2}}\right]
\nonumber
\\
&\leq\mathbb{E}\left[\frac{1}{\tau} \int_T^{T+\tau} \left(\mathbb{E}\left[\eta^{2}(S_{t})V_t|\mathcal{F}_T\right]\right)^{\frac{p}{2}} dt \right]
\nonumber
\\
&\leq\mathbb{E}\left[\frac{1}{\tau} \int_T^{T+\tau} \mathbb{E}\left[\left(\eta^{2}(S_{t})V_t\right)^{\frac{p}{2}}\big|\mathcal{F}_T\right] dt \right]
\nonumber
\\
&=\frac{1}{\tau} \int_T^{T+\tau} \mathbb{E}\left[|\eta|^{p}(S_{t})(V_t)^{\frac{p}{2}}\right] dt,
\end{align}
where under our assumption $\sup_{s\geq 0}|\eta(s)|\leq M_{\eta}$. Moreover, we can compute that
\begin{align}
\mathbb{E}\left[V_{t}^{p/2}\right]
&=V_{0}^{p/2}\mathbb{E}\left[e^{\int_{0}^{t}(\frac{p}{2}\mu(V_{u})-\frac{p}{4}\sigma^{2}(V_{u}))du+\frac{p}{2}\int_{0}^{t}\sigma(V_{u})dZ_{u}}\right]
\nonumber
\\
&\leq
V_{0}^{p/2}e^{\frac{p}{2}M_{\mu}t+\frac{p^{2}}{8}M_{\sigma}^{2}t}
\mathbb{E}\left[e^{-\int_{0}^{t}\frac{p^{2}}{8}\sigma^{2}(V_{u})du+\frac{p}{2}\int_{0}^{t}\sigma(V_{u})dZ_{u}}\right],
\end{align}
which implies that
\begin{equation}\label{VIX:moment:condition}
\mathbb{E}\left[\left(\mathrm{VIX}_T\right)^{p}\right]
\leq M_{\eta}^{p}V_{0}^{p/2}e^{\frac{p}{2}M_{\mu}T+\frac{p^{2}}{8}M_{\sigma}^{2}T}.
\end{equation}
Therefore, under the moment condition \eqref{VIX:moment:condition}, 
by a standard argument for short-maturity options (see e.g. \cite{Pham2007}), 
one can show that
\begin{equation}
\lim_{T\rightarrow 0}T\log\mathbb{E}\left[\left(\mathrm{VIX}_{T}-K\right)^{+}\right]
=\lim_{T\rightarrow 0}T\log\mathbb{Q}(\mathrm{VIX}_{T}\geq K).
\end{equation}
For any $\delta$, by Corollary~\ref{cor:VIXsmalltau}
\begin{align}\label{equiv:upper:bound}
\lim_{T\rightarrow 0}T\log\mathbb{Q}\left(\left|\mathrm{VIX}_{T}-\sqrt{V_{T}}\eta(S_{T})\right|\geq\delta\right)
\leq
\lim_{T\rightarrow 0}T\log\mathbb{Q}\left(\frac{C_{1}(\tau)}{m_{\eta}}\frac{S_{T}}{\sqrt{V_{T}}}+\frac{C_{2}(\tau)}{m_{\eta}}\sqrt{V_{T}}\geq\delta\right),
\end{align}
where $C_{1}(\tau),C_{2}(\tau)$ are given in \eqref{C:1:tau}-\eqref{C:2:tau}.

Since under Assumptions~\ref{assump:bounded} and \ref{assump:LDP}, $\mathbb{Q}((\log S_{T},\log V_{T})\in\cdot)$ satisfies a large deviation principle,
by the contraction principle (see e.g. Theorem 4.2.1. in \cite{Dembo1998}, restated in Theorem~\ref{Contraction:Thm}),
$\mathbb{Q}\left(\frac{C_{1}(\tau)}{m_{\eta}}\frac{S_{T}}{\sqrt{V_{T}}}+\frac{C_{2}(\tau)}{m_{\eta}}\sqrt{V_{T}}\in\cdot\right)$ also satisfies a large deviation principle
for any given $\tau>0$.
Since $\tau\rightarrow 0$ as $T\rightarrow 0$,
and $C_{1}(\tau),C_{2}(\tau)\rightarrow 0$ as $\tau\rightarrow 0$,
by \eqref{equiv:upper:bound}, we obtain the following superexponential estimate:
\begin{equation}\label{super:exp:estimate}
\lim_{T\rightarrow 0}T\log\mathbb{Q}\left(\left|\mathrm{VIX}_{T}-\sqrt{V_{T}}\eta(S_{T})\right|\geq\delta\right)
=-\infty.
\end{equation}
The above estimate \eqref{super:exp:estimate} is also known as the exponential equivalence in large deviations theory (see e.g. \cite{Dembo1998}),
which implies that
\begin{equation}
\lim_{T\rightarrow 0}T\log\mathbb{Q}(\mathrm{VIX}_{T}\geq K)
=\lim_{T\rightarrow 0}T\log\mathbb{Q}\left(\sqrt{V_{T}}\eta(S_{T})\geq K\right)
=\lim_{T\rightarrow 0}T\log\mathbb{Q}\left(V_{T}\eta^{2}(S_{T})\geq K^{2}\right).
\end{equation}

Under Assumptions~\ref{assump:bounded} and \ref{assump:LDP}, by the sample-path large deviations for 
small time diffusions (see for example \cite{Varadhan} and \cite{Robertson2010})
and an application of the contraction principle (see for example Theorem 4.2.1. in \cite{Dembo1998}, restated in Theorem~\ref{Contraction:Thm}), similar to the proof of Theorem~\ref{Thm:E}, we have
\begin{align}
&\lim_{T\rightarrow 0}T\log\mathbb{Q}\left(V_{T}\eta^{2}(S_{T})\geq K^{2}\right)
\nonumber
\\
&=-\inf_{\substack{g(0)=\log S_{0},h(0)=\log V_{0}\\
e^{h(1)}\eta^{2}(e^{g(1)})=K^{2}}}
\Bigg\{\frac{1}{2(1-\rho^{2})}\int_{0}^{1}\left(\frac{g'(t)}{\eta(e^{g(t)})\sqrt{e^{h(t)}}}-\frac{\rho h'(t)}{\sigma(e^{h(t)})}\right)^{2}dt
\nonumber
\\
&\qquad\qquad\qquad\qquad\qquad\qquad\qquad\qquad
+\frac{1}{2}\int_{0}^{1}\left(\frac{h'(t)}{\sigma(e^{h(t)})}\right)^{2}dt\Bigg\}.
\end{align}

Given $h$, we can determine the optimal $g$ as follows.
By Cauchy-Schwarz inequality, 
\begin{align}
&\int_{0}^{1}\left(\frac{g'(t)}{\eta(e^{g(t)})\sqrt{e^{h(t)}}}-\frac{\rho h'(t)}{\sigma(e^{h(t)})}\right)^{2}dt
\cdot\int_{0}^{1}\left(\sqrt{e^{h(t)}}\right)^{2}dt
\nonumber
\\
&\geq
\left(\int_{0}^{1}\left(\frac{g'(t)}{\eta(e^{g(t)})}-\frac{\rho h'(t)\sqrt{e^{h(t)}}}{\sigma(e^{h(t)})}\right)dt\right)^{2},
\end{align}
where
\begin{align}
\label{Thm:V:call:1}
\int_{0}^{1}\frac{g'(t)}{\eta(e^{g(t)})}dt
=\int_{0}^{1}\frac{d(e^{g(t)})}{e^{g(t)}\eta(e^{g(t)})}
=\int_{e^{g(0)}}^{e^{g(1)}}\frac{dx}{x\eta(x)},
\end{align}
where $e^{g(0)}=S_{0}$ and $e^{g(1)}=(\eta^{2})^{-1}(K^{2}e^{-h(1)})$, 
and
\begin{equation}
\label{Thm:V:call:2}
\int_{0}^{1}\frac{\rho h'(t)\sqrt{e^{h(t)}}}{\sigma(e^{h(t)})}dt
=\int_{0}^{1}\frac{\rho h'(t)e^{h(t)}\sqrt{e^{h(t)}}}{e^{h(t)}\sigma(e^{h(t)})}dt
=\int_{e^{h(0)}}^{e^{h(1)}}\frac{\rho dx}{\sqrt{x}\sigma(x)},
\end{equation}
where $h(0)=\log V_{0}$.
Therefore, we have
\begin{align}
&\frac{1}{2}\int_{0}^{1}\left(\frac{g'(t)}{\eta(e^{g(t)})\sqrt{e^{h(t)}}}-\frac{\rho h'(t)}{\sigma(e^{h(t)})}\right)^{2}dt
\nonumber
\\
&\geq
\frac{1}{2}\left(\int_{S_{0}}^{(\eta^{2})^{-1}(K^{2}e^{-h(1)})}\frac{dx}{x\eta(x)}-\int_{V_{0}}^{e^{h(1)}}\frac{\rho dx}{\sqrt{x}\sigma(x)}\right)^{2}
\left(\int_{0}^{1}e^{h(t)}dt\right)^{-1},
\end{align}
and by Cauchy-Schwarz inequality, the equality is achieved
when 
\begin{equation}
\frac{g'(t)}{\eta(e^{g(t)})}-\frac{\rho\sqrt{e^{h(t)}}h'(t)}{\sigma(e^{h(t)})}=C_{1}e^{h(t)},
\end{equation}
for some constant $C_{1}$ so that $g(t)$ can be solved
via the equation:
\begin{equation}
\int_{S_{0}}^{e^{g(t)}}\frac{dx}{x\eta(x)}-\int_{V_{0}}^{e^{h(t)}}\frac{\rho dx}{\sqrt{x}\sigma(x)}=C_{1}\int_{0}^{t}e^{h(s)}ds,
\end{equation}
where
\begin{equation}
C_{1}=\frac{\int_{S_{0}}^{(\eta^{2})^{-1}(K^{2}e^{-h(1)})}\frac{dx}{x\eta(x)}-\int_{V_{0}}^{e^{h(1)}}\frac{\rho dx}{\sqrt{x}\sigma(x)}}{\int_{0}^{1}e^{h(s)}ds}.
\end{equation}
Since with fixed $h$, we can solve for the optimal $g$,
by the discussions above, we conclude that
\begin{align}
&\lim_{T\rightarrow 0}T\log\mathbb{Q}\left(V_{T}\eta^{2}(S_{T})\geq K^{2}\right)
\nonumber
\\
&=-\inf_{h(0)=\log V_{0}}
\Bigg\{\frac{1}{2(1-\rho^{2})}\left(\int_{S_{0}}^{(\eta^{2})^{-1}(K^{2}e^{-h(1)})}\frac{dx}{x\eta(x)}-\int_{V_{0}}^{e^{h(1)}}\frac{\rho dx}{\sqrt{x}\sigma(x)}\right)^{2}
\left(\int_{0}^{1}e^{h(t)}dt\right)^{-1}
\nonumber
\\
&\qquad\qquad\qquad\qquad\qquad
+\frac{1}{2}\int_{0}^{1}\left(\frac{h'(t)}{\sigma(e^{h(t)})}\right)^{2}dt\Bigg\}
\nonumber
\\
&=-\inf_{y,z}
\Bigg\{\frac{1}{2(1-\rho^{2})z}\left(\int_{S_{0}}^{(\eta^{2})^{-1}(K^{2}e^{-y})}\frac{dx}{x\eta(x)}-\int_{V_{0}}^{e^{y}}\frac{\rho dx}{\sqrt{x}\sigma(x)}\right)^{2}
\nonumber
\\
&\qquad\qquad\qquad\qquad\qquad\qquad\qquad
+\inf_{\substack{h(0)=\log V_{0},h(1)=y\\
\int_{0}^{1}e^{h(t)}dt=z}}\frac{1}{2}\int_{0}^{1}\left(\frac{h'(t)}{\sigma(e^{h(t)})}\right)^{2}dt\Bigg\}
\nonumber
\\
&=-\inf_{y,z}
\left\{\frac{1}{2(1-\rho^{2})z}\left(\int_{S_{0}}^{(\eta^{2})^{-1}(K^{2}e^{-y})}\frac{dx}{x\eta(x)}-\int_{V_{0}}^{e^{y}}\frac{\rho dx}{\sqrt{x}\sigma(x)}\right)^{2}
+H(y,z)\right\},
\end{align}
where
\begin{equation}
H(y,z):=\inf_{\substack{h(0)=\log V_{0},h(1)=y\\
\int_{0}^{1}e^{h(t)}dt=z}}\frac{1}{2}\int_{0}^{1}\left(\frac{h'(t)}{\sigma(e^{h(t)})}\right)^{2}dt.
\end{equation}

(ii) OTM VIX put option ($K^2 < V_0 \eta^2(S_0)$).
Since $(K-\mathrm{VIX}_{T})^{+}\leq K$ with probability one, 
similar to the OTM VIX call option case, 
we can show that
\begin{align}
\lim_{T\rightarrow 0}T\log\mathbb{E}\left[(K-\mathrm{VIX}_{T})^{+}\right]
&=\lim_{T\rightarrow 0}T\log\mathbb{Q}(K\geq\mathrm{VIX}_{T})
\nonumber
\\
&=\lim_{T\rightarrow 0}\mathbb{Q}\left(K\geq\sqrt{V_{T}}\eta(S_{T})\right)
=T \log \mathbb{Q}\left(K^{2} \geq V_{T}\eta^{2}(S_{T})\right).
\end{align}
Similar to the proof for the rate function for OTM European put option in Theorem~\ref{Thm:E}, we can also get a similar result for the rate function for OTM VIX put option. Similar to~\eqref{put:variation:1} and~\eqref{put:variation:2}, for put option, we can change the variables in~\eqref{Thm:V:call:1} and~\eqref{Thm:V:call:2} and obtain:
\begin{equation}
\int_{0}^{1}\frac{g'(t)}{\eta(e^{g(t)})}dt = \int_{e^{-g(1)}}^{e^{-g(0)}}\frac{dx}{x\eta(x^{-1})},
\end{equation}
where $e^{g(0)}=S_{0}$ and $e^{g(1)}=(\eta^{2})^{-1}(K^{2}e^{-h(1)})$, 
\begin{equation}
\int_{0}^{1}\frac{\rho h'(t)\sqrt{e^{h(t)}}}{\sigma(e^{h(t)})}dt
=\int_{e^{-h(1)}}^{e^{-h(0)}}\frac{\rho dx}{\sqrt{x^3}\sigma(x^{-1})},
\end{equation}
where $h(0)=\log V_{0}$. We can follow the step for~\eqref{Thm:E:put} to compute that
\begin{align}
& \lim_{T\to 0}
T \log \mathbb{Q}\left(K^{2} \geq V_{T}\eta^{2}(S_{T}))\right)\nonumber \\
&=-\inf_{y,z}
\left\{\frac{1}{2(1-\rho^{2})z}\left(\int_{\frac{1}{(\eta^{2})^{-1}(K^{2}e^{-y})}}^{S_0^{-1}}\frac{dx}{x\eta(x^{-1})}-\int_{e^{-y}}^{V_0^{-1}}\frac{\rho dx}{\sqrt{x^3}\sigma(x^{-1})}\right)^{2}
\right.\nonumber \\
& \left.\qquad\qquad\qquad\qquad\qquad\qquad\qquad\qquad\qquad\qquad +\inf_{\substack{h(0)=\log V_{0},h(1)=y\\
\int_{0}^{1}e^{h(t)}dt=z}}\frac{1}{2}\int_{0}^{1}\left(\frac{h'(t)}{\sigma(e^{h(t)})}\right)^{2}dt\right\}
\nonumber
\\
&=-\inf_{y,z}
\left\{\frac{1}{2(1-\rho^{2})z}\left(\int_{\frac{1}{(\eta^{2})^{-1}(K^{2}e^{-y})}}^{S_0^{-1}}\frac{dx}{x\eta(x^{-1})}-\int_{V_0^{-1}}^{e^{-y}}\frac{\rho dx}{\sqrt{x^3}\sigma(x^{-1})}\right)^{2}
+H(y,z)\right\},\label{change:variable:in:1}
\end{align}
where $(\eta^2)^{-1}(\cdot)$ denotes the inverse function of $\eta^{2}$. 
Finally, by changing the variables $x\mapsto x^{-1}$ in \eqref{change:variable:in:1}, we have
\begin{align}
\left(\int_{\frac{1}{(\eta^{2})^{-1}(K^{2}e^{-y})}}^{S_0^{-1}}\frac{dx}{x\eta(x^{-1})}-\int_{V_0^{-1}}^{e^{-y}}\frac{\rho dx}{\sqrt{x^3}\sigma(x^{-1})}\right)^{2}
&=\left(-\int_{(\eta^{2})^{-1}(K^{2}e^{-y})}^{S_0}\frac{dx}{x\eta(x)}+\int_{V_0}^{e^{y}}\frac{\rho dx}{\sqrt{x}\sigma(x)}\right)^{2}
\nonumber
\\
&=\left(\int_{(\eta^{2})^{-1}(K^{2}e^{-y})}^{S_0}\frac{dx}{x\eta(x)}-\int_{V_0}^{e^{y}}\frac{\rho dx}{\sqrt{x}\sigma(x)}\right)^{2}.
\end{align}
This completes the proof.
\end{proof}

\begin{proof}[Proof of Proposition~\ref{prop:VIX}]
Recall that the VIX option prices under the time homogeneous stochastic volatility model are given in (\ref{CVtime-h}).
Consider the case of the OTM VIX call option.
Proceeding as in the proof of Theorem \ref{Thm:E} for OTM European options, we get, by upper and lower bounds, the relation
\begin{equation}
\lim_{T\to 0} T \log C_V(K,T) = \lim_{T\to 0} T \log \mathbb{Q}(\sqrt{\mathcal{F}(V_T)} > K) = 
\lim_{T\to 0} T \log \mathbb{Q}(V_T > \mathcal{F}^{-1}(K^2))\,.
\end{equation}
Thus, the problem has been reduced to the short maturity asymptotics for OTM European options in the local volatility model for $V_t$, 
see for example \cite{BBF,Pham2007},
which is evaluated with the stated result. 

The OTM VIX put option can be handled in a similar way. This completes the proof. 
\end{proof}
\begin{proof}[Proof of Theorem~\ref{thm:VIX:ATM}]

We only provide the proof for the ATM VIX call option with $K=\sqrt{V_{0}}\eta(S_{0})$. 
The case for the ATM VIX put option can be handled
similarly.


\textbf{Step 1.}
First, by using the estimates in Corollary~\ref{cor:VIXsmalltau},
we can easily show that 
\begin{equation}\label{108}
\left|C_{V}(K,T)-\mathbb{E}\left[\left(\sqrt{V_{T}}\eta(S_{T})-K\right)^{+}\right]\right|
=O\left(T^{\frac12 + \delta}\right),
\end{equation}
for some $\delta > 0$, 
as $T\rightarrow 0$. 
Indeed, since $x\mapsto x^{+}$ is $1$-Lipschitz, we get
\begin{equation}
\left|C_{V}(K,T)-\mathbb{E}\left[\left(\sqrt{V_{T}}\eta(S_{T})-K\right)^{+}\right]\right| \leq 
\mathbb{E}\left| \mathrm{VIX}_T - \sqrt{V_{T}}\eta(S_{T})\right|
\leq O(\sqrt{\tau}) \,,
\end{equation}
where the last inequality follows from Corollary~\ref{cor:VIXsmalltau}.
Then \eqref{108} follows from the assumption $\tau = O(T^{1+\epsilon})$.

\textbf{Step 2.}
Next, we define
\begin{align}
&\hat{S}_{T}=S_{0}+\eta(S_{0})S_{0}\sqrt{V_{0}}\left(\sqrt{1-\rho^{2}}B_{T}+\rho Z_{T}\right)\,,
\\
&\hat{V}_{T}=V_{0}+\sigma(V_{0})V_{0}Z_{T}\,,
\end{align}
where $B_{t},Z_{t}$ are independent standard Brownian motions.

In the proof of Theorem~\ref{thm:European:ATM} for European options, we showed
that
\begin{equation}
\mathbb{E}\left[\left|S_{T}-\hat{S}_{T}\right|^{2}\right]
=O(T^{3/2}),
\end{equation}
as $T\rightarrow 0$.

Next, let us recall that
\begin{equation}
V_{t}=V_{0}+\int_{0}^{t}\mu(V_{s})V_{s}ds
+\int_{0}^{t}\sigma(V_{s})V_{s}dZ_{s},
\end{equation}
and we can also re-write $\hat{V}_{t}$ as
\begin{equation}
\hat{V}_{t}=V_{0}+\int_{0}^{t}\sigma(V_{0})V_{0}dZ_{s}.
\end{equation}

Under our assumptions, $\sigma$ is $L'$-Lipschitz and $M_{\sigma}$-uniformly bounded. 
Therefore, we can compute that for any $t$:
$$
\vert\sigma(V_t)V_t - \sigma(V_0)V_0\vert = \vert\sigma(V_t)(V_t - V_0) + (\sigma(V_t) - \sigma(V_0))V_0\vert \leq L_{\sigma}\vert V_t - V_0\vert,
$$
where $L_{\sigma} := M_{\sigma}+V_0L'$.
Then we can compute that
\begin{align}
\mathbb{E}\left[|V_{t}-\hat{V}_{t}|^{2}\right]
&\leq
2\mathbb{E}\left[\left(\int_{0}^{t}\mu(V_{s})V_{s}ds\right)^{2}\right]
+2\mathbb{E}\left[\left(\int_{0}^{t}(\sigma(V_{s})V_{s}-\sigma(V_{0})V_{0})dZ_{s}\right)^{2}\right]
\nonumber
\\
&\leq
2t\int_{0}^{t}\mathbb{E}\left[\left(\mu(V_{s})V_{s}\right)^{2}\right]ds
+2\int_{0}^{t}\mathbb{E}\left[(\sigma(V_{s})V_{s}-\sigma(V_{0})V_{0})^{2}\right]ds
\nonumber
\\
&\leq
2M_{\mu}^{2}t\int_{0}^{t}\mathbb{E}\left[V_{s}^{2}\right]ds
+2L_{\sigma}^{2}\int_{0}^{t}\mathbb{E}\left[(V_{s}-V_{0})^{2}\right]ds
\end{align}
where we used Cauchy-Schwarz inequality as well as It\^{o}'s isometry.
Moreover, we can compute that
\begin{align}
2L_{\sigma}^{2}\int_{0}^{t}\mathbb{E}\left[(V_{s}-V_{0})^{2}\right]ds
&\leq
4L_{\sigma}^{2}\int_{0}^{t}\mathbb{E}\left[(V_{s}-\hat{V}_{s})^{2}\right]ds
+4L_{\sigma}^{2}\int_{0}^{t}\mathbb{E}\left[(\hat{V}_{s}-V_{0})^{2}\right]ds
\nonumber
\\
&=4L_{\sigma}^{2}\int_{0}^{t}\mathbb{E}\left[(V_{s}-\hat{V}_{s})^{2}\right]ds
+4L_{\sigma}^{2}\int_{0}^{t}\sigma^{2}(V_{0})V_{0}^{2}sds
\nonumber
\\
&=4L_{\sigma}^{2}\int_{0}^{t}\mathbb{E}\left[(V_{s}-\hat{V}_{s})^{2}\right]ds
+2L_{\sigma}^{2}\sigma^{2}(V_{0})V_{0}^{2}t^{2}.
\end{align}
Hence, we conclude that
\begin{align}
\mathbb{E}\left[|V_{t}-\hat{V}_{t}|^{2}\right]
\leq
2M_{\mu}^{2}t^{2}C_{4}
+4L_{\sigma}^{2}\int_{0}^{t}\mathbb{E}\left[(V_{s}-\hat{V}_{s})^{2}\right]ds
+2L_{\sigma}^{2}\sigma^{2}(V_{0})V_{0}^{2}t^{2},
\end{align}
where $C_{4}=\max_{0\leq t\leq T}\mathbb{E}[V_{t}^{2}]=O(1)$ as $T\rightarrow 0$ under our assumptions.
By Gronwall's inequality, we conclude that
\begin{align}
\mathbb{E}\left[|V_{T}-\hat{V}_{T}|^{2}\right]
\leq O(T^{2}),
\end{align}
as $T\rightarrow 0$
and hence $\mathbb{E}|V_{T}-\hat{V}_{T}|
\leq O(T)$ as $T\rightarrow 0$.

\textbf{Step 3.}
Next, we can compute that
\begin{align}
&\left|\mathbb{E}\left[\left(\sqrt{\hat{V}_{T}}\eta(\hat{S}_{T})-K\right)^{+}1_{\hat{V}_{T}\geq\frac{V_{0}}{2}}\right]-\mathbb{E}\left[\left(\sqrt{V_{T}}\eta(S_{T})-K\right)^{+}\right]\right|
\nonumber
\\
&\leq
\left|\mathbb{E}\left[\left(\sqrt{\hat{V}_{T}}\eta(\hat{S}_{T})-K\right)^{+}1_{\hat{V}_{T}\geq\frac{V_{0}}{2}}\right]-\mathbb{E}\left[\left(\sqrt{V_{T}}\eta(\hat{S}_{T})-K\right)^{+}\right]\right|
\nonumber
\\
&\qquad\qquad
+\left|\mathbb{E}\left[\left(\sqrt{V_{T}}\eta(\hat{S}_{T})-K\right)^{+}\right]-\mathbb{E}\left[\left(\sqrt{V_{T}}\eta(S_{T})-K\right)^{+}\right]\right|.\label{two:terms:to:bound:3}
\end{align}

Note that the function $x\mapsto x^{+}$ is $1$-Lipschitz
and $\eta(S)$ is $L$-Lipschitz.

Therefore,
\begin{align}
&\left|\mathbb{E}\left[\left(\sqrt{V_{T}}\eta(\hat{S}_{T})-K\right)^{+}\right]-\mathbb{E}\left[\left(\sqrt{V_{T}}\eta(S_{T})-K\right)^{+}\right]\right|
\nonumber
\\
&\leq
\mathbb{E}\left|\sqrt{V_{T}}\left(\eta(\hat{S}_{T})-\eta(S_{T})\right)\right|
\nonumber
\\
&\leq
\left(\mathbb{E}[V_{T}]\right)^{1/2}
\left(\mathbb{E}\left[\left|\eta(\hat{S}_{T})-\eta(S_{T})\right|^{2}\right]\right)^{1/2}
\nonumber
\\
&\leq
\sqrt{C_{4}}L\left(\mathbb{E}\left|\hat{S}_{T}-S_{T}\right|^{2}\right)^{1/2}
=O\left(T^{3/4}\right),
\end{align}
as $T\rightarrow 0$. This upper bounds the second term in \eqref{two:terms:to:bound:3}.

Furthermore, the first term in \eqref{two:terms:to:bound:3} can be bounded as:
\begin{align}
&\left|\mathbb{E}\left[\left(\sqrt{\hat{V}_{T}}\eta(\hat{S}_{T})-K\right)^{+}1_{\hat{V}_{T}\geq\frac{V_{0}}{2}}\right]-\mathbb{E}\left[\left(\sqrt{V_{T}}\eta(\eta(\hat{S}_T))
-K\right)^{+}\right]\right|
\nonumber
\\
&\leq
\left|\mathbb{E}\left[\left(\sqrt{\hat{V}_{T}}\eta(\hat{S}_{T})-K\right)^{+}1_{\hat{V}_{T}\geq\frac{V_{0}}{2}}\right]-\mathbb{E}\left[\left(\sqrt{V_{T}}\eta(\hat{S}_{T})-K\right)^{+}1_{\hat{V}_{T}\geq\frac{V_{0}}{2}}\right]\right|
\nonumber
\\
&\qquad\qquad\qquad
+\mathbb{E}\left[\left(\sqrt{V_{T}}\eta(\hat{S}_{T})-K\right)^{+}1_{\hat{V}_{T}\leq\frac{V_{0}}{2}}\right].
\label{two:terms:to:bound:4}
\end{align}

\textbf{Step 4.}
To complete the upper bound on the first term in \eqref{two:terms:to:bound:3} in 
\textbf{Step 3}, we need to provide upper bounds
on the two terms in \eqref{two:terms:to:bound:4}.
In particular, we will use large deviations theory
to bound the second term in \eqref{two:terms:to:bound:4}
since $\left\{\hat{V}_{T}\leq\frac{V_{0}}{2}\right\}$ is a rare event
and we will use the $\frac{1}{2\sqrt{\kappa}}$-Lipschitz-continuity of $v\mapsto\sqrt{v}$ for any $v\geq\kappa>0$
to bound the first term in \eqref{two:terms:to:bound:4}.

Let us first bound the second term in \eqref{two:terms:to:bound:4}.
By Cauchy-Schwarz inequality, we can compute that
\begin{align}
\mathbb{E}\left[\left(\sqrt{V_{T}}\eta(\hat{S}_{T})
-K\right)^{+}1_{\hat{V}_{T}\leq\frac{V_{0}}{2}}\right]
\leq
\left(\mathbb{E}\left[\left(\sqrt{V_{T}}\eta(\hat{S}_T)
-K\right)^{2}\right]\right)^{1/2}
\mathbb{Q}\left(\hat{V}_{T}\leq\frac{V_{0}}{2}\right).
\end{align}
Note that by \eqref{V:s:upper:bound}
\begin{equation}
\mathbb{E}\left[\left(\sqrt{V_{T}}\eta(\hat{S}_{T})-K\right)^{2}\right]
\leq
2\mathbb{E}\left[V_{T}\eta^{2}(\hat{S}_{T})\right]
+2K^{2}
\leq
2M_{\eta}^{2}\mathbb{E}[V_{T}]
+2K^{2}
\leq
2M_{\eta}^{2}V_{0}e^{M_{\mu}T}
+2K^{2},
\end{equation}
and $\mathbb{Q}(\hat{V}_{T}\leq\frac{V_{0}}{2})=e^{-O(\frac{1}{T})}$
by the large deviations theory under Assumptions~\ref{assump:bounded} and \ref{assump:LDP} (see the proof of Theorem~\ref{Thm:E}).

Next, let us bound the first term in \eqref{two:terms:to:bound:4}.
Since $x\mapsto x^{+}$ is $1$-Lipschitz, we have
\begin{align}
&\left|\mathbb{E}\left[\left(\sqrt{\hat{V}_{T}}\eta(\hat{S}_{T})-K\right)^{+}1_{\hat{V}_{T}\geq\frac{V_{0}}{2}}\right]-\mathbb{E}\left[\left(\sqrt{V_{T}}\eta(\hat{S}_{T})-K\right)^{+}1_{\hat{V}_{T}\geq\frac{V_{0}}{2}}\right]\right|
\nonumber
\\
&\leq
\mathbb{E}\left[\left|\sqrt{\hat{V}_{T}}\eta(\hat{S}_{T})-\sqrt{V_{T}}\eta(\hat{S}_{T})\right|1_{\hat{V}_{T}\geq\frac{V_{0}}{2}}\right]
\nonumber
\\
&\leq
M_{\eta}\mathbb{E}\left[\left|\sqrt{\hat{V}_{T}}-\sqrt{V_{T}}\right|1_{\hat{V}_{T}\geq\frac{V_{0}}{2}}\right]
\nonumber
\\
&=M_{\eta}\mathbb{E}\left[\left|\sqrt{\hat{V}_{T}}-\sqrt{V_{T}}\right|1_{\hat{V}_{T}\geq\frac{V_{0}}{2}}1_{V_{T}\geq\frac{V_{0}}{2}}\right]
+M_{\eta}\mathbb{E}\left[\left|\sqrt{\hat{V}_{T}}-\sqrt{V_{T}}\right|1_{\hat{V}_{T}\geq\frac{V_{0}}{2}}1_{V_{T}<\frac{V_{0}}{2}}\right].
\end{align}
By using the similar argument before by applying Cauchy-Schwarz inequality and large deviations theory,
one can show that
\begin{equation}
M_{\eta}\mathbb{E}\left[\left|\sqrt{\hat{V}_{T}}-\sqrt{V_{T}}\right|1_{\hat{V}_{T}\geq\frac{V_{0}}{2}}1_{V_{T}<\frac{V_{0}}{2}}\right]=e^{-O(\frac{1}{T})},
\end{equation}
as $T\rightarrow 0$.
Moreover, since $x\mapsto\sqrt{x}$ is $\frac{1}{2\sqrt{\kappa}}$-Lipschitz
for any $x\geq\kappa$, we have
\begin{align}
M_{\eta}\mathbb{E}\left[\left|\sqrt{\hat{V}_{T}}-\sqrt{V_{T}}\right|1_{\hat{V}_{T}\geq\frac{V_{0}}{2}}1_{V_{T}\geq\frac{V_{0}}{2}}\right]
\leq
M_{\eta}
\frac{1}{2\sqrt{V_{0}/2}}\mathbb{E}\left|\hat{V}_{T}-V_{T}\right|=O(T),
\end{align}
as $T\rightarrow 0$.

\textbf{Step 5.} 
By combining the estimates in \textbf{Step 2}, \textbf{Step 3} and \textbf{Step 4}, we showed
that $\mathbb{E}\left[\left(\sqrt{V_{T}\eta(S_{T})}-K\right)^{+}\right]$
can be approximated by 
\begin{equation}\label{to:be:further:computed}
\mathbb{E}\left[\left(\sqrt{\hat{V}_{T}}\eta(\hat{S}_{T})-\sqrt{V_{0}}\eta(S_{0})\right)^{+}1_{\hat{V}_{T}\geq\frac{V_{0}}{2}}\right].
\end{equation}
Next, we focus on the computation on \eqref{to:be:further:computed}.
We will show that the term in \eqref{to:be:further:computed} 
can be approximated by
\begin{equation}
\mathbb{E}\left[\left(\eta(S_{0})\frac{1}{2\sqrt{V_{0}}}\sigma(V_{0})V_{0}Z_{T}
+\sqrt{V_{0}}\eta'(S_{0})\eta(S_{0})S_{0}\sqrt{V_{0}}\left(\sqrt{1-\rho^{2}}B_{T}+\rho Z_{T}\right)\right)^{+}\right],
\end{equation}
and we will make this rigorous via a few steps.

\textbf{Step 5(a).}
First, we will show that 
\begin{equation*}
\mathbb{E}\left[\left(\sqrt{\hat{V}_{T}}\eta(\hat{S}_{T})-\sqrt{V_{0}}\eta(S_{0})\right)^{+}1_{\hat{V}_{T}\geq\frac{V_{0}}{2}}\right] 
\end{equation*}
can be approximated by
\begin{align}
&\mathbb{E}\Bigg[\Bigg(\sqrt{V_{0}+\sigma(V_{0})V_{0}Z_{T}}\cdot\left(\eta(S_{0})+\eta'(S_{0})\eta(S_{0})S_{0}\sqrt{V_{0}}\left(\sqrt{1-\rho^{2}}B_{T}+\rho Z_{T}\right)\right)
-\sqrt{V_{0}}\eta(S_{0})\Bigg)^{+}
\nonumber
\\
&\qquad\qquad\qquad\qquad\qquad\qquad\cdot
1_{V_{0}+\sigma(V_{0})V_{0}Z_{T}\geq\frac{V_{0}}{2}}\Bigg].\label{term:approximated:by:1}
\end{align}
First, we can compute that
\begin{align}
&\mathbb{E}\left[\left(\sqrt{\hat{V}_{T}}\eta(\hat{S}_{T})-\sqrt{V_{0}}\eta(S_{0})\right)^{+}1_{\hat{V}_{T}\geq\frac{V_{0}}{2}}\right]
\nonumber
\\
&=\mathbb{E}\Bigg[\left(\sqrt{V_{0}+\sigma(V_{0})V_{0}Z_{T}}\cdot\eta\left(S_{0}+\eta(S_{0})S_{0}\sqrt{V_{0}}\left(\sqrt{1-\rho^{2}}B_{T}+\rho Z_{T}\right)\right)-\sqrt{V_{0}}\eta(S_{0})\right)^{+}\nonumber
\\
&\qquad\qquad\qquad\qquad\qquad\qquad\qquad\qquad\cdot
1_{V_{0}+\sigma(V_{0})V_{0}Z_{T}\geq\frac{V_{0}}{2}}\Bigg].
\end{align}
We recall the assumption that $\sup_{x\in\mathbb{R}}|\eta''(x)|<\infty$,
Moreover, $x\mapsto x^{+}$ is $1$-Lipschitz.
Therefore, there exists some $C>0$, such that
\begin{align}
&\Bigg|\mathbb{E}\Bigg[\left(\sqrt{V_{0}+\sigma(V_{0})V_{0}Z_{T}}\cdot\eta\left(S_{0}+\eta(S_{0})S_{0}\sqrt{V_{0}}\left(\sqrt{1-\rho^{2}}B_{T}+\rho Z_{T}\right)\right)-\sqrt{V_{0}}\eta(S_{0})\right)^{+}
\nonumber
\\
&\qquad\qquad\qquad\qquad\qquad\qquad\qquad\cdot 
1_{V_{0}+\sigma(V_{0})V_{0}Z_{T}\geq\frac{V_{0}}{2}}\Bigg]
\nonumber
\\
&\qquad
-\mathbb{E}\Bigg[\Bigg(\sqrt{V_{0}+\sigma(V_{0})V_{0}Z_{T}}\cdot\left(\eta(S_{0})+\eta'(S_{0})\eta(S_{0})S_{0}\sqrt{V_{0}}\left(\sqrt{1-\rho^{2}}B_{T}+\rho Z_{T}\right)\right)
\nonumber
\\
&\qquad\qquad\qquad-\sqrt{V_{0}}\eta(S_{0})\Bigg)^{+}1_{V_{0}+\sigma(V_{0})V_{0}Z_{T}\geq\frac{V_{0}}{2}}\Bigg]
\Bigg|
\nonumber
\\
&\leq
C\mathbb{E}\left[\sqrt{V_{0}+\sigma(V_{0})V_{0}Z_{T}}\left(\sqrt{1-\rho^{2}}B_{T}+\rho Z_{T}\right)^{2}1_{V_{0}+\sigma(V_{0})V_{0}Z_{T}\geq\frac{V_{0}}{2}}\right]
\nonumber
\\
&\leq
C\left(\mathbb{E}\left[(V_{0}+\sigma(V_{0})V_{0}Z_{T})1_{V_{0}+\sigma(V_{0})V_{0}Z_{T}\geq\frac{V_{0}}{2}}\right]\right)^{1/2}
\left(\mathbb{E}\left[\left(\sqrt{1-\rho^{2}}B_{T}+\rho Z_{T}\right)^{4}\right]\right)^{1/2}
\nonumber
\\
&\leq
C\left(\mathbb{E}\left[(V_{0}+\sigma(V_{0})V_{0}Z_{T})^{2}\right]\right)^{1/4}
\left(3T^{2}\right)^{1/2}
\nonumber
\\
&=C\left(V_{0}^{2}+\sigma^{2}(V_{0})V_{0}^{2}T\right)^{1/4}
\left(3T^{2}\right)^{1/2}=O(T),
\end{align}
as $T\rightarrow 0$.

\textbf{Step 5(b).}
Next, let us show that the term in \eqref{term:approximated:by:1}
can be approximated by
\begin{align}
&\mathbb{E}\Bigg[\Bigg(\left(\sqrt{V_{0}}+\frac{1}{2\sqrt{V_{0}}}\sigma(V_{0})V_{0}Z_{T}\right)\cdot\left(\eta(S_{0})+\eta'(S_{0})\eta(S_{0})S_{0}\sqrt{V_{0}}\left(\sqrt{1-\rho^{2}}B_{T}+\rho Z_{T}\right)\right)
\nonumber
\\
&\qquad\qquad\qquad-\sqrt{V_{0}}\eta(S_{0})\Bigg)^{+}1_{V_{0}+\sigma(V_{0})V_{0}Z_{T}\geq\frac{V_{0}}{2}}\Bigg].  
\label{term:approximated:by:2}
\end{align}

We notice that 
\begin{equation}
|(\sqrt{V})''|=\left|\frac{1}{4V^{3/2}}\right|\leq\frac{1}{4(V_{0}/2)^{3/2}},
\end{equation}
uniformly in $V\geq\frac{V_{0}}{2}$.
Moreover, $x\mapsto x^{+}$ is $1$-Lipschitz.
Therefore, there exists some $C'>0$ such that
\begin{align}
&\Bigg|\mathbb{E}\Bigg[\Bigg(\sqrt{V_{0}+\sigma(V_{0})V_{0}Z_{T}}\cdot\left(\eta(S_{0})+\eta'(S_{0})\eta(S_{0})S_{0}\sqrt{V_{0}}\left(\sqrt{1-\rho^{2}}B_{T}+\rho Z_{T}\right)\right)
\nonumber
\\
&\qquad\qquad\qquad-\sqrt{V_{0}}\eta(S_{0})\Bigg)^{+}1_{V_{0}+\sigma(V_{0})V_{0}Z_{T}\geq\frac{V_{0}}{2}}\Bigg]
\nonumber
\\
&\quad
-\mathbb{E}\Bigg[\Bigg(\left(\sqrt{V_{0}}+\frac{1}{2\sqrt{V_{0}}}\sigma(V_{0})V_{0}Z_{T}\right)\cdot\left(\eta(S_{0})+\eta'(S_{0})\eta(S_{0})S_{0}\sqrt{V_{0}}\left(\sqrt{1-\rho^{2}}B_{T}+\rho Z_{T}\right)\right)
\nonumber
\\
&\qquad\qquad\qquad-\sqrt{V_{0}}\eta(S_{0})\Bigg)^{+}1_{V_{0}+\sigma(V_{0})V_{0}Z_{T}\geq\frac{V_{0}}{2}}\Bigg]\Bigg|
\nonumber
\\
&\leq
C'\mathbb{E}\left[Z_{T}^{2}\left|\eta(S_{0})+\eta'(S_{0})\eta(S_{0})S_{0}\sqrt{V_{0}}\left(\sqrt{1-\rho^{2}}B_{T}+\rho Z_{T}\right)\right|
\right]
\nonumber
\\
&\leq
C'\left(\mathbb{E}\left[Z_{T}^{4}\right]\right)^{1/2}
\left(\mathbb{E}\left[\left|\eta(S_{0})+\eta'(S_{0})\eta(S_{0})S_{0}\sqrt{V_{0}}\left(\sqrt{1-\rho^{2}}B_{T}+\rho Z_{T}\right)\right|^{2}
\right]\right)^{1/2}
\nonumber
\\
&=C'(3T^{2})^{1/2}
\left(\eta^{2}(S_{0})+\left(\eta'(S_{0})\eta(S_{0})S_{0}\sqrt{V_{0}}\right)^{2}T\right)^{1/2}
\leq O(T),
\end{align}
as $T\rightarrow 0$.

\textbf{Step 5(c).}
Next, let us show that the term in \eqref{term:approximated:by:2}
can be approximated by
\begin{align}
&\mathbb{E}\Bigg[\left(\eta(S_{0})\frac{1}{2\sqrt{V_{0}}}\sigma(V_{0})V_{0}Z_{T}
+\sqrt{V_{0}}\eta'(S_{0})\eta(S_{0})S_{0}\sqrt{V_{0}}\left(\sqrt{1-\rho^{2}}B_{T}+\rho Z_{T}\right)\right)^{+}
\nonumber
\\
&\qquad\qquad\qquad\qquad\qquad\qquad\qquad\qquad\cdot
1_{V_{0}+\sigma(V_{0})V_{0}Z_{T}\geq\frac{V_{0}}{2}}\Bigg]. \label{term:approximated:by:3}   
\end{align}

By using the fact that $x\mapsto x^{+}$ is $1$-Lipschitz and the Cauchy-Schwarz inequality, we can compute that
\begin{align}
&\Bigg|\mathbb{E}\Bigg[\Bigg(\left(\sqrt{V_{0}}+\frac{1}{2\sqrt{V_{0}}}\sigma(V_{0})V_{0}Z_{T}\right)\cdot\left(\eta(S_{0})+\eta'(S_{0})\eta(S_{0})S_{0}\sqrt{V_{0}}\left(\sqrt{1-\rho^{2}}B_{T}+\rho Z_{T}\right)\right)
\nonumber
\\
&\qquad\qquad\qquad-\sqrt{V_{0}}\eta(S_{0})\Bigg)^{+}1_{V_{0}+\sigma(V_{0})V_{0}Z_{T}\geq\frac{V_{0}}{2}}\Bigg]
\nonumber
\\
&\qquad
-\mathbb{E}\Bigg[\left(\eta(S_{0})\frac{1}{2\sqrt{V_{0}}}\sigma(V_{0})V_{0}Z_{T}
+\sqrt{V_{0}}\eta'(S_{0})\eta(S_{0})S_{0}\sqrt{V_{0}}\left(\sqrt{1-\rho^{2}}B_{T}+\rho Z_{T}\right)\right)^{+}
\nonumber
\\
&\qquad\qquad\qquad\qquad\qquad\qquad\qquad\qquad\cdot
1_{V_{0}+\sigma(V_{0})V_{0}Z_{T}\geq\frac{V_{0}}{2}}\Bigg]\Bigg|
\nonumber
\\
&\leq
\mathbb{E}
\left[\left|\frac{1}{2\sqrt{V_{0}}}\sigma(V_{0})V_{0}Z_{T}\eta'(S_{0})\eta(S_{0})S_{0}\sqrt{V_{0}}\left(\sqrt{1-\rho^{2}}B_{T}+\rho Z_{T}\right)\right|\right]
\nonumber
\\
&\leq
\frac{1}{2\sqrt{V_{0}}}\sigma(V_{0})V_{0}\eta'(S_{0})\eta(S_{0})S_{0}\sqrt{V_{0}}
\left(\mathbb{E}[Z_{T}^{2}]\right)^{1/2}
\left(\mathbb{E}\left[\left(\sqrt{1-\rho^{2}}B_{T}+\rho Z_{T}\right)^{2}\right]\right)^{1/2}
\nonumber
\\
&=\frac{1}{2\sqrt{V_{0}}}\sigma(V_{0})V_{0}\eta'(S_{0})\eta(S_{0})S_{0}\sqrt{V_{0}}T=O(T),
\end{align}
as $T\rightarrow 0$.

\textbf{Step 5(d).}
Next, let us show that the term in \eqref{term:approximated:by:3}
can be approximated by
\begin{equation}
\mathbb{E}\left[\left(\eta(S_{0})\frac{1}{2\sqrt{V_{0}}}\sigma(V_{0})V_{0}Z_{T}
+\sqrt{V_{0}}\eta'(S_{0})\eta(S_{0})S_{0}\sqrt{V_{0}}\left(\sqrt{1-\rho^{2}}B_{T}+\rho Z_{T}\right)\right)^{+}\right].
\end{equation}
We can compute that
\begin{align}
&\Bigg|\mathbb{E}\Bigg[\left(\eta(S_{0})\frac{1}{2\sqrt{V_{0}}}\sigma(V_{0})V_{0}Z_{T}
+\sqrt{V_{0}}\eta'(S_{0})\eta(S_{0})S_{0}\sqrt{V_{0}}\left(\sqrt{1-\rho^{2}}B_{T}+\rho Z_{T}\right)\right)^{+}
\nonumber
\\
&\qquad\qquad\qquad\qquad\qquad\qquad\cdot
1_{V_{0}+\sigma(V_{0})V_{0}Z_{T}\geq\frac{V_{0}}{2}}\Bigg]
\nonumber
\\
&\qquad
-\mathbb{E}\left[\left(\eta(S_{0})\frac{1}{2\sqrt{V_{0}}}\sigma(V_{0})V_{0}Z_{T}
+\sqrt{V_{0}}\eta'(S_{0})\eta(S_{0})S_{0}\sqrt{V_{0}}\left(\sqrt{1-\rho^{2}}B_{T}+\rho Z_{T}\right)\right)^{+}\right]
\Bigg|
\nonumber
\\
&=\mathbb{E}\Bigg[\left(\eta(S_{0})\frac{1}{2\sqrt{V_{0}}}\sigma(V_{0})V_{0}Z_{T}
+\sqrt{V_{0}}\eta'(S_{0})\eta(S_{0})S_{0}\sqrt{V_{0}}\left(\sqrt{1-\rho^{2}}B_{T}+\rho Z_{T}\right)\right)^{+}
\nonumber
\\
&\qquad\qquad\qquad\qquad\qquad\qquad\cdot
1_{V_{0}+\sigma(V_{0})V_{0}Z_{T}<\frac{V_{0}}{2}}\Bigg]
\nonumber
\\
&\leq
\Bigg(\mathbb{E}\left[\left(\eta(S_{0})\frac{1}{2\sqrt{V_{0}}}\sigma(V_{0})V_{0}Z_{T}
+\sqrt{V_{0}}\eta'(S_{0})\eta(S_{0})S_{0}\sqrt{V_{0}}\left(\sqrt{1-\rho^{2}}B_{T}+\rho Z_{T}\right)\right)^{2}\right]\Bigg)^{1/2}
\nonumber
\\
&\qquad\qquad\qquad\qquad\qquad\cdot
\left(\mathbb{Q}\left(V_{0}+\sigma(V_{0})V_{0}Z_{T}<\frac{V_{0}}{2}\right)\right)^{1/2}
\nonumber
\\
&\leq
\Bigg(\left(\eta(S_{0})\frac{1}{2\sqrt{V_{0}}}\sigma(V_{0})V_{0}+\sqrt{V_{0}}\eta'(S_{0})\eta(S_{0})S_{0}\sqrt{V_{0}}\rho\right)^{2}T
\nonumber
\\
&\qquad\qquad
+\left(\sqrt{V_{0}}\eta'(S_{0})\eta(S_{0})S_{0}\sqrt{V_{0}}\sqrt{1-\rho^{2}}\right)^{2}T\Bigg)^{1/2}
\left(\mathbb{Q}\left(V_{0}+\sigma(V_{0})V_{0}Z_{T}<\frac{V_{0}}{2}\right)\right)^{1/2}.
\nonumber
\end{align}

By the large deviations theory, 
\begin{equation}
\mathbb{Q}\left(V_{0}+\sigma(V_{0})V_{0}Z_{T}<\frac{V_{0}}{2}\right)=e^{-O(\frac{1}{T})},
\end{equation}
as $T\rightarrow 0$.

Therefore, we conclude that
\begin{align}
&\Bigg|\mathbb{E}\Bigg[\left(\eta(S_{0})\frac{1}{2\sqrt{V_{0}}}\sigma(V_{0})V_{0}Z_{T}
+\sqrt{V_{0}}\eta'(S_{0})\eta(S_{0})S_{0}\sqrt{V_{0}}\left(\sqrt{1-\rho^{2}}B_{T}+\rho Z_{T}\right)\right)^{+}
\nonumber
\\
\nonumber
&\qquad\qquad\qquad\qquad\qquad\qquad\cdot
1_{V_{0}+\sigma(V_{0})V_{0}Z_{T}\geq\frac{V_{0}}{2}}\Bigg]
\nonumber
\\
&\qquad
-\mathbb{E}\left[\left(\eta(S_{0})\frac{1}{2\sqrt{V_{0}}}\sigma(V_{0})V_{0}Z_{T}
+\sqrt{V_{0}}\eta'(S_{0})\eta(S_{0})S_{0}\sqrt{V_{0}}\left(\sqrt{1-\rho^{2}}B_{T}+\rho Z_{T}\right)\right)^{+}\right]
\Bigg|
\nonumber
\\
&=e^{-O(1/T)},
\end{align}
as $T\rightarrow 0$.

\textbf{Step 6.}
From the previous steps, we conclude 
that $C_{V}(K,T)$ can be approximated by:
\begin{equation}
\mathbb{E}\left[\left(\eta(S_{0})\frac{1}{2\sqrt{V_{0}}}\sigma(V_{0})V_{0}Z_{T}
+\sqrt{V_{0}}\eta'(S_{0})\eta(S_{0})S_{0}\sqrt{V_{0}}\left(\sqrt{1-\rho^{2}}B_{T}+\rho Z_{T}\right)\right)^{+}\right].    
\end{equation}
Finally, we can compute that
\begin{align}
&\mathbb{E}\left[\left(\eta(S_{0})\frac{1}{2\sqrt{V_{0}}}\sigma(V_{0})V_{0}Z_{T}
+\sqrt{V_{0}}\eta'(S_{0})\eta(S_{0})S_{0}\sqrt{V_{0}}\left(\sqrt{1-\rho^{2}}B_{T}+\rho Z_{T}\right)\right)^{+}\right]
\nonumber
\\
&=\sqrt{\left((\eta(S_{0})\frac{\sigma(V_{0})V_{0}}{2\sqrt{V_{0}}}+\sqrt{V_{0}}\eta'(S_{0})\eta(S_{0})S_{0}\sqrt{V_{0}}\rho\right)^{2}+\left(\sqrt{V_{0}}\eta'(S_{0})\eta(S_{0})S_{0}\sqrt{V_{0}}\sqrt{1-\rho^{2}}\right)^{2}}
\nonumber
\\
&\qquad\qquad\qquad\qquad\cdot\sqrt{T}\mathbb{E}\left[X^{+}\right]
\nonumber
\\
&=\sqrt{\left((\eta(S_{0})\frac{1}{2}\sigma(V_{0})\sqrt{V_{0}}+\eta'(S_{0})\eta(S_{0})S_{0}V_{0}\rho\right)^{2}+\left(\eta'(S_{0})\eta(S_{0})S_{0}V_{0}\sqrt{1-\rho^{2}}\right)^{2}}
\sqrt{T}\frac{1}{\sqrt{2\pi}},
\end{align}
where $X\sim\mathcal{N}(0,1)$.

Hence, we conclude that for ATM VIX call options, with $K=\sqrt{V_{0}}\eta(S_{0})$,
\begin{align}
&\lim_{T\rightarrow 0}\frac{1}{\sqrt{T}}C_{V}(K,T)
\nonumber
\\
&=\frac{1}{\sqrt{2\pi}}\sqrt{\left((\eta(S_{0})\frac{1}{2}\sigma(V_{0})\sqrt{V_{0}}+\eta'(S_{0})\eta(S_{0})S_{0}V_{0}\rho\right)^{2}+\left(\eta'(S_{0})\eta(S_{0})S_{0}V_{0}\sqrt{1-\rho^{2}}\right)^{2}}.
\end{align}
This completes the proof.
\end{proof}


\section{Proofs for Section~\ref{sec:H}}\label{proofs:sec:H}

\begin{proof}[Proof of Proposition~\ref{prop:Hlognorm}]
The starting point is an alternative expression for the function $H(y,z)$.
An application of the contraction principle (see for example Theorem 4.2.1. in \cite{Dembo1998}, restated in Theorem~\ref{Contraction:Thm}) from large deviations theory shows that (\ref{Hgen}) is the rate function for the large deviation principle for $\mathbb{Q}\left(\left(\frac{1}{T}\int_0^T V_t dt, V_T\right)\in\cdot\right)$ so that
\begin{equation}\label{Hlim}
H(y,z) =- \lim_{\delta\rightarrow 0}\lim_{T\to 0} T \log \mathbb{Q} \left( \frac{1}{T} \int_0^T V_t dt \in (z-\delta,z+\delta), \log V_T \in(y-\delta,y+\delta)\right)\,.
\end{equation}
For the purpose of computing $H(y,z)$, it is sufficient to 
take $\mu(v)\equiv 0$, since $H(y,z)$ is independent of the drift term in the underlying SDE for $V_{t}$ process.
For this case we have $V_t = V_0 e^{\sigma Z_t - \frac12 \sigma^2 t}$, and the probability in (\ref{Hlim}) reduces to the joint distribution of the time average of the geometric Brownian motion and its terminal value.

A closed form expression for this joint distribution was given by Yor \cite{Yor} in terms of the Hartman-Watson distribution. Define $A_t^{(\mu)} = \int_0^t e^{2(B_s+\mu s)} ds$ where $B_s$ is a standard Brownian motion. Then we have \cite{Yor}, see also Theorem~4.1 in \cite{Matsumoto}
\begin{equation}\label{ProbYor}
\mathbb{Q}\left( \frac{1}{t} A_t^{(\mu)} \in da , B_t+\mu t \in dx
\right) = e^{\mu v - \frac12 \mu^2 t} e^{-\frac{1+e^{2x}}{2 a t}}
\theta_{\frac{e^x}{at}}(t)
\frac{da dx}{a}\,,
\end{equation}
where $\theta_r(t)$ is the Hartman-Watson integral defined by
\begin{equation}
\theta_r(t) = \frac{r}{\sqrt{2\pi^3 t}} e^{\frac{\pi^2}{2t}}
\int_0^\infty e^{-\frac{\xi^2}{2t}} e^{-r\cosh \xi} \sinh \xi 
\sin \frac{\pi\xi}{t} d\xi\,.
\end{equation}

The relation (\ref{ProbYor}) can be expressed alternatively as
\begin{equation}\label{ProbYor2}
\mathbb{Q}\left( \frac{1}{t} A_t^{(\mu)} \in da , e^{B_t+\mu t} \in dv
\right) = v^\mu e^{- \frac12 \mu^2 t} e^{-\frac{1+v^2}{2 a t}}
\theta_{\frac{v}{at}}(t)
\frac{da dv}{a v} := \Psi(a,v;t) da dv\,.
\end{equation}

Next we express the probability in (\ref{Hlim}) in terms of the function $\Psi(a,v;t)$ defined in (\ref{ProbYor2}). Using the scaling property of the standard Brownian motion we have
\begin{equation}\label{scaling}
\frac{1}{T} \int_0^T V_t dt = V_0 \frac{1}{\tau} A_\tau^{(-1)}\,,\quad
V_T = V_0 e^{2(B_\tau-\tau)}\,,\qquad \tau := \frac14 \sigma^2 T\,.
\end{equation}
We get that the probability in (\ref{Hlim}) is
\begin{equation}\label{PPsi}
\mathbb{Q} \left( \frac{1}{T} \int_0^T V_t dt \in dz, \log V_T \in dy\right) = \Psi\left( \frac{z}{V_0}, \sqrt{\frac{y}{V_0}} ;\tau\right) \frac{1}{2\sqrt{y/V_0}} dz dy\,.
\end{equation}

We use next the leading $t\to 0$ asymptotics of the Hartman-Watson integral $\theta_{\rho/t}(t)$ from Proposition 1 in \cite{Pirjol2020}
\begin{equation}
\theta_{\rho/t}(t) = \frac{1}{2\pi t} G(\rho) e^{-\frac{1}{t}(F(\rho) - \frac{\pi^2}{2})} ( 1 + O(t))\,,
\end{equation}
with $F(\rho),G(\rho)$ being known functions, and $F(\rho)$ given above in (\ref{Fdef}).
Substituting into (\ref{ProbYor2}) gives
\begin{equation}
\Psi(a,v;t) = \frac{1}{2\pi t} \psi(a,v;t) e^{-\frac{1}{t} (\frac{1+v^2}{2a} + F(v/a) - \frac{\pi^2}{2})} (1 + O(t))\,,
\end{equation}
with $\psi(a,v;t) = v^{\mu-1} a^{-1} e^{-\frac12\mu^2 t} G(v/a)$.
Thus we have
\begin{equation}\label{sub:into}
- \lim_{t\to 0} t \log \Psi(a,v;t) = \left(\frac{1+v^2}{2a} + F(v/a) - \frac{\pi^2}{2}\right) :=  \frac18 I(a,v)\,.
\end{equation}
Substituting \eqref{sub:into} into (\ref{PPsi}) we get
\begin{align}
H(y,z) &=  -\lim_{\delta\rightarrow 0}\lim_{T\to 0} T \log \mathbb{Q} \left( \frac{1}{T} \int_0^T V_t dt \in(z-\delta,z+\delta), \log V_T \in (y-\delta,y+\delta)\right)\nonumber
\\
&= - \frac{4}{\sigma^2} \lim_{\tau\to 0} \tau \log\Psi\left(z/V_0, \sqrt{y/V_0};\tau\right)\nonumber 
\\
&= \frac{1}{2\sigma^2} I\left(z/V_0,\sqrt{y/V_0}\right)\,,
\end{align}
which completes the proof of (\ref{Iuv}).
\end{proof}

\begin{proof}[Proof of Proposition~\ref{prop:Heston:2}]
This proposition follows directly from the G\"{a}rtner-Ellis theorem from large deviations theory; 
see e.g. \cite{Dembo1998}. 
In order to apply the G\"{a}rtner-Ellis theorem, we will show that the limit \eqref{cumulant:limit} exists
and compute it out explicitly as follows (so that it can be seen easily that the essential smoothness condition 
for the G\"{a}rtner-Ellis theorem is satisfied).

The expectation $M(T;\theta,\phi) := \mathbb{E}\left[e^{\theta \int_0^T V_t dt + \phi V_T}\right]$ can be computed exactly for the case of constant $\mu(v) \equiv r$.
Since the rate function $I_H(x,y)$ is independent of $\mu(v)$ (provided that $\mu(v)$ satisfies the technical assumptions required for the existence of the large deviations property), we will use a constant drift function $\mu(v)\equiv r$ to compute the cumulant function $\Lambda(\theta,\phi)$. 
This has the form
\begin{equation}
M(T;\theta,\phi) = e^{V_0 A(T;\theta, \phi)}\,,
\end{equation}
where the function $A(T;\theta,\phi)$ can be found in closed-form,
which can be extracted from the proof of Theorem~14 in \cite{PZAsianCEV}.

Using this result we get the following expression for the cumulant function 
\begin{align}\label{LambdaSol}
\Lambda_H(\theta,\phi) &= \lim_{T\to 0} T \log A\left(T;\frac{\theta}{T^2},\frac{\phi}{T}\right) 
\nonumber
\\
&=
\begin{cases}
\frac{\sqrt{2\theta}}{\sigma} \frac{\sqrt{2\theta} \tan(\frac{\sigma}{2}\sqrt{2\theta}) +
\sigma\phi}
{\sqrt{2\theta} - \sigma\phi \tan(\frac{\sigma}{2}\sqrt{2\theta}) }\,,
& 0 \leq \theta < \theta_c(\phi), \\
\frac{\sqrt{-2\theta}}{\sigma} \frac{\sigma\phi -\sqrt{-2\theta} \tanh(\frac{\sigma}{2}\sqrt{-2\theta}) }
{\sqrt{-2\theta} - \sigma\phi \tanh(\frac{\sigma}{2}\sqrt{-2\theta}) }\,,
& \theta < 0\,, 0 \leq \phi < \phi_c(\theta),
\end{cases}
\end{align}
where $\theta_c(\phi)$ and $\phi_c(\theta)$ are the boundary curves given by the solutions of the equation:
\begin{equation}
\frac{\sigma}{2} \sqrt{2\theta} T + \tan^{-1} \left( \frac{\sigma\phi}{\sqrt{2\theta}} \right)  =
\frac{\pi}{2}\,,
\end{equation}
or equivalently
\begin{equation}
\frac{\sigma \phi}{\sqrt{2\theta}} \tan\left( \frac{\sigma}{2} \sqrt{2\theta} T \right)= 1\,.
\end{equation}
This completes the proof.
\end{proof}

\begin{proof}[Proof of Proposition~\ref{prop:HHeston}]
The rate function is given by the double Legendre transform
\begin{equation}\label{Isol}
I_H(x,y) = \sup_{\theta,\phi} \left[\theta x + \phi y - \Lambda_H(\theta, \phi)\right]\,,
\end{equation}
where the cumulant function $\Lambda_H(\theta, \phi)$ is given in explicit form in (\ref{LambdaSol}).

Denote the minimizers of this problem as $\theta_*,\phi_*$. They are given by the solutions of the equations
\begin{equation}\label{eqstar}
\partial_\theta \Lambda(\theta,\phi) = x\,,\quad
\partial_\theta \Lambda(\theta,\phi) = y \,.
\end{equation}

One can expand the minimizers $\theta_*,\phi_*$ in powers of $\epsilon_{x,y}$ as
\begin{eqnarray}
&& \theta_* = a_{1,1} \epsilon_x + a_{1,2} \epsilon_y + O(\epsilon^2)\,, \\
&& \phi_* = b_{1,1} \epsilon_x + b_{1,2} \epsilon_y + O(\epsilon^2)\,.
\end{eqnarray}
The expansion of the cumulant function in powers of $(\theta,\phi)$ has the form
\begin{equation}
\Lambda(\theta,\phi) = \theta + \phi + \sigma^2 \left( \frac16 \theta^2 + \frac12 \theta
\phi + \frac12 \phi^2 \right) + O\left(\theta^3, \theta^2 \phi , \theta \phi^2, \phi^3\right)\,.
\end{equation}

The coefficients $a_{j,k}, b_{j,k}$ are determined by substituting their expansion into the equations (\ref{eqstar}), and expanding in powers of $\epsilon_{x,y}$. The first few terms are
\begin{eqnarray}
&& \sigma^2 \theta_* = 6(2\epsilon_x - \epsilon_y) - \frac{24}{5} \epsilon_x^2  + \frac{24}{5}
\epsilon_x \epsilon_y - \frac{11}{5} \epsilon_y^2 + O(\epsilon^3)\,, \\
&& \sigma^2 \phi_* = -2(3\epsilon_x - 2\epsilon_y) - \frac{3}{5} \epsilon_x^2  + \frac{8}{5}
\epsilon_x \epsilon_y - \frac{2}{5} \epsilon_y^2 + O(\epsilon^3) \,.
\end{eqnarray}

Substituting into the expression for the rate function (\ref{Isol}) gives an expansion 
in $(\epsilon_x,\epsilon_y)$. In order to get the expansion of the rate function up to and including terms of order $\epsilon^n$ with $n\geq 3$, one has to compute the expansion of $(\theta_*,\phi_*)$ in $\epsilon$ to the same order. In particular, obtaining the expansion to $O(\epsilon^4)$ in (\ref{IHexp}) requires the expansions of 
$(\theta_*,\phi_*)$ including the $O(\epsilon^4)$ terms.
\end{proof}

\section{Proofs for Section~\ref{sec:applications}}\label{proofs:sec:applications}

\begin{proof}[Proof of Proposition~\ref{prop:Elognorm}]

We would like to compute the expansion of the European rate function in powers of log-strike
\begin{equation}\label{JEexp}
J_E(K) = j_1^E k^2 + j_2^E k^3 + \cdots\,.
\end{equation}
The implied volatility option has the corresponding expansion
\begin{equation}\label{sigSPXexp}
\sigma_{\mathrm{BS}}(K) = \frac{|k|}{\sqrt{2J_E(K)}} = \frac{1}{\sqrt{2j_1^E}} - \frac{j_2^E}{2\sqrt2 (j_1^{E})^{3/2}} k + O(k^2)\,,
\end{equation}
where the first term is the ATM implied volatility,  the second term is the ATM skew, and so on,
and $k = \log\left(K/S_0\right)$ is the log-strike.

The problem was reduced to that of computing the rate function $J_E(K)$ for OTM European options.
This rate function is given by Theorem~\ref{Thm:E}. 
Using the explicit result for $H(y,z)$ for $\sigma(v) = \sigma$ in equation~\eqref{HLN}, the rate function has the form
\begin{equation}
J_E(K) = \inf_{y,z} \left\{
\frac{1}{2\rho_\perp^2 z} \left( \int_{S_0}^{K}
\frac{dx}{x\eta(x)} - \rho \int_{V_0}^{e^y} \frac{dx}{\sqrt{x} \sigma} \right)^2
+ \frac{1}{2\sigma^2} I\left( \frac{z}{V_0} , \frac{e^{y/2}}{\sqrt{V_0}} \right) \right\}\,.
\end{equation}

Let us introduce new notation
\begin{equation}\label{uvdef}
u := \frac{z}{V_0}\,,\qquad 
v := \frac{e^{y/2}}{\sqrt{V_0}}\,.
\end{equation}

The rate function becomes
\begin{equation}\label{JE2dim}
J_E(K) = \inf_{u,v} \left\{
\frac{1}{2\rho_\perp^2 V_0 u} \left( I_S(K/S_0) - \rho \frac{2\sqrt{V_0}}{\sigma}
(v-1) \right)^2
+ \frac{1}{2\sigma^2}I(u,v) \right\}\,,
\end{equation}
where we defined
\begin{equation}\label{ISdef}
I_S(z) := \int_{S_0}^{S_0 z}
\frac{dx}{x\eta(x)}\,.
\end{equation}


For an ATM European option we have $k=0$, and the infimum in (\ref{JE2dim}) is realized at $u_*=1, v_*=1$. This gives $J_V\left(K=\eta_0\sqrt{V_0}\right) = 0$.

The idea of the proof is to expand the minimizers in the extremal problem 
(\ref{JE2dim}) in powers of log-strike
\begin{eqnarray}\label{uvexp}
&& \log u_* = a_1^E k + a_2^E k^2 + \cdots\,, \\
&& \log v_* = b_1^E k + b_2^E k^2 + \cdots \,,\nonumber
\end{eqnarray}
and solve explicitly the coefficients $a_i^E, b_i^E$ at each order in $k$. We give the details only for the leading coefficient 
$j_1^E$ in the rate function $J_E(K)$, the higher order coefficients are obtained in a similar way.

We will parameterize the local volatility function $\eta(x)$ as an expansion in $\log(x/S_0)$
\begin{equation}\label{etaexp}
\eta(S) = \eta(S_{0}) + \eta_1 \log\left(\frac{S}{S_0}\right) + \eta_2 \log^2\left(\frac{S}{S_0}\right) + \cdots \,.
\end{equation}

We start by expanding the integral $I_S(z)$ defined in (\ref{ISdef}), as
\begin{equation}\label{ISexp}
I_S(z) = \frac{1}{\eta_0} \log z - \frac{\eta_1}{2\eta_0^2} \log^2 z + \frac13 
\left(\frac{\eta_1^2 }{\eta_0^3} - \frac{\eta_2}{\eta_0^2} \right) \log^3 z + \cdots\,,
\end{equation}
where $\eta_{0}:=\eta(S_{0})$.
This is obtained by expanding the integrand of $I_S(z)$ using (\ref{etaexp}) and integrating term-by-term.

We expand the argument of the extremal problem (\ref{JE2dim}) in powers of log-strike $k = \log\frac{K}{S_0}$. The leading order term in this expansion is
\begin{equation}
\Lambda_E(u,v) := \frac{1}{2\rho_\perp^2 V_0 u} \left(\frac{1}{\eta_0 }
k + O(k^2)  
 - \rho \frac{2\sqrt{V_0}}{\sigma}
(v-1) \right)^2
+ \frac{1}{2\sigma^2} I(u,v)\,. 
\end{equation}

We find the solutions of $\partial_u \Lambda_E(u,v)=0, \partial_v\Lambda_E(u,v)=0$ by
substituting here the expansions (\ref{uvexp}) and keeping only terms of the same power in $k$. At leading order, we get
\begin{equation}
\partial_u \Lambda_E(u,v)= \frac{12}{\sigma^2}\left(a_1^E - b_1^E\right) k + O(k^2) = 0\,.
\end{equation}
Requiring that the coefficient of the $O(k)$ term vanishes gives $a_1^E=b_1^E$.

Analogously,
\begin{equation}
\partial_v \Lambda_E(u,v)= \left( -\frac{4}{\sigma^2}\left(3a_1^E - 4 b_1^E\right) - \frac{2\rho}{\rho_\perp^2 \sigma\sqrt{V_0}}
\left( \frac{1}{\eta_0} - \frac{2b_1^E \rho \sqrt{V_0}}{\sigma} \right) \right) k + O(k^2) = 0\,,
\end{equation}
which gives a second equation for $a_1^E,b_1^E$.  

The solution of these equations for $\left(a_1^E, b_1^E\right)$ is
\begin{equation}\label{a1Esol}
a_1^E = b_1^E = \frac{\rho\sigma}{2\eta_0 \sqrt{V_0}}\,.
\end{equation}

The expansion of (\ref{JE2dim}) in powers of $k$ reads 
\begin{equation}
\Lambda_E(u_*(k),v_*(k)) = 
\left\{ 
\frac{1}{2\rho_\perp^2 V_0}
\left[ \frac{1}{\eta_0}  - \frac{2\rho\sqrt{V_0}}{\sigma} b_1^E \right]^2 + 
\frac{6 (a_1^{E})^2 - 12 a_1^E b_1^E  + 8 (b_1^{E})^2}{\sigma^2}
\right\} k^2
+ O(k^3)\,.
\end{equation}
Substituting here the solution (\ref{a1Esol}) for $(a_1^E,b_1^E)$ gives the leading order coefficient in the expansion of the rate function $J_E(K)$ 
\begin{equation}
j_1^E = \frac{1}{2 \eta_0^2 V_0}\,.
\end{equation}
This yields the stated result (\ref{SPXvolATM}) for the ATM European implied volatility.

This approach can be extended to higher orders in log-strike $k$ to compute the terms $j_i^E$ with $i=2,3,\ldots$. We get
\begin{equation}
j_2^E = - \frac{\rho \sigma + 2\eta_1 \sqrt{V_0}}{4\eta_0^3 V_0^{3/2}}\,,
\end{equation}
and
\begin{equation}
j_3^E = \frac{(15\rho^2-4) \sigma^2 + 36\eta_1 \rho \sigma \sqrt{V_0} + 4(11\eta_1^2 - 8\eta_0\eta_2) V_0}{96\eta_0^4 V_0^2} \,.
\end{equation}

Substituting into (\ref{sigSPXexp})
gives the higher order derivatives of the implied volatility at the ATM point (skew and convexity) quoted above in (\ref{SPXskewATM}) and (\ref{SPXcvxATM}).
This completes the proof.
\end{proof}

\begin{proof}[Proof of Proposition~\ref{prop:Vlognorm}]
Using the same approach as that used above for the European options, we 
compute the expansion of the VIX rate function around the ATM point
\begin{equation}\label{JVexp}
J_V(K) = j_1^V x^2 + j_2^V x^3 + O(x^4)\,,
\end{equation}
where $x = \log\left(K/(\eta(S_{0})\sqrt{V_0})\right)$.
The implied volatility of the VIX option has the expansion
\begin{equation}\label{sigVIXexp:2}
\sigma_{\mathrm{VIX}}(K) = \frac{|x|}{\sqrt{2J_V(K)}} = \frac{1}{\sqrt{2j_1^V}} - \frac{j_2^V}{(2j_1^V)^{3/2}} x + O(x^2)\,.
\end{equation}
The first term is the ATM VIX implied volatility and the second term is the ATM VIX skew. 

The rate function for VIX options $J_V(K)$ is given by Theorem~\ref{Thm:VIX}. 
Using the explicit result for $H(y,z)$ for $\sigma(v)\equiv\sigma$ in equation~\eqref{HLN}, the rate function has the form
\begin{equation}
J_V(K) = \inf_{y,z} \left\{
\frac{1}{2\rho_\perp^2 z} \left( \int_{S_0}^{(\eta^2)^{-1}(e^{-y} K^2)}
\frac{dx}{x\eta(x)} - \rho \int_{V_0}^{e^y} \frac{dx}{\sqrt{x} \sigma} \right)^2
+ \frac{1}{2\sigma^2} I\left( \frac{z}{V_0} , \frac{e^{y/2}}{\sqrt{V_0}} \right) \right\}\,.
\end{equation}

Changing variables to $(u,v)$ defined as in (\ref{uvdef}), the rate function becomes
\begin{equation}\label{JV2dim}
J_V(K) = \inf_{u,v} \left\{
\frac{1}{2\rho_\perp^2 V_0 u} \left( I_S(\zeta(K,v)) - \rho \frac{2\sqrt{V_0}}{\sigma}
(v-1) \right)^2
+ \frac{1}{2\sigma^2} I(u , v) \right\}\,,
\end{equation}
where $I_S(z)$ was defined above in (\ref{ISdef}),
and $\zeta(K,v)$ is the solution of the equation
\begin{equation}\label{zetaeq}
\frac{K^2}{V_0} \frac{1}{v^2} =\eta^2(S_0 \zeta(K,v))\,.
\end{equation}

The proof parallels closely the proof of Proposition~\ref{prop:Elognorm} so we give only the proof outline. 

For an ATM VIX option $x=0$, the infimum in (\ref{JV2dim}) is realized at $u_*=1, v_*=1$. This gives $J_V(K=\eta_0\sqrt{V_0}) = 0$.
We expand the minimizers in the extremal problem 
(\ref{JV2dim}) in powers of log-strike
\begin{align}
& \log u_* = a_1^V x + a_2^V x^2 + \cdots\,, \nonumber\\
& \log v_* = b_1^V x + b_2^V x^2 + \cdots \,,\label{uvexpV}
\end{align}
and solve for the coefficients $a_i^V, b_i^V$ at each order in $x$.

The solution of the equation (\ref{zetaeq}) gives an expansion for $\zeta(K,v)$ of the form
\begin{equation}
\log\zeta(K,v) = \frac{\eta_0}{\eta_1} \left( \frac{e^x}{v} - 1\right) - \frac{\eta_2\eta_0^2}{\eta_1^3}\left( \frac{e^x}{v} - 1\right)^2 + \cdots\,.
\end{equation}
This can be substituted into the expansion of $I_S(z)$ in (\ref{ISexp}) to get an expansion of $I_S(\zeta(K,v))$ in 
powers of $( \frac{e^x}{v} - 1)$.

Keeping only the leading order term in this expansion, the argument of (\ref{JV2dim}) becomes
\begin{equation}
\Lambda_V(u,v) := \frac{1}{2\rho_\perp^2 V_0 u} \left(\frac{\eta_0}{\eta_1 }
\left( \frac{e^x}{v} - 1\right) + \cdots 
 - \rho \frac{2\sqrt{V_0}}{\sigma}
(v-1) \right)^2
+ \frac{1}{2\sigma^2} I( u , v) \,.
\end{equation}

We find the solutions of $\partial_u \Lambda_V(u,v)=0, \partial_v\Lambda_V(u,v)=0$ by
substituting here the expansions (\ref{uvexpV}) and keeping only terms of the same power in $x$. 
We have
\begin{equation}\label{a1solV}
a_1^V = b_1^V = \sigma \frac{\sigma + 2\rho \eta_1 \sqrt{V_0}}{(\sigma + 2\rho \eta_1 \sqrt{V_0})^2 + 2\rho_\perp^2 \eta_1^2 V_0}\,.
\end{equation}

The expansion of (\ref{JV2dim}) in powers of $x$ reads 
\begin{equation}
\Lambda_V(u_*(x),v_*(x)) = \left\{ 
\frac{1}{2\rho_\perp^2 V_0}
\left[ \frac{1}{\eta_1} \left(1-a_1^V\right) - \frac{2\rho\sqrt{V_0}}{\sigma} a_1^V \right]^2 + 
\frac{4 (a_1^{V})^2}{\sigma^2}
\right\} x^2 + O(x^3)\,.
\end{equation}
Substituting here the solution (\ref{a1solV}) for $a_1^V$ gives 
the leading order coefficient in the VIX rate function 
\begin{equation}
j_1^V = \frac{2}{(\sigma + 2\rho \eta_1 \sqrt{V_0})^2 + 4  \rho_\perp^2 \eta_1^2 V_0}\,.
\end{equation}
Substituting into (\ref{sigVIXexp:2})  gives the stated result for the ATM VIX implied volatility. 

The ATM VIX skew requires the coefficient $j_2^V$ which is found by expanding to order $O(x^2)$ and solving for $a_2^V,b_2^V$. The result is
\begin{equation}
j_2^V =-\frac{4V_0 (2\eta_1 \sqrt{V_0} + \rho \sigma)}{(\sigma + 2\rho \eta_1 \sqrt{V_0})^2 + 4  \rho_\perp^2 \eta_1^2 V_0}
\left( \eta_1(8\eta_0\eta_2 V_0 + \sigma^2) + 2 (2\eta_0 \eta_2 +\eta_1^2) \rho \sigma \sqrt{V_0}\right)\,.
\end{equation}
Substituting into the second term of (\ref{sigVIXexp:2}) gives the stated result for the ATM VIX skew. The convexity $\kappa_{VIX}$ requires also the coefficient $j_3^V$ which we do not give in complete form due to the lengthy expression.
This completes the proof.
\end{proof}


\begin{proof}[Proof of Proposition~\ref{prop:S:moments}]
For any $p>1$, by the Cauchy-Schwarz inequality, we can compute that
\begin{align}
\mathbb{E}[S_{t}^{p}]
&=\mathbb{E}\left[S_{0}^{p}e^{\int_{0}^{t}(p(r-q)-\frac{p}{2}\eta^{2}(S_{u})V_{u})du+\int_{0}^{t}p\eta(S_{u})\sqrt{V_{u}}dW_{u}}\right]
\nonumber
\\
&\leq
S_{0}^{p}
\left(\mathbb{E}\left[e^{\int_{0}^{t}(2p(r-q)+(-p+2p^{2})\eta^{2}(S_{u})V_{u})du}\right]\right)^{1/2}
\nonumber
\\
&\qquad\qquad\qquad\cdot
\left(\mathbb{E}\left[e^{\int_{0}^{t}\frac{-(2p)^{2}}{2}\eta^{2}(S_{u})V_{u}du+\int_{0}^{t}2p\eta(S_{u})\sqrt{V_{u}}dW_{u}}\right]\right)^{1/2}.
\end{align}
Notice that $e^{\int_{0}^{t}\frac{-(2p)^{2}}{2}\eta^{2}(S_{u})V_{u}du+\int_{0}^{t}2p\eta(S_{u})\sqrt{V_{u}}dW_{u}}$ is a non-negative
local martingale. Since any local martingale that is bounded from below is a supermartingale, we conclude that
for any $p>1$:
\begin{align}
\mathbb{E}[S_{t}^{p}]
&\leq
S_{0}^{p}
\left(\mathbb{E}\left[e^{\int_{0}^{t}(2p(r-q)+(-p+2p^{2})\eta^{2}(S_{u})V_{u})du}\right]\right)^{1/2}
\nonumber
\\
&\leq
S_{0}^{p}e^{p|r-q|t}
\left(\mathbb{E}\left[e^{\int_{0}^{t}(-p+2p^{2})M_{\eta}^{2}V_{u}du}\right]\right)^{1/2} < \infty \,,
\end{align}
which implies that
\begin{equation}
\max_{0\leq t\leq T}\mathbb{E}[S_{t}^{p}]
\leq
S_{0}^{p}e^{p|r-q|T}
\left(\mathbb{E}\left[e^{\int_{0}^{T}(-p+2p^{2})M_{\eta}^{2}V_{u}du}\right]\right)^{1/2}
=O(1),
\end{equation}
as $T\rightarrow 0$ where we applied Assumption~\ref{assump:V:integral}.
This concludes the proof.
\end{proof}


\begin{proof}[Proof of Proposition~\ref{prop:EHeston}]
The proof is similar to that of Proposition~\ref{prop:Elognorm} apart from the use of the function $H(y,z)$ for the Heston-type model.
\end{proof}

\begin{proof}[Proof of Proposition~\ref{prop:VHeston}]
The proof is similar to that of Proposition~\ref{prop:Vlognorm} apart from the use of the function $H(y,z)$ for the Heston-type model.
\end{proof}

\section{At-the-Money Convexity for the VIX Options}
\label{app:VIXcvx}

We give in this Appendix the full result for the ATM VIX implied volatility convexity in the local-stochastic volatility model with SABR-type volatility quoted in Proposition~\ref{prop:Vlognorm}. 
This is defined as the coefficient of the
quadratic term in expansion of the VIX implied volatility in log-strike $x$
\begin{equation}
\sigma_{\mathrm{VIX}}(x) = \sigma_{\mathrm{VIX},\mathrm{ATM}} + s_{\mathrm{VIX}} x + \kappa_{\mathrm{VIX}} x^2 + O(x^3)\,,
\end{equation}
and has the explicit result
\begin{equation}\label{VIXcvx}
\kappa_{\mathrm{VIX}} = \frac{1}{6} \frac{\sqrt{V_0}}{(\sigma^2 + 4\eta_1 \rho \sigma \sqrt{V_0} + 4\eta_1^2 V_0)^{7/2}} K_{\mathrm{VIX}}\,,
\end{equation}
where $K_{\mathrm{VIX}}$ is given by
\begin{equation}
K_{\mathrm{VIX}} := \sum_{i=0}^7 k_i \sigma^i\,,
\end{equation}
where the coefficients $k_j$ are
\begin{align}
&k_0 := 256 
\eta_0 \eta_1^4 V_0^{7/2}
(\eta_1^2 \eta_2 - 3\eta_0 \eta_2^2 + 3\eta_0 \eta_1 \eta_3)\,, 
\\
&k_1 := 128 \eta_0 \eta_1^3 \rho V_0^{3}
(15 \eta_0 \eta_1 \eta_3 -  12\eta_0 \eta_2^2 + 5 \eta_1^2 \eta_2)\,,  
\\
&k_2 := 16 \eta_1^2 V_0^{5/2} \Big(12 \eta_0^2 \eta_1 \eta_3 (9\rho^2 + 1)    
+ 24 \eta_0^2\eta_2^2 (1-4\rho^2) \nonumber
\\
&\qquad\qquad\qquad\qquad\qquad\qquad
+ 4\eta_0 \eta_1^2 \eta_2 (15\rho^2-2) + \eta_1^4 (2 - 3\rho^2) \Big)\,,
\\
&k_3 := 16 \eta_1 \rho V_0^{2} \left(6 \eta_0^2 \eta_1 \eta_3 (7 \rho^2 + 3)
+ 6 \eta_0^2 \eta_2^2 (4 -8\rho^2) 
+ 4 \eta_0 \eta_1^2 \eta_2 (8\rho^2 + 3) - \eta_1^4 \rho^2 \right)\,, 
\\
&k_4 := 4 V_0^{3/2} \Big(12 \eta_0^2 \eta_1 \eta_3 \rho^2 (2\rho^2 + 3) 
 + 12 \eta_0^2 \eta_2^2 \rho^2 (2 - 3\rho^2) 
 \nonumber
 \\
 &\qquad\qquad\qquad\qquad\qquad\qquad
 + 4\eta_0 \eta_1^2 \eta_2 (5\rho^4 + 12 \rho^2  + 6) 
  - \eta_1^4 ( \rho^4 - 6\rho^2 + 3) \Big)\,,
\\
&k_5 := 4\rho V_0 \left( 6 \eta_0^2 \eta_3 \rho^2 + 2 \eta_0 \eta_1 \eta_2 (4\rho^2 + 9) + \eta_1^3 (\rho^2+3) \right) \,,
\\
&k_6 := \sqrt{V_0} \left( 12 \eta_0 \eta_2 \rho^2+ \eta_1^2 (3\rho^2+4)  \right)\,, 
\end{align}
and $k_7 := \eta_1 \rho$.


\bibliographystyle{plain}
\bibliography{ShortMaturityVIX}


\end{document}